\documentclass[lettersize, journal]{IEEEtran}
\IEEEoverridecommandlockouts
\usepackage{cite}
\usepackage{amsmath,amssymb,amsfonts}
\usepackage{verbatim}
\usepackage{nohyperref}
\usepackage{graphicx}
\usepackage{textcomp}
\usepackage{cases}
\usepackage{amsmath,amssymb,amsfonts}
\usepackage{fancyhdr,graphicx,amsmath,amssymb}
\usepackage[ruled,vlined]{algorithm2e}
\usepackage[monochrome]{xcolor}
\usepackage{subfigure}
\usepackage[utf8]{inputenc}
\usepackage{soul}
\usepackage[noend]{algpseudocode}
\usepackage{multirow}
\usepackage{fontenc}
\usepackage{amsthm}
\usepackage{enumitem}
\usepackage[caption=false,font=normalsize,labelfont=sf,textfont=sf]{subfig}
\usepackage{multicol}

\newtheorem{theorem}{Theorem}
\newtheorem{lemma}{Lemma}
\newtheorem{corollary}{Corollary}
\newtheorem{definition}{Definition}

\def\BibTeX{{\rm B\kern-.05em{\sc i\kern-.025em b}\kern-.08em
    T\kern-.1667em\lower.7ex\hbox{E}\kern-.125emX}}

\begin{document}

\title{AoI, Timely-Throughput, and Beyond: A Theory of Second-Order Wireless Network Optimization
}

\author{
\IEEEauthorblockN{Daojing Guo, Khaled Nakhleh, I-Hong Hou, Sastry Kompella, Clement Kam}

\thanks{
Daojing Guo is with Amazon Web Services (AWS), Seattle, WA (email: daojing.guo@gmail.com). The work of Daojing Guo was done when he was with the ECE Department, Texas A\&M University, College Station, TX.\\
Khaled Nakhleh and I-Hong Hou are with the ECE Department, Texas A\&M University, College Station, TX (email: \{khaled.jamal, ihou\}@tamu.edu). \\
Sastry Kompella is with Nexcepta Inc, Gaithersburg, MD (email: sk@ieee.org). \\
Clement Kam is with the U.S. Naval Research Laboratory (NRL) (email: ckk@ieee.org). \\
This material is based upon work supported in part by NSF under Award Numbers ECCS-2127721 and CCF-2332800 and in part by the U.S. Army Research Laboratory and the U.S. Army Research Office under Grant Number W911NF-22-1-0151. \\
Part of this work has been presented at IEEE INFOCOM 2022 \cite{guo2022}. 
}
}


\maketitle
\begin{abstract}
This paper introduces a new theoretical framework for optimizing second-order behaviors of wireless networks. Unlike existing techniques for network utility maximization, which only consider first-order statistics, this framework models every random process by its mean and temporal variance. The inclusion of temporal variance makes this framework well-suited for modeling {\color{blue}Markovian} fading wireless channels and emerging network performance metrics such as age-of-information (AoI) and timely-throughput. Using this framework, we sharply characterize the second-order capacity region of wireless access networks. We also propose a simple scheduling policy and prove that it can achieve every interior point in the second-order capacity region. 
To demonstrate the utility of this framework, we apply it to an unsolved network optimization problem where some clients wish to minimize AoI while others wish to maximize timely-throughput. We show that this framework accurately characterizes AoI and timely-throughput. Moreover, it leads to a tractable scheduling policy that outperforms other existing work.
\end{abstract}


\IEEEpeerreviewmaketitle

\section{introduction}

\IEEEPARstart{T}{here} are two seemingly contradictory trends happening in the field of wireless network optimization. On one hand, the study of network utility maximization (NUM) has witnessed tremendous success in the past two decades. Techniques based on dual decomposition, Lyapunov function, etc., have been shown to produce tractable and optimal solutions in complex networks for a wide range of objectives, including maximizing spectrum efficiency, minimizing power consumption, enforcing fairness among clients, and the combination of these objectives. Recent studies have also established iterative algorithms that not only converge to the optimum, but also have provably fast convergence rate \cite{huang2014power, huang2009delay, liu2016heavy, chen2017learn, liu2016achieving}. On the other hand, there have been growing interests in new performance metrics for emerging network applications, such as quality-of-experience (QoE) \cite{joseph2015,zhangxu2019, zhangyinjie2020,bampis2021, Bhattacharyya2022} for video streaming and age-of-information (AoI) \cite{kadota2019,lou2020,yates2021,chen2022,pappas2022age} for real-time state estimation, and timely-throughput \cite{chen2018, singh2019, yang2019, sun2021, singh2021} for real-time communications. Surprisingly, except for a few special cases, the problem of optimizing these new performance metrics remains largely open. This raises the question: Why do existing NUM techniques fail to solve the optimization problem for these new performance metrics?

The fundamental reason is that current NUM techniques are only applicable to first-order performance metrics, while emerging new performance metrics involve higher-order behaviors. Existing NUM problems typically define the utility of a flow $n$ as $U_n(x_n)$, where $x_n$ is an asymptotic first-order performance metric, such as throughput (\emph{long-term average} number of packet deliveries per unit time), power consumption (\emph{long-term average} amount of energy consumption per unit time), and channel utilization (\emph{long-term average} number of transmissions per unit time). However, emerging performance metrics like QoE and AoI require the characterization of short-term network behaviors, and hence cannot be fully captured by asymptotic first-order statistics. 

To bridge the gap between NUM techniques and emerging performance metrics, we present a new framework of second-order wireless optimization. This framework consists of the second-order models, that is, the means and the temporal variances, of all random processes, including the channel qualities and packet deliveries of wireless clients. The incorporation of temporal variances enables this framework to better characterize {\color{blue}Markovian} fading wireless channels, such as Gilbert-Elliott channels, and emerging performance metrics. 

Using this framework, we sharply characterize the second-order capacity region of wireless networks, which entails the set of means and temporal variances of packet deliveries that are feasible under the constraints of the second-order models of channel qualities. As a result, the problem of optimizing emerging performance metrics is reduced to one that finds the optimal means and temporal variances of packet deliveries within the second-order capacity region. We also propose a simple scheduling policy and show that it can achieve every interior point of the second-order capacity region.

To demonstrate the utility of our framework, we apply it to an unsolved network optimization problem. This problem considers a wireless system where some clients wish to minimize their AoIs while other clients wish to maximize their timely-throughputs. Moreover, the wireless channel of each client follows the Gilbert-Elliott (GE) model. We theoretically derive the closed-form expressions of the means and temporal variances for Gilbert-Elliott channels. We also show that both AoI and timely-throughput can be well-approximated by the mean and the temporal variance of the packet delivery process. As a result, we are able to directly apply our theoretical framework to obtain the optimal scheduling policy {\color{blue}under the approximations}. Simulation results show that our policy significantly outperforms existing policies. Importantly, despite being a generic policy for second-order network optimization, our policy is able to achieve smaller AoI than those specifically designed to minimize AoI, and higher timely-throughput than those specifically designed to maximize timely-throughput.

The rest of the paper is organized as follows: Section \ref{sec:model} formally defines the second-order models of channel qualities and packet deliveries and the problem of second-order optimization. Section \ref{sec: aoi} uses the second-order models to formulate a yet unsolved network optimization problem: the problem of minimizing AoI of real-time sensing clients and maximizing timely-throughput of live video streaming clients over Gilbert-Elliott channels. Section \ref{sec: capacity region} derives an outer bound of the second-order capacity region. Section \ref{sec: scheduling policy} proposes a simple scheduling policy and shows that it achieves every interior point of the second-order capacity region. Section \ref{sec:simulation} presents our simulation results. Section \ref{sec:related} surveys some related studies. Finally, Section \ref{sec:conclusion} concludes the paper.
\section{System Model for Second-Order Wireless Network Optimization}
\label{sec:model}

We begin by describing a generic network optimization problem. Consider a wireless system where one access point (AP) serves $N$ clients, numbered as $n = 1,2,\dots, N$. Time is slotted and denoted by $t = 1,2,3,\dots.$ We consider the ON-OFF channel model where the AP can schedule a client for transmission if and only if the channel for the client is ON. Let $X_n(t)$ be the indicator function that the channel for client $n$ is ON at time $t$. We assume that the sequence $\{X_n(1), X_n(2),\dots\}$ is governed by a stochastic positive-recurrent Markov process with finite states. In each time slot, if there is at least one client having an ON channel, then the AP selects a client with an ON channel and transmits a packet to it. Let $Z_n(t)$ be the indicator function that client $n$ receives a packet at time $t$. The empirical performance of client $n$ is modeled as a function of the entire sequence $\{Z_n(1), Z_n(2),\dots\}$. 
The network optimization problem is to find a scheduling policy that maximizes the total performance of the network.

Solving this generic network optimization problem may be difficult because it requires solving {\color{blue} a high-dimensional Markov decision process, as the state of the system consists of the states of all channels and the delivery processes of all clients.}  As a result, except for a few special cases, there remain no tractable optimal solutions for many emerging network performance metrics like AoI. To circumvent this challenge, we propose capturing each random process by its second-order model, namely, its mean and temporal variance. 

We first define the second-order model for channels. With a slight abuse of notations, let $X_S(t):=\max\{X_n(t)|n\in S\}$ be the indicator function that at least one client in $S$ has an ON channel at time $t$. Since all channels are governed by stochastic positive-recurrent Markov processes, the strong law of large numbers for Markov chains states that $\frac{\sum_{t=1}^TX_S(t)}{T}$ converges to a constant almost surely as $T\rightarrow\infty$. Hence, we can define the mean of $X_S$ as
\begin{equation}
    {\color{blue}m_S:=\lim_{T\rightarrow\infty}\frac{\sum_{t=1}^TX_S(t)}{T}.}
\end{equation}
The Markov central limit theorem further states that $\frac{\sum_{t=1}^TX_S(t)-Tm_S}{\sqrt{T}}$ converges in distribution to a Gaussian random variable as $T\rightarrow\infty$. Hence, we define the temporal variance of $X_S$ as
\begin{equation}
    v_S^2:=\lim_{T\rightarrow\infty}E[(\frac{\sum_{t=1}^TX_S(t)-Tm_S}{\sqrt{T}})^2].
\end{equation}
The second-order channel model is then expressed as the collection of the means and temporal variances of all $X_S$, namely, $\{(m_S, v_S^2)|S\subseteq\{1,2,\dots,N\}\}$.

The second-order model for packet deliveries is defined similarly. {\color{blue}We assume that the AP employs a scheduling policy under which the packet delivery process of each client $n$, $\{Z_n(1), Z_n(2),\dots\}$, follows a positive-recurrent Markov process. Then, we can define the mean and the temporal variance of $Z_n$ as}
\begin{equation}
    \mu_n:=\lim_{T\rightarrow\infty}\frac{\sum_{t=1}^TZ_n(t)}{T},
\end{equation}
and
\begin{equation}
 {\color{blue}\sigma_n^2:=\lim_{T\rightarrow\infty}E[(\frac{\sum_{t=1}^TZ_n(t)-T\mu_n}{\sqrt{T}})^2].}
\end{equation}
The second-order delivery model is $\{(\mu_n,\sigma_n^2)|1\leq n \leq N\}$. The performance of client $n$ is modeled as a function of $(\mu_n,\sigma_n^2)$, which we denote by $F_n(\mu_n,\sigma_n^2)$.

Since clients want to have large means and small variances for their delivery processes, we define the second-order capacity region of a network as follows:
\begin{definition}
[Second-order capacity region] Given a second-order channel model $\{(m_S, v_S^2)|S\subseteq\{1,2,\dots,N\}\}$, the second-order capacity region is the set of all $\{(\mu_n,\sigma_n^2)|1\leq n \leq N\}$ such that there exists a scheduling policy under which $\lim_{T\rightarrow\infty}\frac{\sum_{t=1}^TZ_n(t)}{T}=\mu_n$, {\color{blue} almost surely}, and {\color{blue}$\lim_{T\rightarrow\infty}E[(\frac{\sum_{t=1}^TZ_n(t)-T\mu_n}{\sqrt{T}})^2]\leq\sigma_n^2,\forall n$}. $\Box$
\end{definition}

The second-order network optimization problem entails finding the scheduling policy that maximizes $\sum_{n=1}^NF_n(\mu_n,\sigma_n^2)$.

\section{A Motivating Example and Its Second-Order Model} \label{sec: aoi}

To show the effectiveness of our methodology, we apply it to a network optimization problem that is yet to be solved. The network consists of two types of clients: \emph{real-time sensing} clients indexed as $n=1,2,\dots, I$, and \emph{live video streaming} clients indexed as $n=I+1, I+2, \dots, I+J$ with the total number of clients $N = I+J$. Real-time sensing clients generate surveillance updates to the AP to make control decisions, and each client aims to minimize the age-of-information to ensure data freshness for good decision-making. Live video streaming clients require both low video latency and smooth playback, and are measured by their relative deadline and timely-throughput. Previous research has focused on optimizing real-time sensing \cite{moltafet2020, buyukates20} and live streaming clients \cite{zuo2022, singh2021} separately. However, due to the substantial differences between the two client types, achieving network-wide performance optimization when both are present remains as a significant challenge.

To make the optimization problem even harder, we further consider that each wireless link follows the Markovian Gilbert-Elliott channel model, whose channel quality is not i.i.d. over time. In this section, we will derive the second-order models of Gilbert-Elliott channels, real-time sensing clients, and live video streaming clients. This enables us to address all of them in a unified theoretical framework.

\subsection{The Second-Order Model of Gilbert-Elliott Channels}

\begin{figure}
\centering
\includegraphics[width=2in]{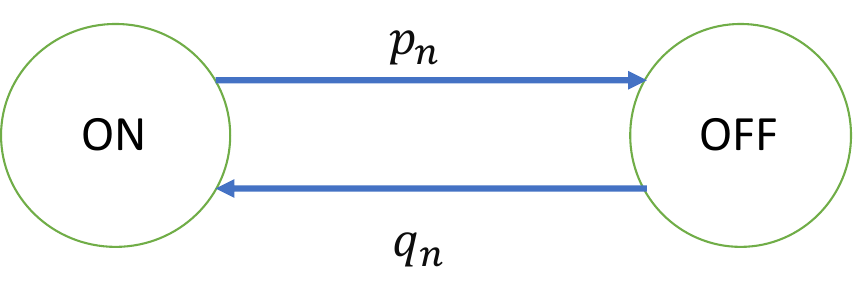}
\caption{The Gilbert-Elliott model.}\label{fig:GE} 
\vspace{-5pt}
\end{figure}

In Gilbert-Elliott channels \cite{Gilbert1960GEmodel,Elliot1963GEmodel}, the channel for each client $n$ is modeled as a two-state Markov process, as shown in Fig. ~\ref{fig:GE}. The channel is ON if it is in the good (G) state, and is OFF if it is in the bad (B) state. The transition probabilities from G to B and from B to G are $p_n$ and $q_n$, respectively. The channels are independent from each other.
We now show the second-order model of Gilbert-Elliott channels.
\begin{lemma} \label{lemma:second-order-GE}
Under the Gilbert-Elliott channels, for all $S$,
\begin{align}
    m_S=& 1-\prod_{n\in S}\frac{p_n}{p_n+q_n},\label{eq:GE mean}\\ 
    v_S^2=&2\sum_{k=1}^\infty\Big(\prod_{n\in S}G_n(k+1)-\prod_{n\in S}\frac{p_n}{p_n+q_n}\Big)\prod_{n\in S}\frac{p_n}{p_n+q_n}\nonumber\\
    &+\prod_{n\in S}\frac{p_n}{p_n+q_n}-(\prod_{n\in S}\frac{p_n}{p_n+q_n})^2,\label{eq:GE variance}
\end{align}
where $G_n(k)=\frac{p_n}{p_n+q_n}+\frac{q_n}{p_n+q_n}(1-p_n-q_n)^{k-1}$.
\end{lemma}
\begin{proof}
Let $Y_n(t):= 1- X_n(t)$ be the indicator function that client $n$ has an OFF channel at time $t$. Let $Y_S(t):=1-X_S(t)$ be the indicator function that all clients in the subset $S$ have OFF channels at time $t$. Hence, we have $Y_S(t)=\prod_{n\in S}Y_n(t)$. Suppose the Markov process of each channel is in the steady-state at time $t$, then we have $Prob(Y_n(t)=1)=\frac{p_n}{p_n+q_n}$. Hence, $E[Y_S(t)]=\prod_{n\in S}\frac{p_n}{p_n+q_n}$ and $E[X_S(t)] = 1-E[Y_S(t)]=1-\prod_{n\in S}\frac{p_n}{p_n+q_n}$. This establishes (\ref{eq:GE mean}).

Next, we establish (\ref{eq:GE variance}). We have $(\sum_{t=1}^TX_S(t)-Tm_S)^2=(\sum_{t=1}^TY_S(t)-T(1-m_S))^2$. By the Markov central limit theorem \cite{geyer2011introduction}, we can calculate $v_S^2$ by assuming that the Markov process of each channel is in the steady-state at time $1$ and using the following formula:
\begin{equation}
    v_S^2=Var(Y_S(1))+2\sum_{k=1}^\infty Cov(Y_S(1), Y_S(1+k)).
\end{equation}
Since $Y_S(1)$ is a Bernoulli random variable with mean $\prod_{n\in S}\frac{p_n}{p_n+q_n}$, we have 
\begin{equation}Var(Y_S(1))=\prod_{n\in S}\frac{p_n}{p_n+q_n}-(\prod_{n\in S}\frac{p_n}{p_n+q_n})^2. \label{eq:Y_variance}
\end{equation}

Let $G_n(k)=Prob(Y_n(k)=1|Y_n(1)=1)$. Then,
\begin{align}
    &E[Y_S(1)Y_S(1+k)]=Prob(Y_S(1+k)=1|Y_S(1)=1) \notag \\
    &\times Prob(Y_S(1)=1)\notag\\
    =&Prob(Y_n(1+k)=1,\forall n\in S|Y_n(1)=1, \forall n\in S)\notag \\
    &\times \prod_{n\in S}\frac{p_n}{p_n+q_n}\notag\\
    =&\prod_{n\in S}G_n(k+1)\prod_{n\in S}\frac{p_n}{p_n+q_n},
\end{align}
and
\begin{align}
    &Cov(Y_S(1), Y_S(1+k))\notag\\
    =&E[Y_S(1)Y_S(1+k)]-E[Y_S(1)]E[Y_S(1+k)]\notag\\
    =&\Big(\prod_{n\in S}G_n(k+1)-\prod_{n\in S}\frac{p_n}{p_n+q_n}\Big)\prod_{n\in S}\frac{p_n}{p_n+q_n}.\label{eq:Y_cov}
\end{align}
Combining (\ref{eq:Y_variance}) and (\ref{eq:Y_cov}) establishes (\ref{eq:GE variance}).

It remains to find the closed-form expression of $G_n(k)$. We have
\begin{align}
    &G_n(k)=Prob(Y_n(k)=1|Y_n(1)=1)\notag\\
    =&G_n(k-1)(1-q_n)+(1-G_n(k-1))p_n\notag\\
    =&p_n+(1-p_n-q_n)G_n(k-1),
\end{align}
if $k>1$, and $G_n(k)=1$, if $k=1$. Solving this recursive equation yields $G_n(k)=\frac{p_n}{p_n+q_n}+\frac{q_n}{p_n+q_n}(1-p_n-q_n)^{k-1}$. This completes the proof.
\end{proof}

When $p_n+q_n=1$, the Gilbert-Elliott channel reduces to the i.i.d. channel model where $X_n(t)=1$ with probability $q_n$, independent from any prior events. By replacing $p_n=1-q_n$, we obtain the second-order model of i.i.d. channels as below:
\begin{corollary}
Under the i.i.d. channels with $Prob(X_n(t)=1)=q_n$,
\begin{align}
    m_S=& 1-\prod_{n\in S}(1-q_n), 
    v_S^2=& \prod_{n\in S}(1-q_n)-\prod_{n\in S}(1-q_n)^2,
\end{align}
for all $S$. $\Box$
\end{corollary}

\subsection{The Second-Order Model of Real-Time Sensing}

Real-time sensing clients are sensors that generate information updates to be transmitted to the AP. The performance of a sensor $n$ is measured by its long-term average age-of-information (AoI), denoted by $\overline{AoI}_n$.

In a nutshell, the AoI corresponding to a sensor at a given time is defined as the age of the newest information update that it has ever delivered to the centralized server. We consider that each sensor $n$ generates new updates by a Bernoulli random process. In each time slot $t$, sensor $n$ generates a new update with probability $\lambda_n$, independent from any prior events. In most scenarios, newer updates are more relevant to the server's decision making. Hence, each sensor only keeps the most recent update in its memory, and it transmits the most recent update whenever it is scheduled for transmission. The controller knows $\lambda_n$ but not the exact times at which sensors generate new updates. Hence, we assume that the scheduling decision is independent from {\color{blue}update generation processes}. Let $g_n(k)$ be the time when sensor $n$ generates the $k-$th update and let $AoI_n(t)$ be the AoI of sensor $n$ at time $t$. Then we have $AoI_n(t) = t - \max\{g_n(k)|g_n(k)<t, k = 1, 2, ,\dots\}$, if sensor $n$ delivers a packet at time $t$, and $AoI_n(t) = AoI_n(t-1)+1$, otherwise.
 

Let $A_n(m):=\min\{\tau|\sum_{t=1}^\tau Z_n(t)=m\}$ be the time of the $m$-th delivery for client $n$, and let $B_n(m) := A_n(m+1)-A_n(m)$ be the time between the $m$-th and the $(m+1)$-th deliveries. Since scheduling decisions are independent from {\color{blue}update generation processes}, we have the following:
\begin{lemma}
If $\{B_n(0), B_n(1), \dots\}$ is independent from the update generation processes of sensor $n$, then the long-term average AoI of sensor $n$ is
\begin{equation}
    \overline{AoI}_n=\frac{E[B_n^2]}{2E[B_n]} + \frac{1}{\lambda_n}-\frac{1}{2},
\end{equation}
where $E[B_n^2]:=\lim_{b\rightarrow\infty}\sum_{m=1}^b B_n(m)^2/b$ and $E[B_n]:=\lim_{b\rightarrow\infty}\sum_{m=1}^b B_n(m)/b$.
\end{lemma}
\begin{proof}
This lemma can be established by combining techniques in the proof of Proposition 2 in \cite{kadota2019minimizing} and the fact that $B_n(m)$ is independent from {\color{blue}update generation processes}.
\end{proof}

We aim to express $\overline{AoI}_n$ as a function of the second-order delivery model of client $n$, $(\mu_n, \sigma_n^2)$. Since there can be multiple sequences of $\{Z_n(1), Z_n(2),\dots\}$ with the same $(\mu_n, \sigma_n^2)$, we will derive $\overline{AoI}_n$ with respect to a \emph{second-order reference delivery process} as defined below.

Let $BM_{\mu_n,\sigma_n^2}(t)$ be a Brownian motion random process with mean $\mu_n$ and variance $\sigma_n^2$ \cite{chen2001fundamentals} and initial value of $BM_{\mu_n,\sigma_n^2}(0)=0$. An important property of the Brownian motion random process is that for any $t_1<t_2$, $BM_{\mu_n,\sigma_n^2}(t_2)-BM_{\mu_n,\sigma_n^2}(t_1)$ is a Gaussian random variable with mean $(t_2-t_1)\mu_n$ and variance $(t_2-t_1)\sigma_n^2$. {\color{blue}Given such a Brownian motion random process, we define the reference delivery process such that, if, at the end of a time slot, the Brownian motion random process has increased by one since the last packet delivery, then there is a packet delivery in the reference delivery process.} The formal definition of the reference delivery process is shown below:


\begin{definition}
Given $(\mu_n, \sigma_n^2)$, the second-order reference delivery process, denoted by $\{Z'_n(1), Z'_n(2),\dots\}$ is defined to be
\begin{equation}
    Z'_n(t)=\left\{\begin{array}{ll}1&\mbox{if $BM_{\mu_n,\sigma_n^2}(t)-BM_{\mu_n,\sigma_n^2}(t^-)\geq 1$,}\\
    0&\mbox{else,}
    \end{array}
    \right.\label{eq:Z'def}
\end{equation}
where $t^-:=\max\{\tau|\tau<t, Z'_n(\tau)=1\}$. {\color{blue}To address the boundary condition, we define $Z'_n(0)=1$.} $\Box$
\end{definition}

We now derive $\overline{AoI}_n$ with respect to the sequence $\{Z'_n(1), Z'_n(2),\dots\}$. Consider the time between the $m$-th and the $(m+1)$-th deliveries, which is denoted by $B_n(m)$, under the sequence $\{Z'_n(1), Z'_n(2),\dots\}$. From (\ref{eq:Z'def}), $B_n(m)$ can be approximated by the amount of time needed for the Brownian motion random process to increase by 1, which is equivalent to the first-hitting time for a fixed level 1 and we denote it by $H_n$. It has been shown that the the first-hitting time for a fixed level 1 follows the inverse Gaussian distribution $IG(\frac{1}{\mu_n}, \frac{1}{\sigma_n^2})$ \cite{Schrodinger1915, folks1978inverse}. Hence, we have $E[H_n]=1/\mu_n$ and $E[H_n^2]=\sigma_n^2/\mu_n^3+1/\mu_n^2$. We now have
\begin{align}
&\overline{AoI}_n=\frac{E[B_n^2]}{2E[B_n]}+\frac{1}{\lambda_n}-\frac{1}{2}\nonumber\\
    \approx&\frac{E[H_n^2]}{2E[H_n]}+\frac{1}{\lambda_n}-\frac{1}{2}=\frac{1}{2}(\frac{\sigma_n^2}{\mu_n^2}+\frac{1}{\mu_n}) +\frac{1}{\lambda_n}-\frac{1}{2}. \label{eq:AoI approx}
\end{align}

\subsection{The Second-Order Model of Live Video Streaming}\label{subsec:live_video_streaming}

Some clients watch live video streams that have stringent deadline delivery requirements from the centralized server to the clients. Live video streams generate video frames at a constant rate and these frames should be played by the end users after a fixed delay. Specifically, we say that the stream of client $n$ generates one packet every $w_n$ slots. Each packet has a strict relative deadline of $\ell_n \cdot w_n$ slots, that is, a packet generated at time $t$ needs to be delivered by time $t + \ell_n \cdot w_n$.
Packets that cannot be delivered by their deadlines are considered to be expired and are dropped from the system. When the AP schedules client $n$ for transmission, it transmits the packet with the earliest deadline among all available packets for client $n$ so as to minimize packet drops. If the AP schedules client $n$ for transmission but there are no available packets, i.e. all packets for client $n$ are either delivered or dropped, then the AP transmits a dummy packet that contains no information.

In the context of video streaming, each packet drop causes an outage in the video playback. Hence, when $\ell_n$ is given and fixed, we measure the performance of a live video streaming client $n$ by its outage rate, defined as the average number of packet drops per time slot. We use $\overline{Out}_n$ to denote the outage rate of clients $n= I+1, I+2,\hdots, I+J$. Hsieh and Hou \cite{hsieh20} has shown that, when $\mu_n = 1 / w_n$,
\begin{equation} \label{eq:outage_approx}
    \overline{Out}_n \approx \frac{\sigma_n^2}{2 \ell_n}.
\end{equation}

We also note that the timely-throughput, i.e. the throughput of timely packet deliveries, of client $n$ is $1/w_n - \overline{Out}_n$.

\subsection{Model Validation} \label{subsec:model_validation}

We now verify whether the second-order model provides a good approximation of AoI and timely-throughput over Gilbert-Elliott channels. We consider a system with only one client. The centralized server schedules the client for transmission whenever the client has an ON channel. Hence, we have $\mu_1=m_{\{1\}}$ and $\sigma_1^2=v_{\{1\}}^2$. {\color{blue}We will further validate the model in multi-user systems in Section~\ref{sec:simulation}.}

We first evaluate the case when the sole client is a real-time sensing client.
Given, $p_1$, $q_1$, and $\lambda_1$, we can combine \eqref{eq:GE mean}, \eqref{eq:GE variance}, and \eqref{eq:AoI approx} to obtain a theoretical approximation of the AoI. We note that (\ref{eq:GE variance}) involves a summation of infinite terms $\sum_{k=1}^\infty (G_1(k)-\frac{p_1}{p_1+q_1})$. {\color{blue}Since $G_1(k)$ converges to $\frac{p_1}{p_1+q_1}$ exponentially fast, we replace this term with $\sum_{k=1}^{K} (G_1(k)-\frac{p_1}{p_1+q_1})$ when calculating $v_{\{1\}}^2$ and evaluate the cases when $K=100$ and when $K=1000$.}
For each $(p_1, q_1, \lambda_1)$, we obtain the empirical AoI by simulating the system for 1000 runs, where each run contains 50,000 time slots.

The results averaged over a 1000 runs are shown in Fig. \ref{fig:single_client_validation} for four channel and packet generation probability settings. It can be observed that the theoretical AoI is very close to the empirical AoI in almost all cases. {\color{blue}The only point when there is a considerable difference betwen the theoretical AoI and the empirical AoI is when $p_1=q_1=0.01$ and $K = 100$. This is because, when $p_1=q_1=0.01$, we have $G_1(101)-\frac{p_1}{p_1+q_1}=0.066$. Thus, setting $K=100$ results in a non-negligible error in calculating $v_{\{1\}}^2$. On the other hand, we have $G_1(1001)-\frac{p_1}{p_1+q_1}<10^{-9}$. Thus, setting $K=1000$ makes the theoretical AoI and the empirical AoI almost identical. We recommend choosing a $K$ with $G_1(K)-\frac{p_1}{p_1+q_1}<0.001$ when calculating the temporal variance.}

\begin{figure}[ht]

\begin{center} 
\subfigure[$q=0.01$. $\lambda = 1$.]{\includegraphics[width=1.7in, height=1.37in]{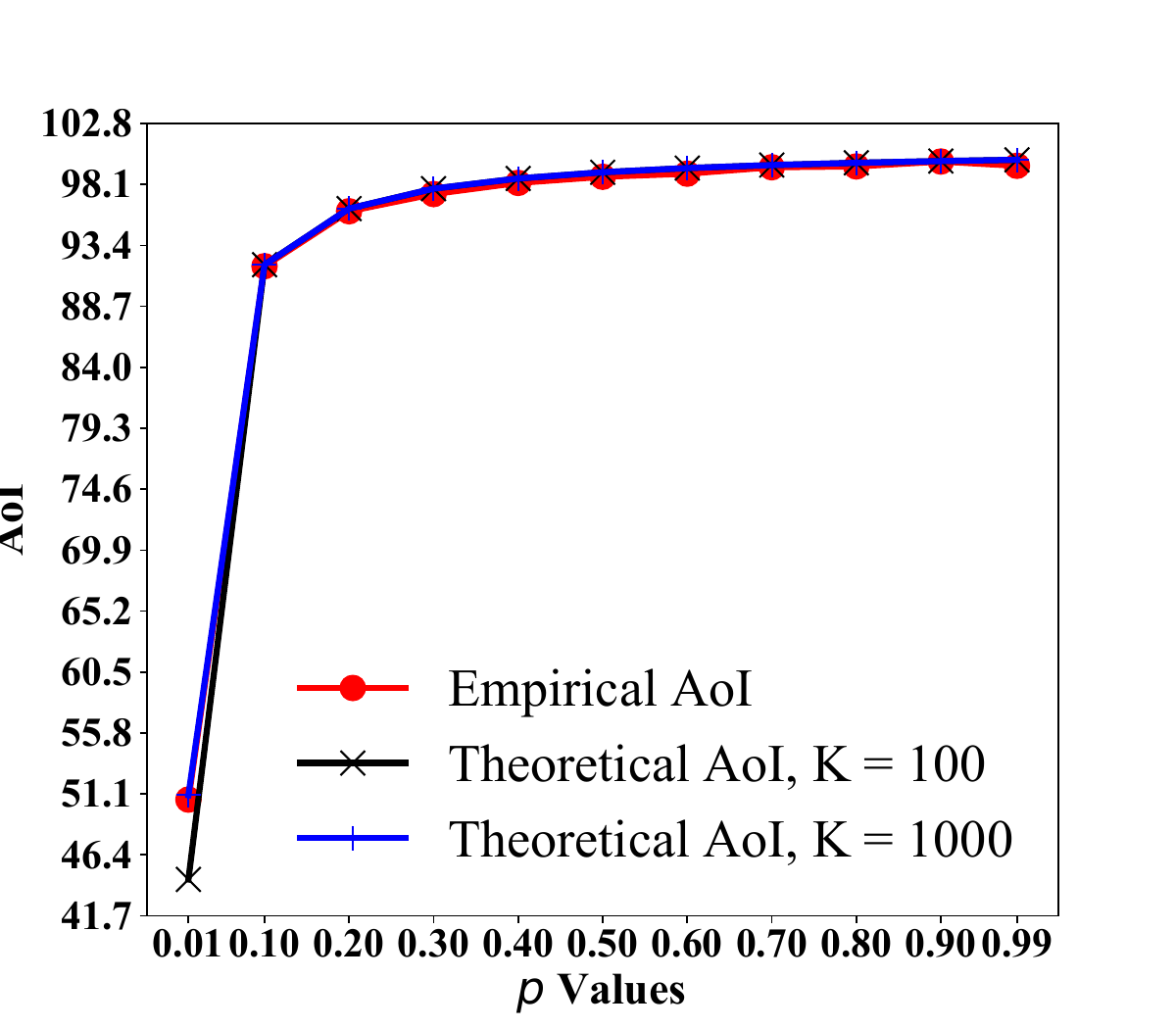}}
\subfigure[$q=0.01$. $\lambda = 0.1$.]{\includegraphics[width=1.7in, height=1.37in]{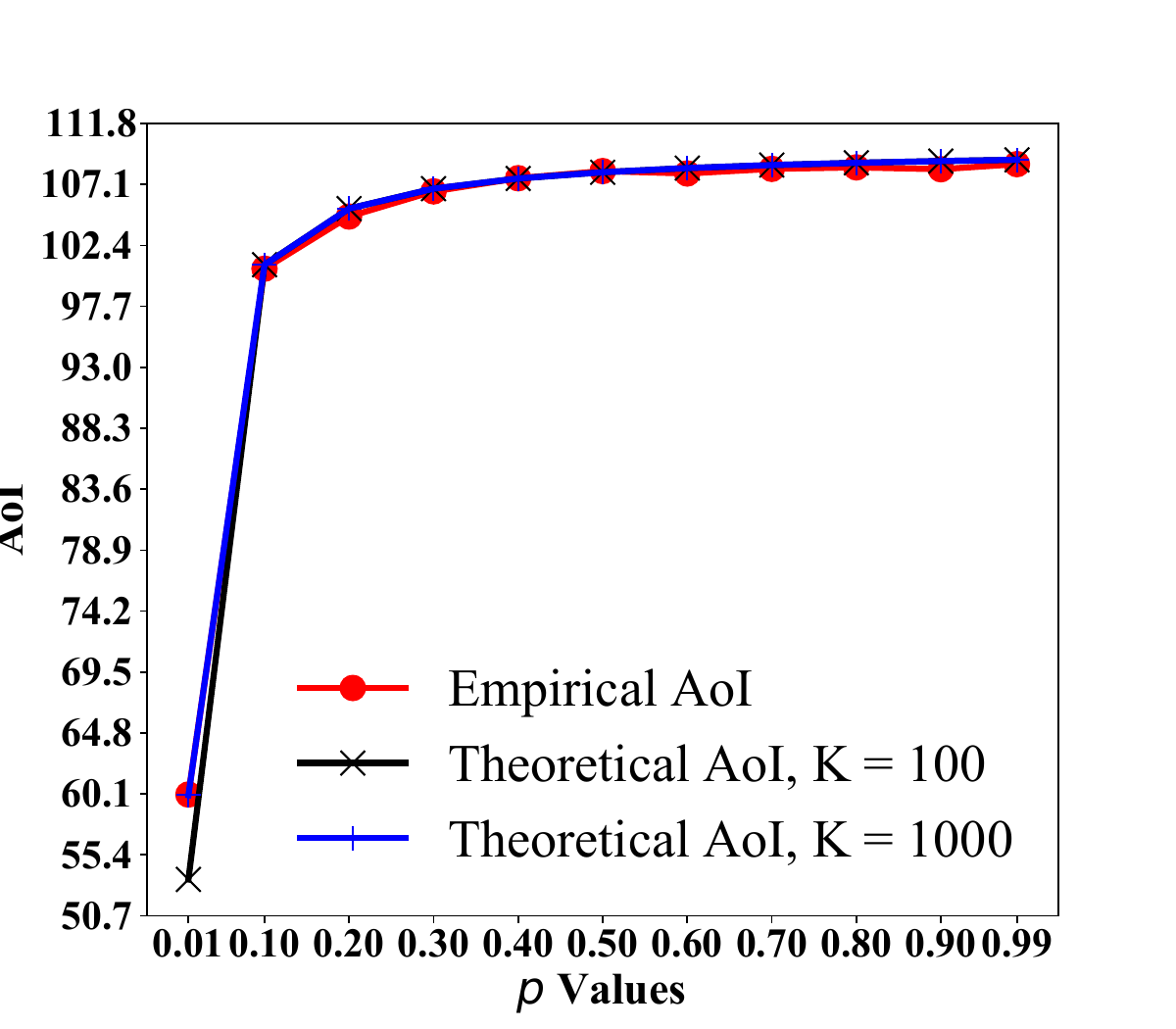}}
\subfigure[$q=0.2$. $\lambda = 1$.]{\includegraphics[width=1.7in, height=1.37in]{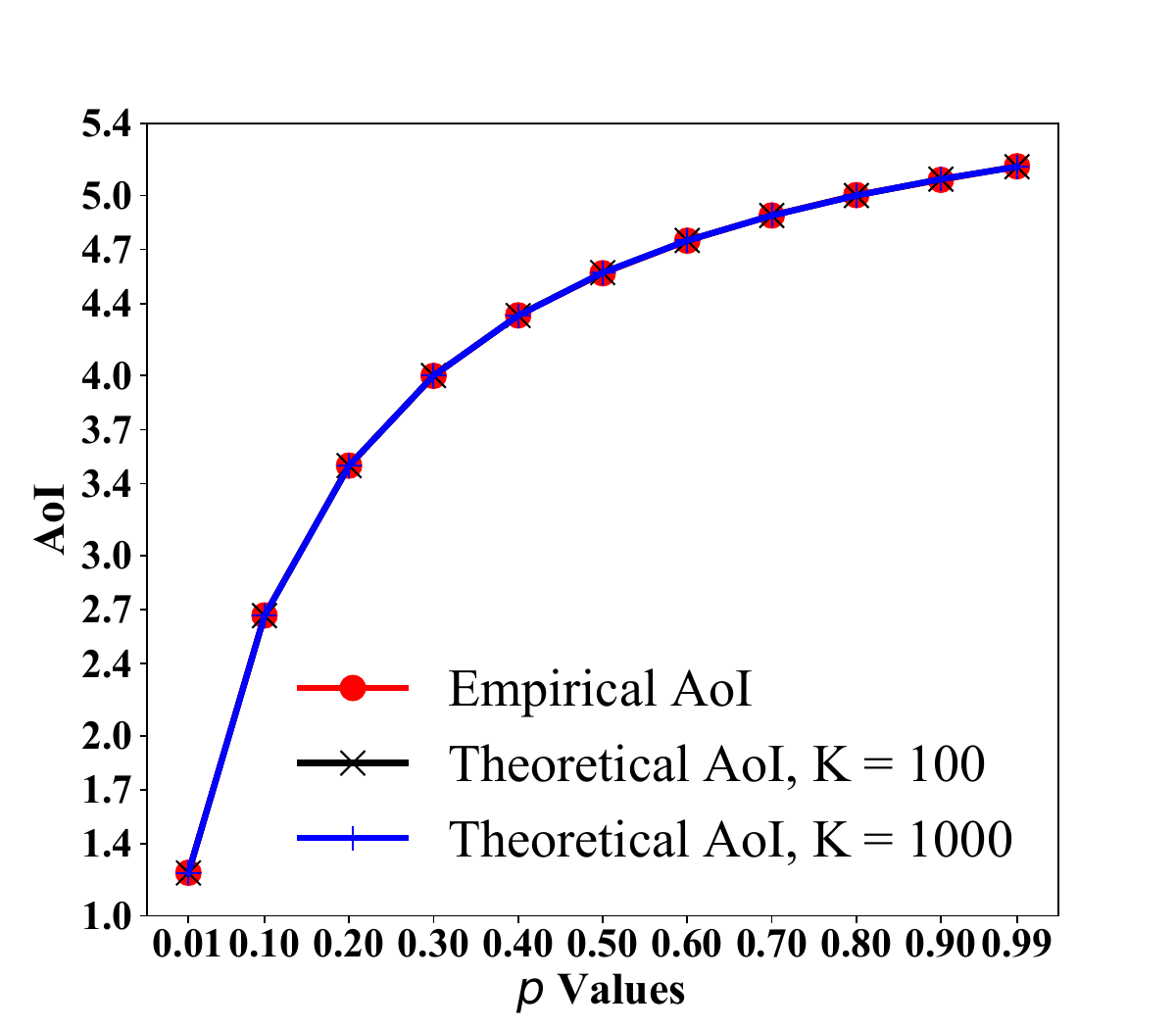}}
\subfigure[$q=0.2$. $\lambda = 0.1$.]{\includegraphics[width=1.7in, height=1.37in]{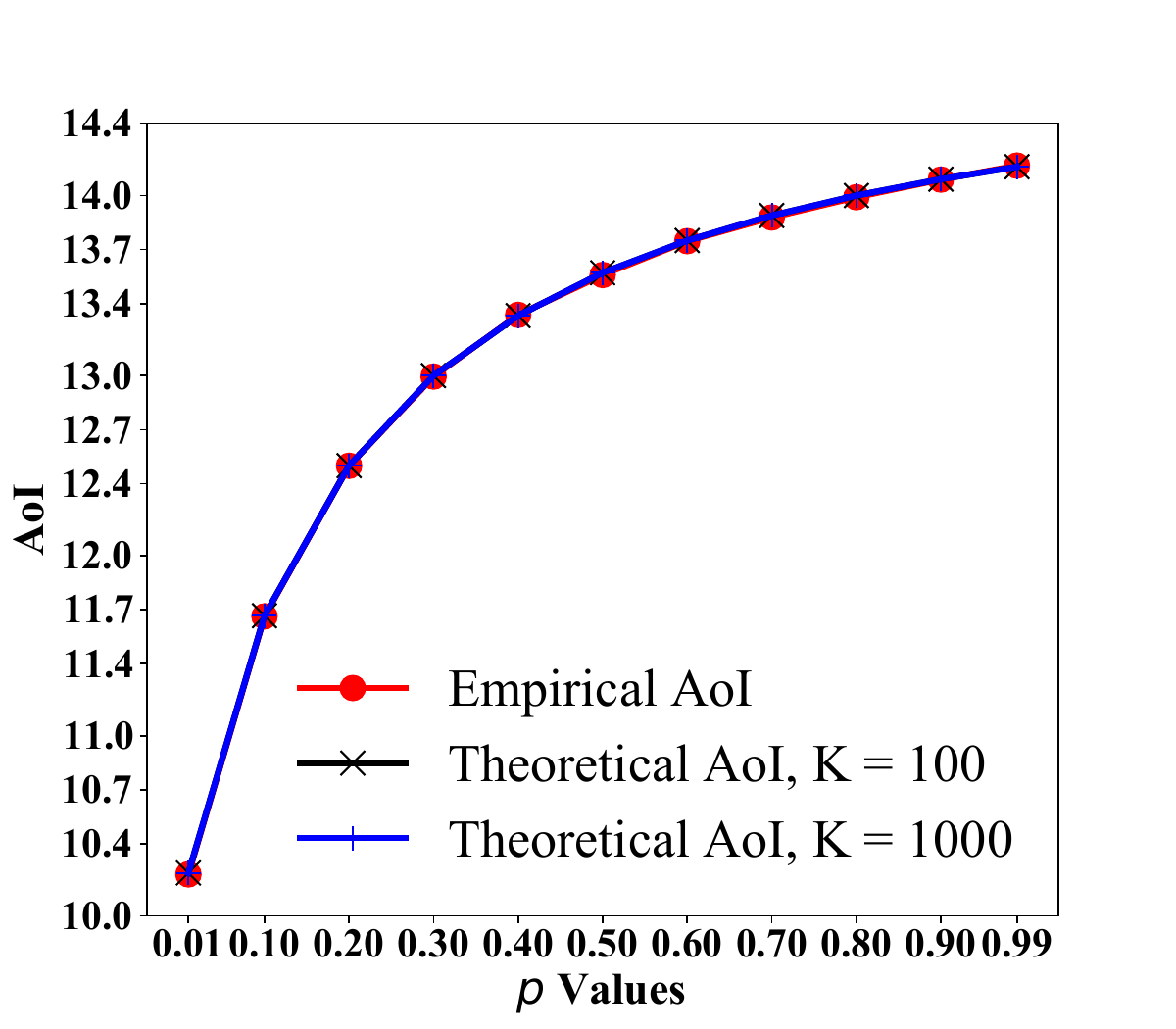}}
\subfigure[$q=0.8$. $\lambda = 1$.]{\includegraphics[width=1.7in, height=1.37in]{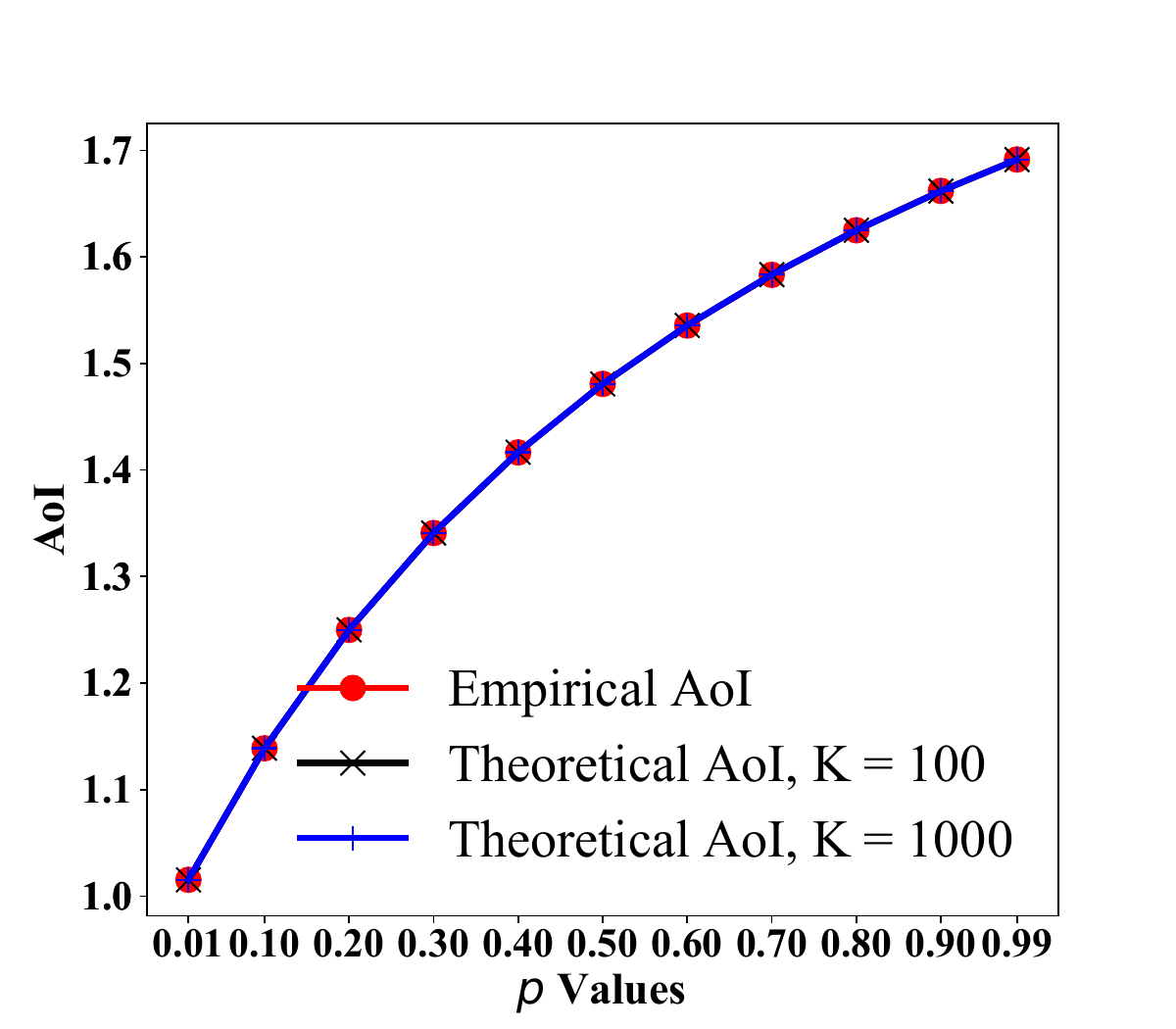}}
\subfigure[$q=0.8$. $\lambda = 0.1$.]{\includegraphics[width=1.7in, 
height=1.37in]{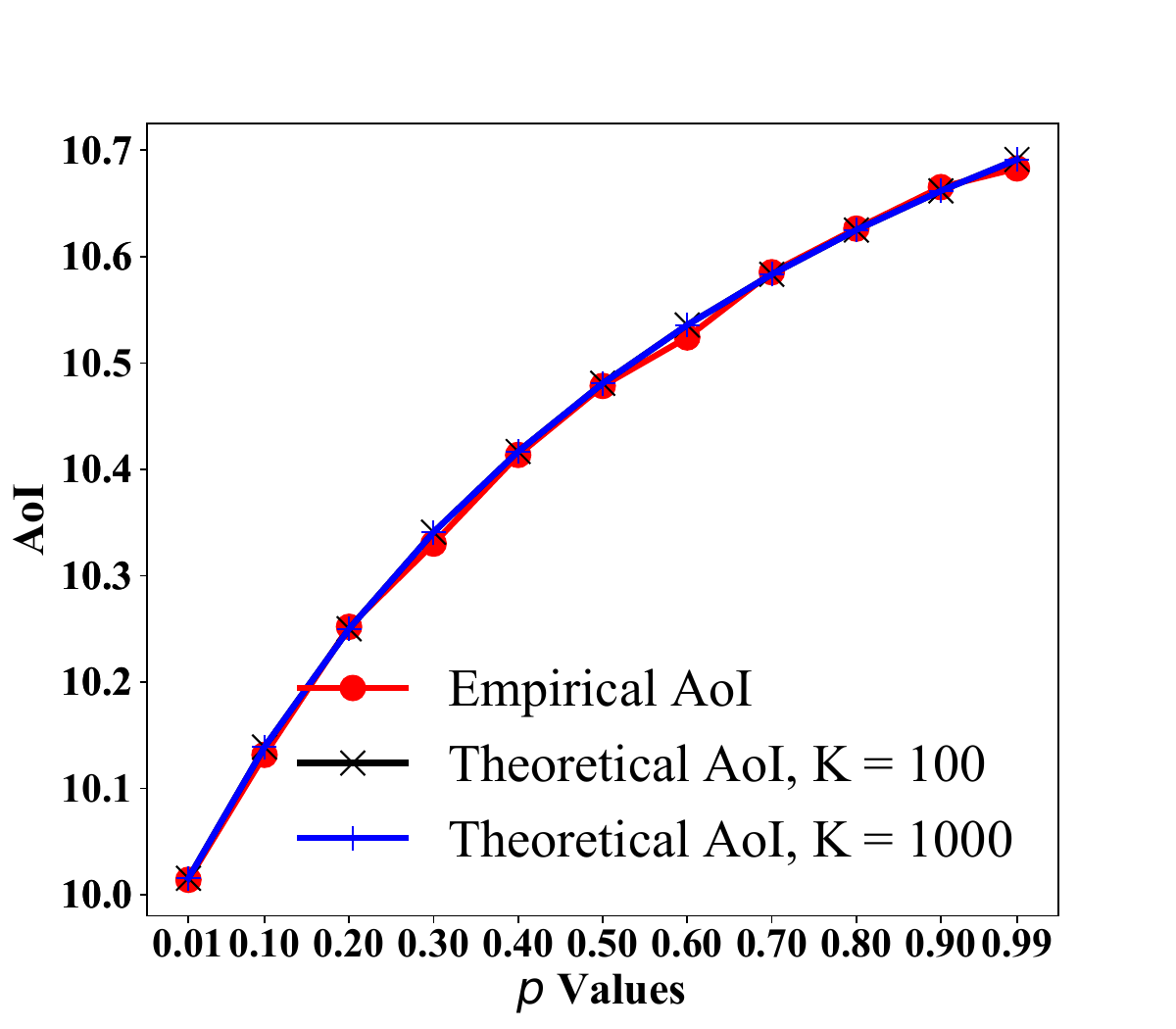}}
\end{center}
\caption{Model validation for single real-time sensing client.}
\label{fig:single_client_validation}
\end{figure}

Next, we evaluate the case when the sole client is a live video streaming client. Given $w_1$ and $\ell_1$, we first choose appropriate $p_1$ and $q_1$ values so that the resulting $m_{1} = 1/w_1$. We then combine \eqref{eq:GE mean}, \eqref{eq:GE variance}, and \eqref{eq:outage_approx} to obtain theoretical approximation of the outage rate $\overline{Out}$. We also obtain the empirical outage rate by simulating the system for $1000$ runs, each containing $500,000$ time slots.

Simulation results averaged over a 1000 runs are shown in Fig.~\ref{fig:timely_throughput_single_client_validation} for four channel and period settings. It is shown that the empirical and theoretical outage rate become virtually the same when the delay $\ell_1$ increases under all considered channel and period settings.

\begin{figure}[ht]

\begin{center} 
\subfigure[$p=0.3$. $q=0.3$. $w = 2$.]{\includegraphics[width=1.7in, height=1.37in]{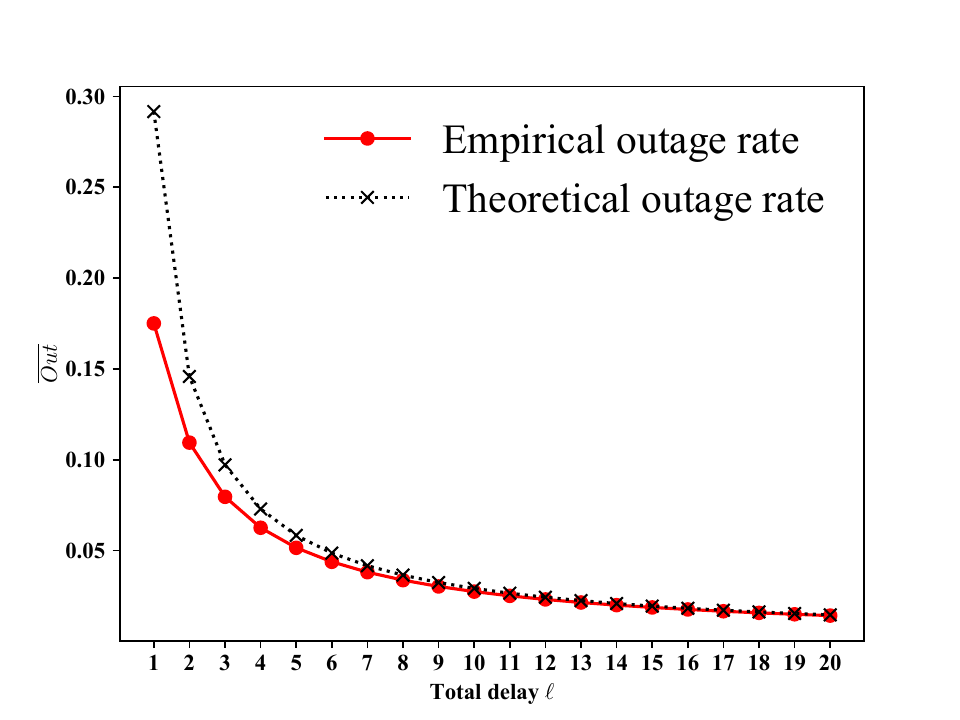}}
\subfigure[$p=0.6$. $q=0.15$. $w = 5$.]{\includegraphics[width=1.7in, height=1.37in]{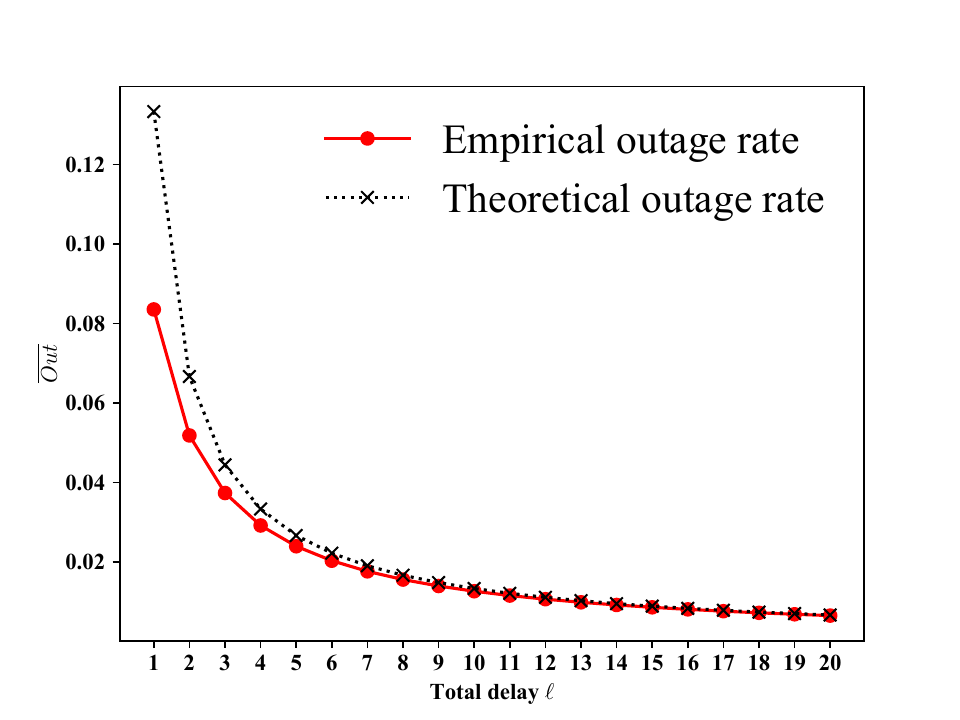}}
\subfigure[$p=0.9$. $q=0.3$. $w = 4$.]{\includegraphics[width=1.7in, height=1.37in]{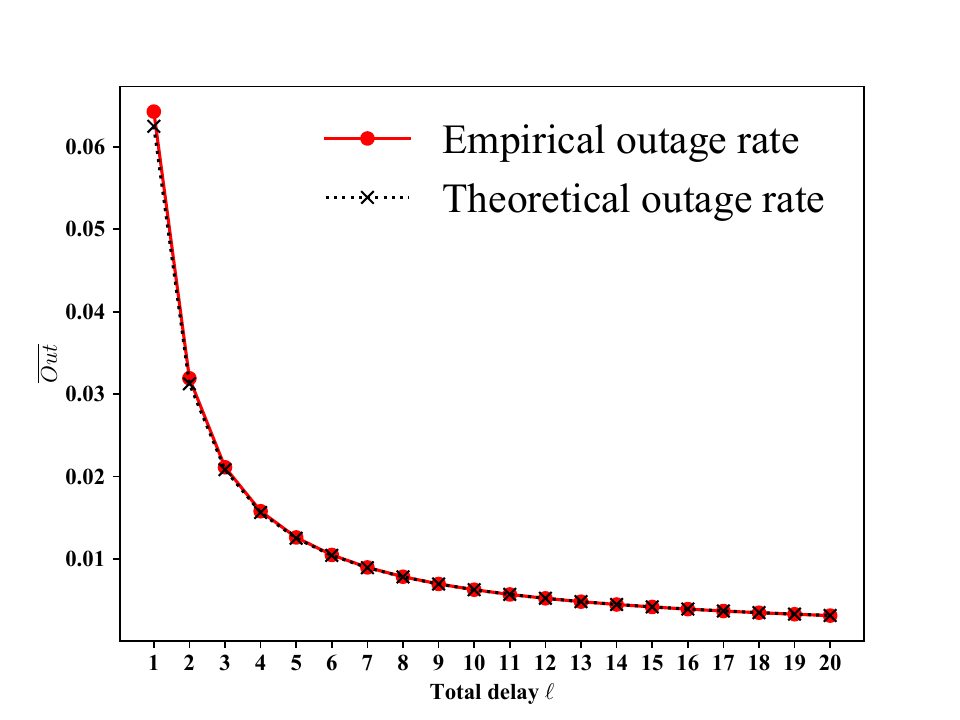}}
\subfigure[$p=0.98$. $q=0.14$. $w = 8$.]{\includegraphics[width=1.7in, 
height=1.37in]{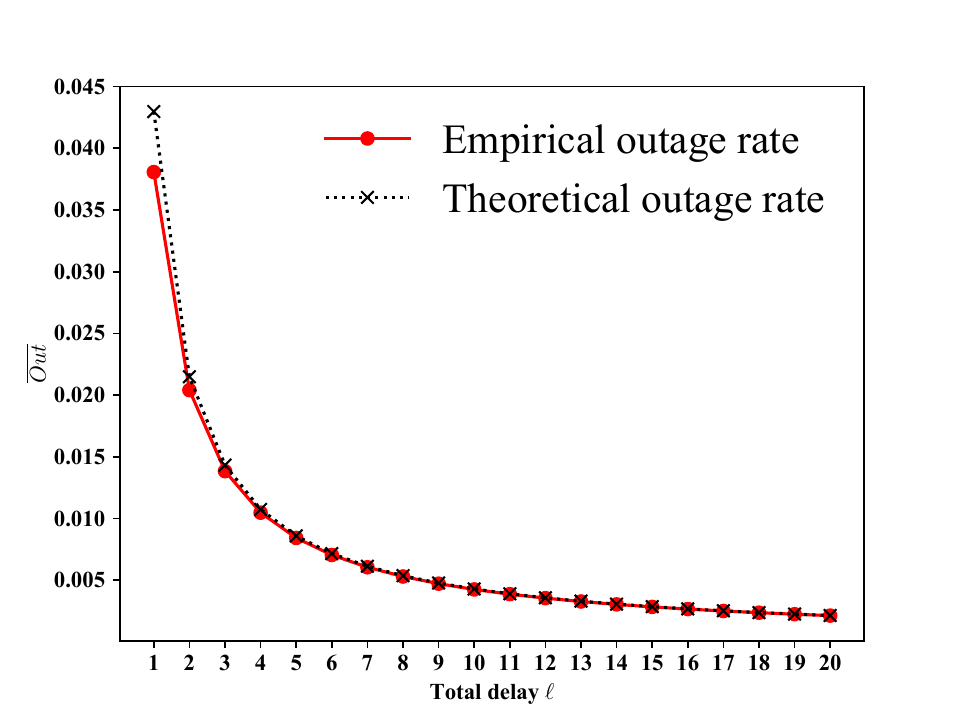}}
\end{center}
\caption{Model validation for single live video streaming client.}
\label{fig:timely_throughput_single_client_validation}
\end{figure}

\subsection{Problem Formulation}

We consider a system with $I$ real-time sensing clients, numbered as $n=1,2, \dots, I$, and $J$ live video streaming clients, numbered as $n=I+1, I+2, \dots, I+J$, with the total number of clients $N = I + J$. Real-time sensing clients want to have a low $\overline{AoI}_n$ while live video streaming clients want to maximize their timely-throughput, or equivalently want to have both a low $\overline{Out}_n$ and a low delay $\ell_n$.
Hence, we aim to minimize the network objective function for $I$ real-time sensing and $J$ live video streaming clients

\begin{align}
&\sum_{n=1}^I \alpha_n \cdot \overline{AoI}_n+\sum_{n=I+1}^{I+J}\beta_n \overline{Out}_n+\gamma_n \ell_n^2 \\
&\approx\sum_{n=1}^I \alpha_n \Big(\frac{1}{2}(\frac{\sigma_n^2}{\mu_n^2}+\frac{1}{\mu_n}) +\frac{1}{\lambda_n}-\frac{1}{2} \Big) + \nonumber \\
&\sum_{n=I+1}^{I+J} \beta_n \big(\frac{\sigma_n^2}{2 \ell_n} \big) + \gamma_n \ell_n^2,
\end{align}

with the values $\alpha_n, \beta_n,$ and $\gamma_n$ being weights chosen for each client. Although the system has two types of clients with very different behaviors and preferences, our framework of second-order optimization allows us to characterize the optimization problem as one that only involves the means and temporal variances of the delivery process of each client.

{\color{blue}A previous work \cite{hsieh20} has studied the above optimization problem when there are only live video streaming clients. However, it requires $\ell_n/ \hspace{0.5em}\overline{Out}_n = \ell_u/ \hspace{0.5em} \overline{Out}_u$ for any $n \neq u$. We note that our formulation does not require this condition. Hence, even for the special case when there are only live video streaming clients, our formulation generalizes the result of \cite{hsieh20}.}

\section{An Outer Bound of the Second-Order Capacity Region} \label{sec: capacity region}

In this section, we derive a necessary condition for the second-order delivery model $\{(\mu_n,\sigma_n^2)|1\leq n\leq N\}$ to be in the second-order capacity region.

\begin{theorem}\label{theorem:outer bound}
Given a second-order channel model $\{(m_S, v_S^2)|S\subseteq\{1,2,\dots,N\}\}$, if a second-order delivery model $\{(\mu_n,\sigma_n^2)|1\leq n\leq N\}$ is in the second-order capacity region, then the following needs to hold:
\begin{align}
    &\sum_{n\in S}\mu_n\leq m_S, \forall S\subseteq\{1,2,\dots,N\},\label{eq:necessary:mean}\\
    &\sum_{n=1}^N\mu_n=m_{\{1,2,\dots,N\}},\label{eq:necessary:total mean}\\
    &\sum_{n=1}^{N}\sqrt{\sigma_n^2}\geq \sqrt{v_{\{1,2,\dots,N\}}^2},\label{eq:necessary:variance}\\
    &\mu_n\geq 0, \forall n.\label{eq:necessary:non-negative}
\end{align}
\end{theorem}
\begin{proof}
We first establish (\ref{eq:necessary:mean}). The AP can transmit a packet to a client $n$ at time $t$ only if the client has an ON channel, that is, $X_n(t)=1$. Moreover, the AP can transmit to at most one client in each time slot. Hence, we have $\sum_{n\in S}Z_n(t)\leq X_S(t)$ under any scheduling policy. This gives us
\begin{align}
    &\sum_{n\in S}\mu_n=\lim_{T\rightarrow\infty}\frac{\sum_{n\in S}\sum_{t=1}^TZ_n(t)}{T}\nonumber\\
    \leq&\lim_{T\rightarrow\infty}\frac{\sum_{t=1}^TX_S(t)}{T}=m_S, \forall S\subseteq\{1,2,\dots, N\}.
\end{align}

We can similarly establish (\ref{eq:necessary:total mean}) by noting that $\sum_{n=1}^NZ_n(t)= X_{\{1,2,\dots,N\}}(t)$, since the AP always transmits one packet as long as at least one client has an ON channel.

Finally, we establish (\ref{eq:necessary:variance}). Let $\hat{X}_S$ be the random variable $\lim_{T\rightarrow\infty}\frac{\sum_{t=1}^TX_S(t)-Tm_S}{\sqrt{T}}$ and $\hat{Z}_n$ be the random variable $\lim_{T\rightarrow\infty}\frac{\sum_{t=1}^TZ_n(t)-T\mu_n}{\sqrt{T}}$. Since $\sum_{n=1}^NZ_n(t)= X_{\{1,2,\dots,N\}}(t)$ and (\ref{eq:necessary:total mean}), {\color{blue}$\sum_{n=1}^N\hat{Z}_n$ and $\hat{X}_{\{1,2,\dots,N\}}$ have the same distribution.}
 We then have
\begin{align}
    &(\sum_{n=1}^N\sqrt{\sigma_n^2})^2=(\sum_{n=1}^N\sqrt{E[\hat{Z}_n^2]})^2\nonumber\\
    =&\sum_{n=1}^NE[\hat{Z}_n^2]+2\sum_{n\neq u}\sqrt{E[\hat{Z}_n^2]E[\hat{Z}_u^2]}\nonumber\\
    \geq&\sum_{n=1}^NE[\hat{Z}_n^2]+2\sum_{n\neq u}E[\hat{Z}_n\hat{Z}_u]\quad(\mbox{Cauchy-Schwarz inequality})\nonumber\\
    =&E[(\sum_{n=1}^N\hat{Z}_n)^2]=E[\hat{X}_{\{1,2,\dots,N\}}^2]=v_{\{1,2,\dots,N\}}^2.
\end{align}
This completes the proof.
\end{proof}

\section{Scheduling Policy with Tight Inner Bound}\label{sec: scheduling policy}

In this section, we derive a sufficient condition for the second-order delivery model $\{(\mu_n,\sigma_n^2)|1\leq n\leq N\}$ to be in the second-order capacity region. We also propose a simple scheduling policy that delivers the desirable second-order delivery models as long as they satisfy the sufficient condition. We state the sufficient condition as follows:

\begin{theorem} \label{theorem:inner bound}
Given a second-order channel model $\{(m_S, v_S^2)|S\subseteq\{1,2,\dots,N\}\}$, a second-order delivery model $\{(\mu_n,\sigma_n^2)|1\leq n\leq N\}$ is in the second-order capacity region if

\begin{align}
    &\sum_{n\in S}\mu_n< m_S, \forall S\subsetneq\{1,2,\dots,N\},\label{eq:sufficient:mean}\\
    &\sum_{n=1}^N\mu_n=m_{\{1,2,\dots,N\}},\label{eq:sufficient:total mean}\\
    &\sum_{n=1}^{N}\sqrt{\sigma_n^2}\geq \sqrt{v_{\{1,2,\dots,N\}}^2},\label{eq:sufficient:variance}\\
    &\mu_n\geq 0, \sigma_n^2>0 \forall n.\label{eq:sufficient:non-negative}
\end{align}
$
\Box
$
\end{theorem}

Before proving Theorem~\ref{theorem:inner bound}, we first discuss its implications. Comparing the conditions in Theorems~\ref{theorem:outer bound} and \ref{theorem:inner bound}, we note that the only difference is that the sufficient condition requires strict inequality for (\ref{eq:necessary:mean}) for all proper subsets. Hence, the sufficient condition describes an inner bound that is almost tight except on some boundaries.

We prove Theorem~\ref{theorem:inner bound} by proposing a scheduling that achieves every point in the inner bound. Given $\{(\mu_n,\sigma_n^2)|1\leq n\leq N\}$, define the \emph{deficit} of a client $n$ at time $t$ as $d_n(t) = t\mu_n-\sum_{\tau=1}^tZ_n(\tau)$. In each time slot $t$, the AP chooses the client with the largest $d_n(t-1)/\sqrt{\sigma_n^2}$ among those with ON channels and transmits a packet to the chosen client. We call this scheduling policy the \emph{variance-weighted-deficit} (VWD) policy. 

We now analyze the performance of the VWD policy. Let $D(t):=\sum_{n=1}^Nd_n(t)/\sum_{n=1}^N\sqrt{\sigma_n^2}$. We then have 
\begin{align}
    &\Delta d_n(t):=d_n(t)-d_n(t-1)=\mu_n-Z_n(t),\\
    &\Delta D(t):=D(t)-D(t-1)\nonumber\\
    =&\frac{\sum_{n=1}^N\mu_n-\sum_{n=1}^NZ_n(t)}{\sum_{n=1}^N\sqrt{\sigma_n^2}}\nonumber \\
    =&\frac{m_{\{1,2,\dots,N\}}-X_{\{1,2,\dots,N\}}(t)}{\sum_{n=1}^N\sqrt{\sigma_n^2}}.
\end{align}
Consider the Lyapunov function $L(t):=\frac{1}{2}\sum_{n=1}^N\sqrt{\sigma_n^2}\Big(\frac{d_n(t)}{\sqrt{\sigma_n^2}}-D(t)\Big)^2$. Let $H^t$ be the system history, {\color{blue}that is, all events of channels, packet generations/deliveries, and scheduling decisions,} up to time $t$. We can derive the expected one-step Lyapunov drift as
\begin{align}
    &\Delta (L(t)) := E[L(t) - L(t-1)|H^{t-1}]\notag \\ 
    =& E[\frac{1}{2} \sum_{n=1}^N \sqrt{\sigma_n^2}\Big(\frac{d_n(t)}{\sqrt{\sigma_n^2}}-D(t))\Big)^2\nonumber\\
    &-\frac{1}{2} \sum_{n=1}^N\sqrt{\sigma_n^2}\Big(\frac{d_n(t-1)}{\sqrt{\sigma_n^2}}-D(t-1)\Big)^2|H^{t-1}]\nonumber\\
    =&E[\sum_{n=1}^N\sqrt{\sigma_n^2}\Big(\frac{d_n(t-1)}{\sqrt{\sigma_n^2}}-D(t-1)\Big)\Big(\frac{\Delta d_n(t)}{\sqrt{\sigma_n^2}}-\Delta D(t)\Big)\nonumber\\
    &+\frac{1}{2} \sum_{n=1}^N \sqrt{\sigma_n^2}\Big(\frac{\Delta d_n(t)}{\sqrt{\sigma_n^2}}-\Delta D(t))\Big)^2|H^{t-1}]\nonumber\\
    \leq &B + E[\sum_{n=1}^N\Big(\frac{d_n(t-1)}{\sqrt{\sigma_n^2}}-D(t-1)\Big)\Delta d_n(t)\nonumber\\
    &-\sum_{n=1}^N\sqrt{\sigma_n^2}\Big(\frac{d_n(t-1)}{\sqrt{\sigma_n^2}}-D(t-1)\Big)\Delta D(t)|H^{t-1}]\nonumber\\
    =&B + E[\sum_{n=1}^N\Big(\frac{d_n(t-1)}{\sqrt{\sigma_n^2}}-D(t-1)\Big)\Delta d_n(t)|H^{t-1}], \label{eq:one-step-Lyapunov}
\end{align}
where $B$ is a bounded constant. The last two steps follow because {\color{blue}$|\Delta d_n(t)|\leq 1$ and $|\Delta D(t)|\leq \frac{1}{\sum_{n=1}^N\sqrt{\sigma_n^2}}$} are bounded and because $\sum_{n=1}^Nd_n(t-1)=\sum_{n=1}^N\sqrt{\sigma_n^2}D(t-1)$.

The VWD policy schedules the client with the largest $d_n(t-1)/\sqrt{\sigma_n^2}$, which is also the client with the largest $d_n(t-1)/\sqrt{\sigma_n^2} - D(t-1)$, among those with ON channels. Hence, under the VWD policy, the system can be modeled as a Markov process whose state consists of the channel states and $d_n(t-1)/\sqrt{\sigma_n^2} - D(t-1)$ of all clients. Further, the VWD policy is the policy that minimizes $E[\sum_{n=1}^N\Big(\frac{d_n(t-1)}{\sqrt{\sigma_n^2}}-D(t-1)\Big)\Delta d_n(t)|H^{t-1}]$ for all $t$. We first show that the Markov process is positive-recurrent.

\begin{lemma}
Assume that (\ref{eq:sufficient:mean}) -- (\ref{eq:sufficient:non-negative}) are satisfied. {\color{blue}Also assume that $\mu_n$ and $\sqrt{\sigma_n^2}$ are rational numbers for all $n$.} Then, under the VWD policy, the system-wide Markov process, whose state consists of the channel states and $d_n(t-1)/\sqrt{\sigma_n^2} - D(t-1)$ of all clients, is positive-recurrent.
\end{lemma}
\begin{proof}
Due to (\ref{eq:sufficient:mean}), we can define 
\begin{equation}
    \delta := \min\{m_S-\sum_{n\in S}\mu_n|S\subsetneq\{1,2,\dots,N\}\}>0. \label{eq:delta}
\end{equation}
Further, since the channel of each client follows a positive-recurrent Markov process with finite states, there exists a finite number $\mathbb{T}$ such that
\begin{equation}
    \mathbb{T}m_S-\frac{\delta}{2}\leq E[\sum_{t=\tau+1}^{\tau+\mathbb{T}}X_S(t)|H^{\tau}]\leq \mathbb{T}m_S+\frac{\delta}{2}, \label{eq:T-steps-mean}
\end{equation}
for any $H^{\tau}$.

Let $L^V(t)$ and $\Delta d_n^V(t)$ be the values of $L(t)$ and $d_n(t)$ under the VWD policy. From (\ref{eq:one-step-Lyapunov}), we can bound the $\mathbb{T}$-step Lyapunov drift by
\begin{align}
    &E[L^V(\tau+\mathbb{T})-L^V(\tau)|H^\tau]\nonumber\\
    \leq &B\mathbb{T}+E[\sum_{t=\tau+1}^{\tau+\mathbb{T}}\sum_{n=1}^N\Big(\frac{d_n(t-1)}{\sqrt{\sigma_n^2}}-D(t-1)\Big)\Delta d_n^V(t)|H^{\tau}]\nonumber\\
    \leq &B\mathbb{T}+E[\sum_{t=\tau+1}^{\tau+\mathbb{T}}\sum_{n=1}^N\Big(\frac{d_n(t-1)}{\sqrt{\sigma_n^2}}-D(t-1)\Big)\Delta d_n^\eta(t)|H^{\tau}]\nonumber\\
    \leq &A+E[\sum_{n=1}^N\Big(\frac{d_n(\tau)}{\sqrt{\sigma_n^2}}-D(\tau)\Big)(\sum_{t=\tau+1}^{\tau+\mathbb{T}}\Delta d_n^\eta(t))|H^{\tau}], \label{eq:t steps drift}
\end{align}
for any other scheduling policy $\eta$, where $d_n^\eta(t)$ is the value of $d_n(t)$ under $\eta$ and $A$ is a bounded constant. The last inequality follows because $\mathbb{T}$, $|d_n(t)-d_n(\tau)|$, and $\Delta d_n(t)$ are all bounded for all $t\in [\tau+1, \tau+\mathbb{T}]$.

We now consider the scheduling policy $\eta$ that schedules the flow with the largest $d_n(\tau)/\sqrt{\sigma_n^2}$ among those with ON channels in all time slots $t\in [\tau+1, \tau+\mathbb{T}]$.

Without loss of generality, we assume that $d_1(\tau)/\sqrt{\sigma_1^2}\geq d_2(\tau)/\sqrt{\sigma_2^2}\geq\dots$. Under $\eta$, a client $n$ will be scheduled in time slot $t$ if it has an ON channel and all clients in $\{1,2,\dots, n-1\}$ have OFF channels, that is, $X_{\{1,2,\dots n\}}(t)=1$ and $X_{\{1,2,\dots n-1\}}(t)=0$. We hence have $\sum_{t=\tau+1}^{\tau+\mathbb{T}}Z_n(t)=\sum_{t=\tau+1}^{\tau+\mathbb{T}}X_{\{1,2,\dots n\}}(t)-\sum_{t=\tau+1}^{\tau+\mathbb{T}}X_{\{1,2,\dots n-1\}}(t)$. Therefore,
\begin{align}
    &E[\sum_{n=1}^N\Big(\frac{d_n(\tau)}{\sqrt{\sigma_n^2}}-D(\tau)\Big)(\sum_{t=\tau+1}^{\tau+\mathbb{T}}\Delta d_n^\eta(t))|H^{\tau}]\nonumber\\
    =&E[\sum_{n=1}^{N-1}\Big(\frac{d_n(\tau)}{\sqrt{\sigma_n^2}}-\frac{d_{n+1}(\tau)}{\sqrt{\sigma_{n+1}^2}}\Big)(\mathbb{T}\sum_{u=1}^n\mu_u\nonumber\\
    &-\sum_{t=\tau+1}^{\tau+\mathbb{T}}X_{\{1,2,\dots,n\}}(t))+\Big(\frac{d_N(\tau)}{\sqrt{\sigma_N^2}}-D(\tau)\Big)\nonumber\\
    &\times(\mathbb{T}\sum_{u=1}^N\mu_u-\sum_{t=\tau+1}^{\tau+\mathbb{T}}X_{\{1,2,\dots,N\}}(t))|H^{\tau}]\nonumber\\
    \leq& \sum_{n=1}^{N-1}\Big(\frac{d_n(\tau)}{\sqrt{\sigma_n^2}}-\frac{d_{n+1}(\tau)}{\sqrt{\sigma_{n+1}^2}}\Big)(-\delta/2)
    +\Big(\frac{d_N(\tau)}{\sqrt{\sigma_N^2}}-D(\tau)\Big)(-\delta/2)\nonumber\\
    =&\Big(\frac{d_1(\tau)}{\sqrt{\sigma_1^2}}-D(\tau)\Big)(-\delta/2), \label{eq:eta drift}
\end{align}
where the inequality holds due to (\ref{eq:sufficient:total mean}), (\ref{eq:delta}), and (\ref{eq:T-steps-mean}).

Combining (\ref{eq:t steps drift}) and (\ref{eq:eta drift}), and we have
\begin{equation}
    E[L^V(\tau+\mathbb{T})-L^V(\tau)|H^\tau]<-\delta,
\end{equation}
if $\max_n\Big(\frac{d_n(\tau)}{\sqrt{\sigma_n^2}}-D(\tau)\Big)>2(A/\delta+1)$, and
\begin{equation}
    E[L^V(\tau+\mathbb{T})-L^V(\tau)|H^\tau]\leq A,
\end{equation}
if $\max_n\Big(\frac{d_n(\tau)}{\sqrt{\sigma_n^2}}-D(\tau)\Big)\leq 2(A/\delta+1)$. Recall that $\sum_n \Big(\frac{d_n(\tau-1)}{\sqrt{\sigma_n^2}}-D(\tau-1)\Big)=0$ and the channel of each client follows a Markov process with finite states. Hence, all states of the system with $\max_n\Big(\frac{d_n(\tau)}{\sqrt{\sigma_n^2}}-D(\tau)\Big)\leq 2(A/\delta+1)$ belong to a finite set of states. By the Foster-Lyapunov Theorem, the system-wide Markov process is positive-recurrent {\color{blue}if it is irreducible. Since the channel state of a client is irreducible and $d_n(t-1)/\sqrt{\sigma_n^2} - D(t-1)$ is a rational number that decreases every time a packet is delivered for client $n$ and increases every time a packet is delivered for a client other than $n$, it is trivial to show that the system-wide Markov process is irreducible. This completes the proof.} 
\end{proof}

We now show that the VWD policy delivers all desirable second-order delivery models that satisfy the sufficient conditions (\ref{eq:sufficient:mean}) -- (\ref{eq:sufficient:non-negative}), and thereby establishing Theorem~\ref{theorem:inner bound}.

\begin{theorem}\label{thm: approach solution}
Assume that (\ref{eq:sufficient:mean}) -- (\ref{eq:sufficient:non-negative}) are satisfied. Then, under the VWD policy, $\lim_{T\rightarrow\infty}\frac{\sum_{t=1}^TZ_n(t)}{T}=\mu_n$ and $E[(\lim_{T\rightarrow\infty}\frac{\sum_{t=1}^TZ_n(t)-T\mu_n}{\sqrt{T}})^2]\leq\sigma_n^2,\forall n$.
\end{theorem} 
\begin{proof}

Since the system-wide Markov process is positive recurrent under the VWD policy, we have: \begin{align}
    \lim_{T\to\infty} \frac{d_n(T)/\sqrt{\sigma_n^2}-D(T)}{T} \to 0, \forall n, \label{f:positive recurrent for mean} \\
    \lim_{T\to\infty} \frac{d_n(T)/\sqrt{\sigma_n^2}-D(T)}{\sqrt{T}} \to 0, \forall n. \label{f:positive recurrent for var}
\end{align}

First, we show that $\lim_{T\rightarrow\infty}\frac{\sum_{t=1}^TZ_n(t)}{T}=\mu_n,\forall n$.
Recall that $d_n(t) = t\mu_n - \sum_{\tau=1}^tZ_n(\tau)$ and $D(t) = {\sum_{n=1}^N d_n(t)}/{\sum_{n=1}^N \sqrt{\sigma_n^2}}$.
By (\ref{eq:sufficient:total mean}), we have:
\begin{align}
     &\lim_{T \to \infty} \frac{D(T)}{T} = \lim_{T \to \infty} \frac{\sum_{n=1}^N T\mu_n -\sum_{t=1}^T\sum_{n=1}^N  Z_n(t)}{T\sum_{n=1}^N \sqrt{\sigma_n^2}} \notag \\
     =&\lim_{T\to \infty} \frac{Tm_{\{1,2,\dots,N\}} - \sum_{t=1}^TX_{\{1,2,\dots,N\}}}{{T\sum_{n=1}^N \sqrt{\sigma_n^2}}} = 0. \label{f:proving mean converge for D}
\end{align}
Hence, by (\ref{f:positive recurrent for mean}), we have $\lim_{T\to\infty} \frac{d_n(T)}{T}=\mu_n-\lim_{T\to\infty} \frac{\sum_{t=1}^TZ_n(t)}{T}=0$, for all $n$.

Next, we show that $E[(\lim_{T\rightarrow\infty}\frac{\sum_{t=1}^TZ_n(t)-T\mu_n}{\sqrt{T}})^2]\leq\sigma_n^2,\forall n$. We have, by (\ref{eq:sufficient:variance}),
\begin{align}
    E[(\lim_{T\rightarrow\infty}\frac{D(T)}{\sqrt{T}})^2]=\frac{v^2_{\{1,2,\dots,N\}}}{(\sum_{n=1}^N\sqrt{\sigma_n^2})^2}\leq 1,
\end{align}

and, hence,
\begin{align}
    &E[(\lim_{T\rightarrow\infty}\frac{\sum_{t=1}^TZ_n(t)-T\mu_n}{\sqrt{T}})^2] = E[(\lim_{T\rightarrow\infty}\frac{d_n(T)}{\sqrt{T}})^2]\notag\\
    =&\sigma_n^2E[(\lim_{T\rightarrow\infty}\frac{D(T)}{\sqrt{T}})^2]\leq \sigma_n^2.
\end{align}
\end{proof}

\begin{figure*}[hbt!]
\begin{center} 
\subfigure[N = 5 Clients.]{\includegraphics[width=0.3\textwidth, height=1.8in]{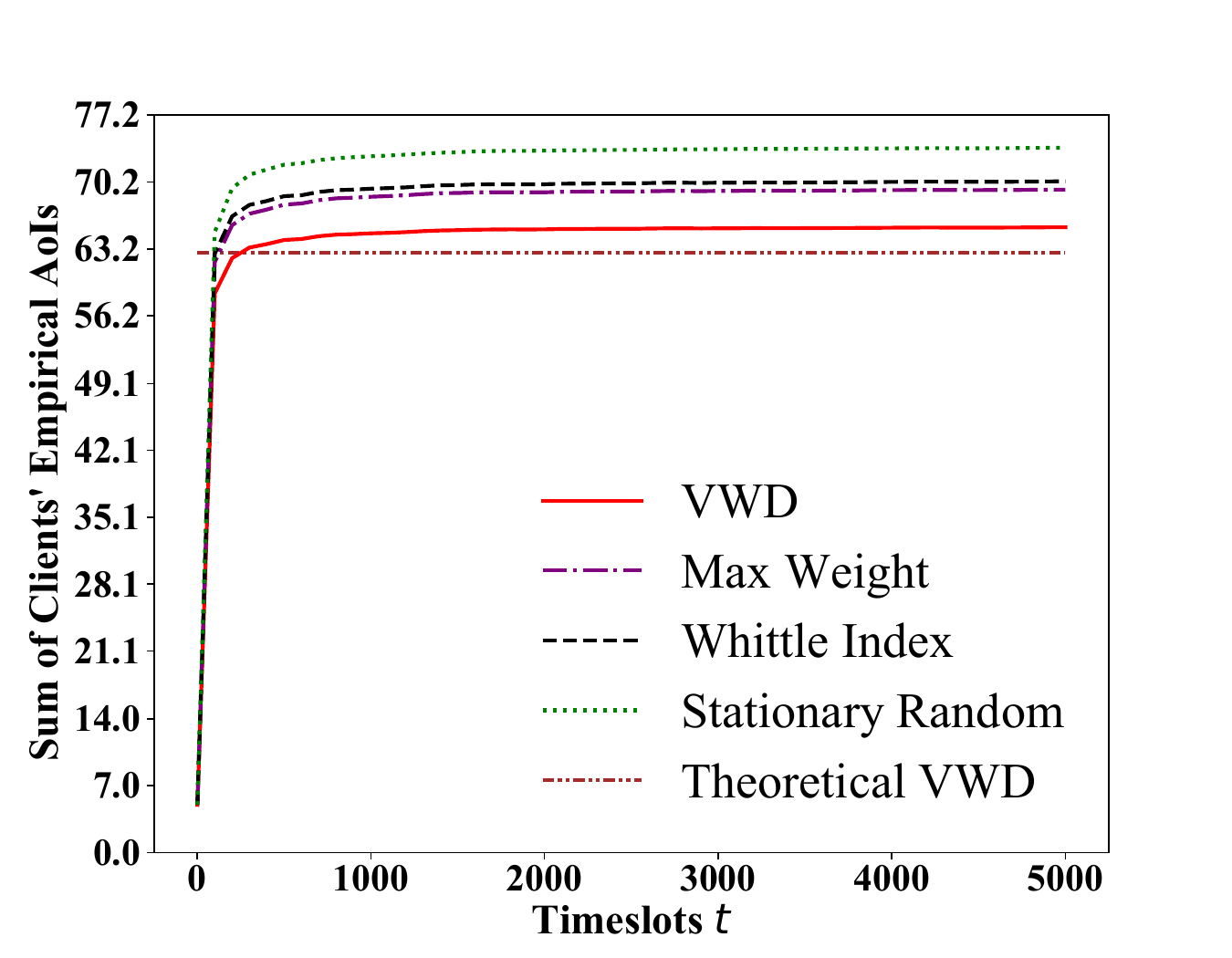}}
\subfigure[N = 10 Clients.]{\includegraphics[width=0.3\textwidth, height=1.8in]{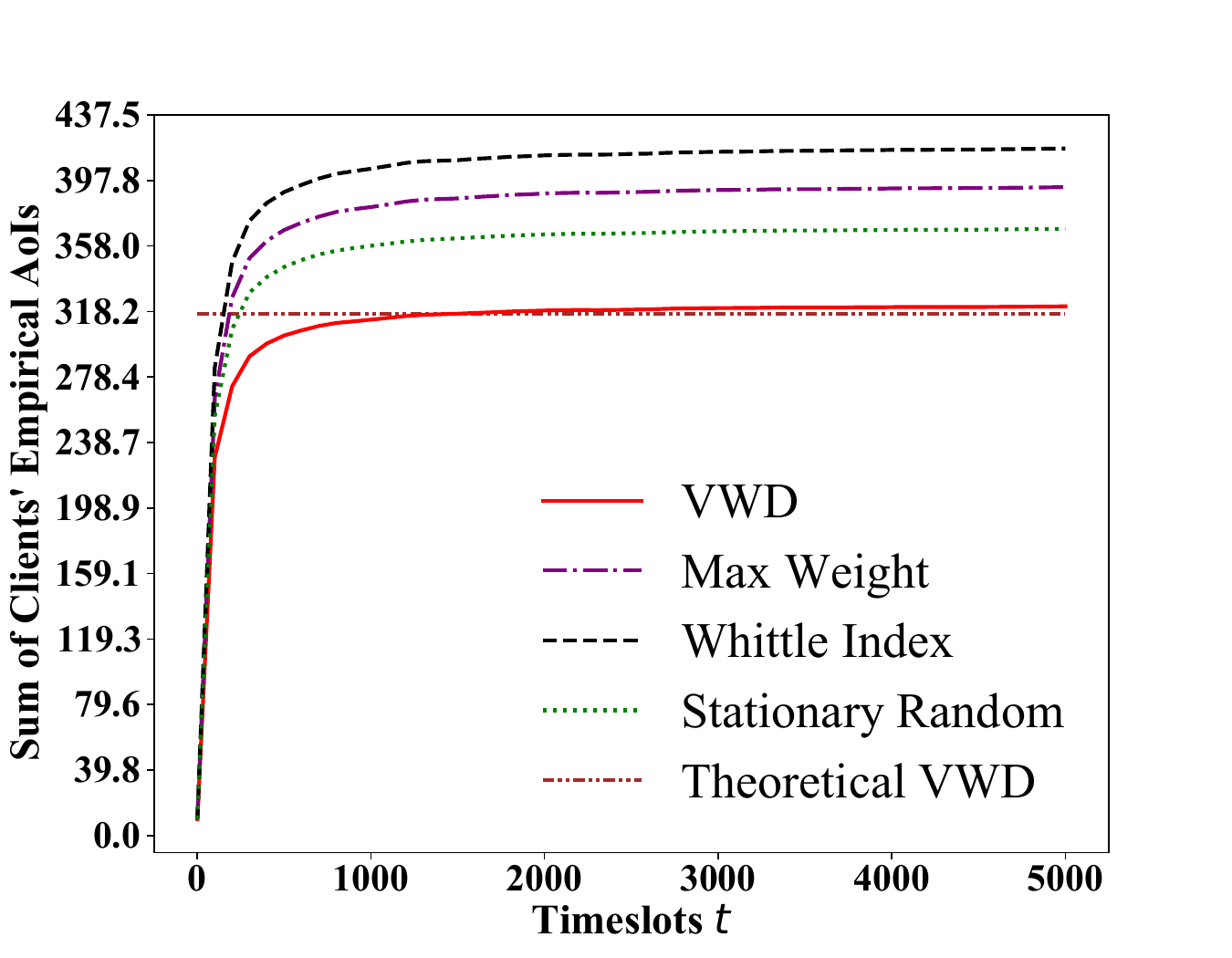}}
\subfigure[N = 20 Clients.]{\includegraphics[width=0.3\textwidth, height=1.8in]{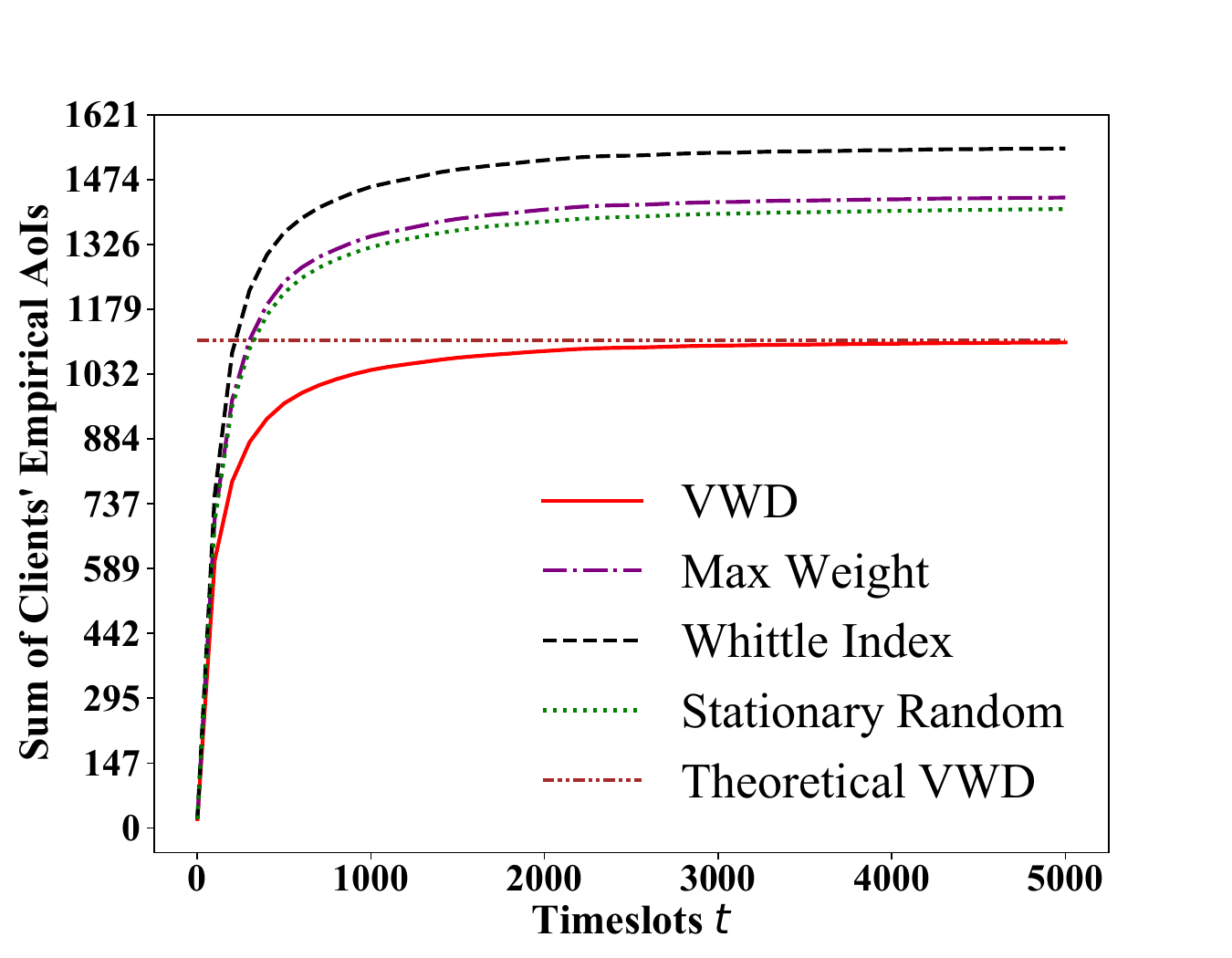}}
\end{center}
\caption{Total non-weighted empirical age of information (AoI) for real-time sensing clients.}
\label{fig:empirical_aoi_sum}
\end{figure*}

We conclude this section by leveraging theorems~\ref{theorem:inner bound} and~\ref{thm: approach solution} to solve the second-order optimization problem. The optimization problem can be written as
\begin{align} \label{eq:objective_function}
    \min & \sum_{n=1}^NF_n(\mu_n,\sigma_n^2) \nonumber \\
    \mbox{s.t. } & \mbox{(\ref{eq:sufficient:mean}) -- (\ref{eq:sufficient:non-negative})}.
\end{align}
The condition (\ref{eq:sufficient:mean}) involves strict inequalities, which cannot be used by standard optimization solvers. We change (\ref{eq:sufficient:mean}) to $\sum_{n\in S}\mu_n\leq m_S-\delta$, where $\delta$ is a small positive number. After the change, the optimization problem can be directly solved by standard solvers to find the optimal $\{\mu_n, \sigma_n^2|1\leq n\leq N\}$. After finding the optimal $\{\mu_n, \sigma_n^2|1\leq n\leq N\}$, one can use the VWD policy to attain the optimal network performance for $N$ wireless clients.

\section{Simulation Results} \label{sec:simulation}

\begin{figure*}[hbt!]
\begin{center} 
\subfigure[N = 5 Clients.]{\includegraphics[width=0.3\textwidth, height=1.8in]{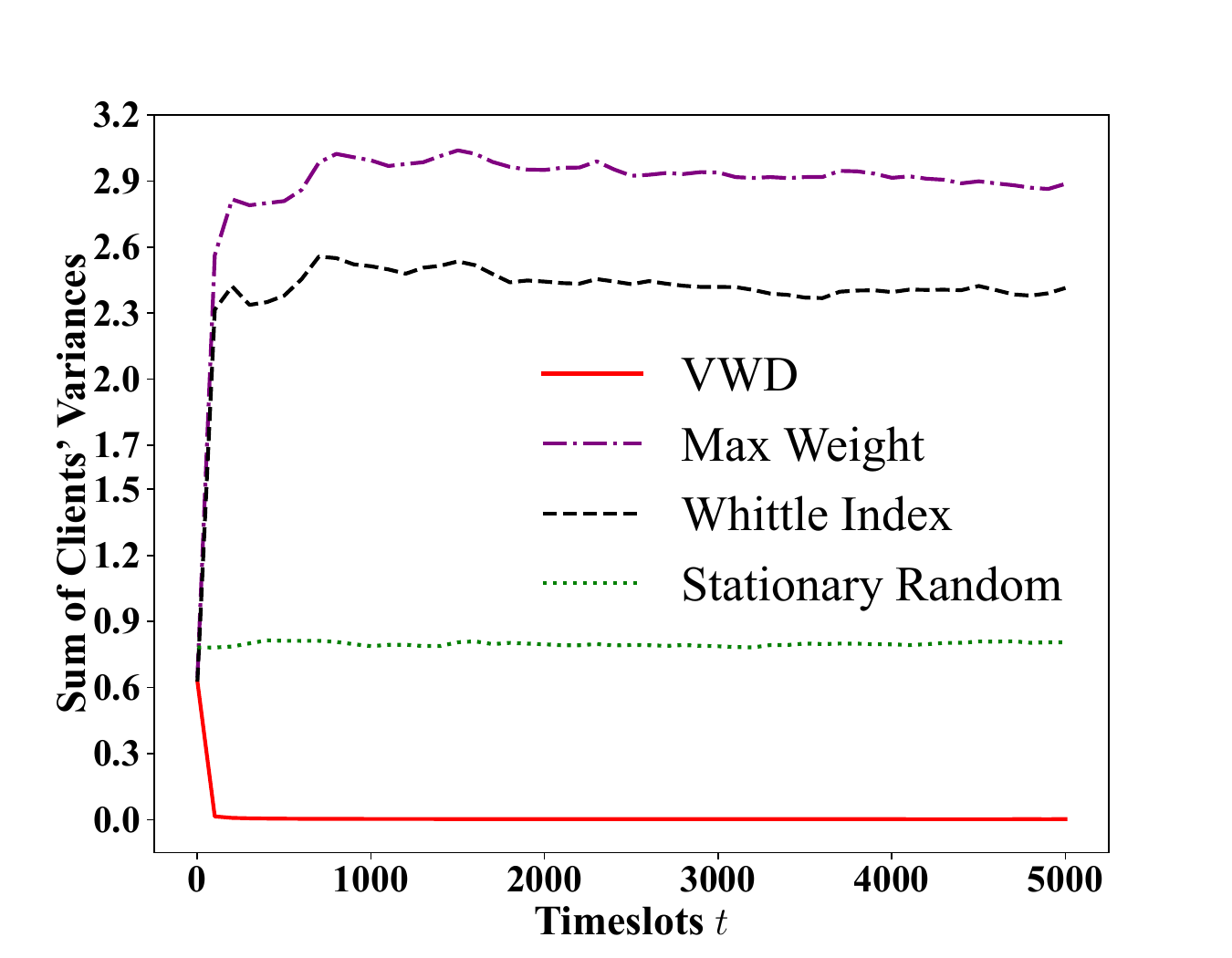}}
\subfigure[N = 10 Clients.]{\includegraphics[width=0.3\textwidth, height=1.8in]{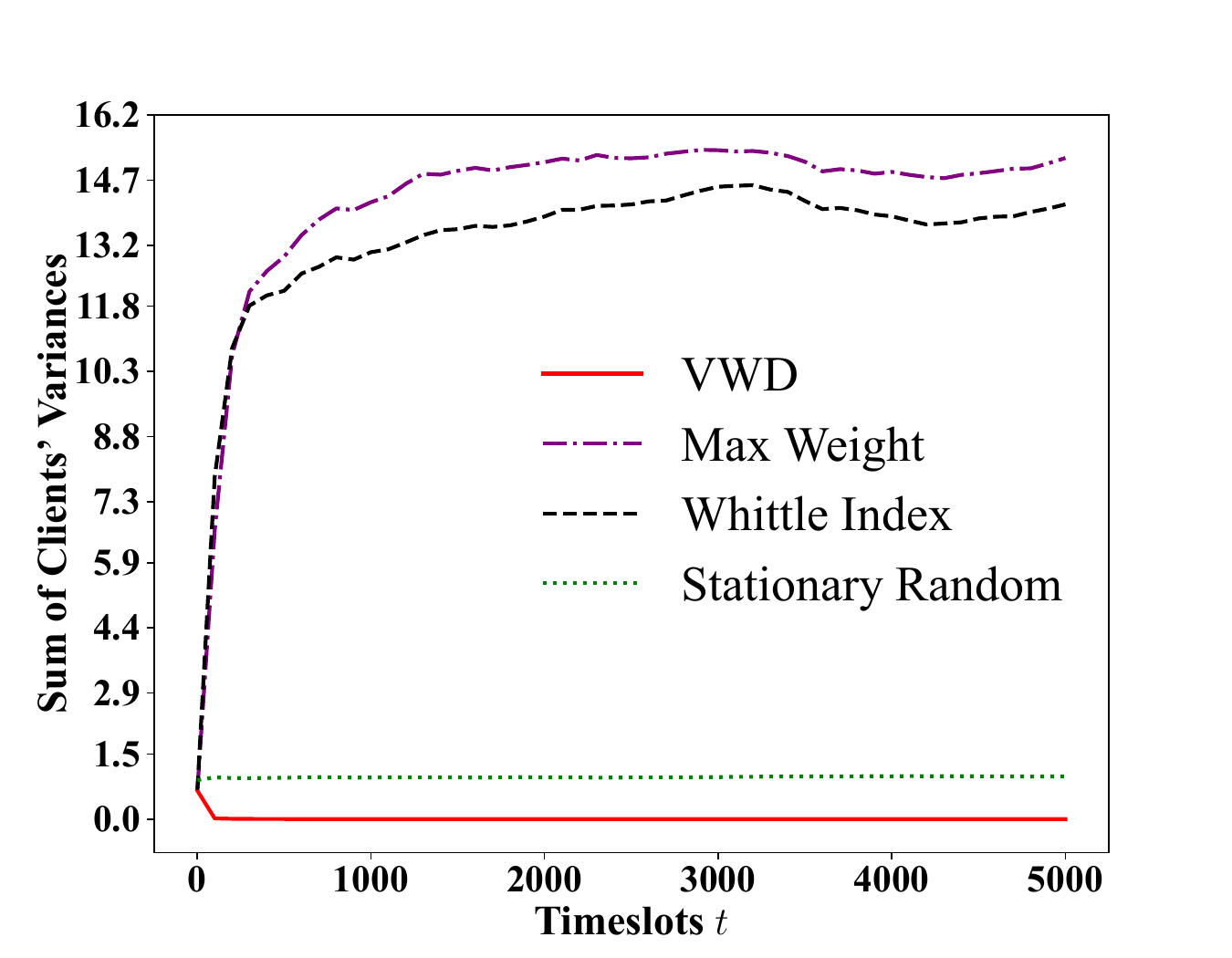}}
\subfigure[N = 20 Clients.]{\includegraphics[width=0.3\textwidth, height=1.8in]{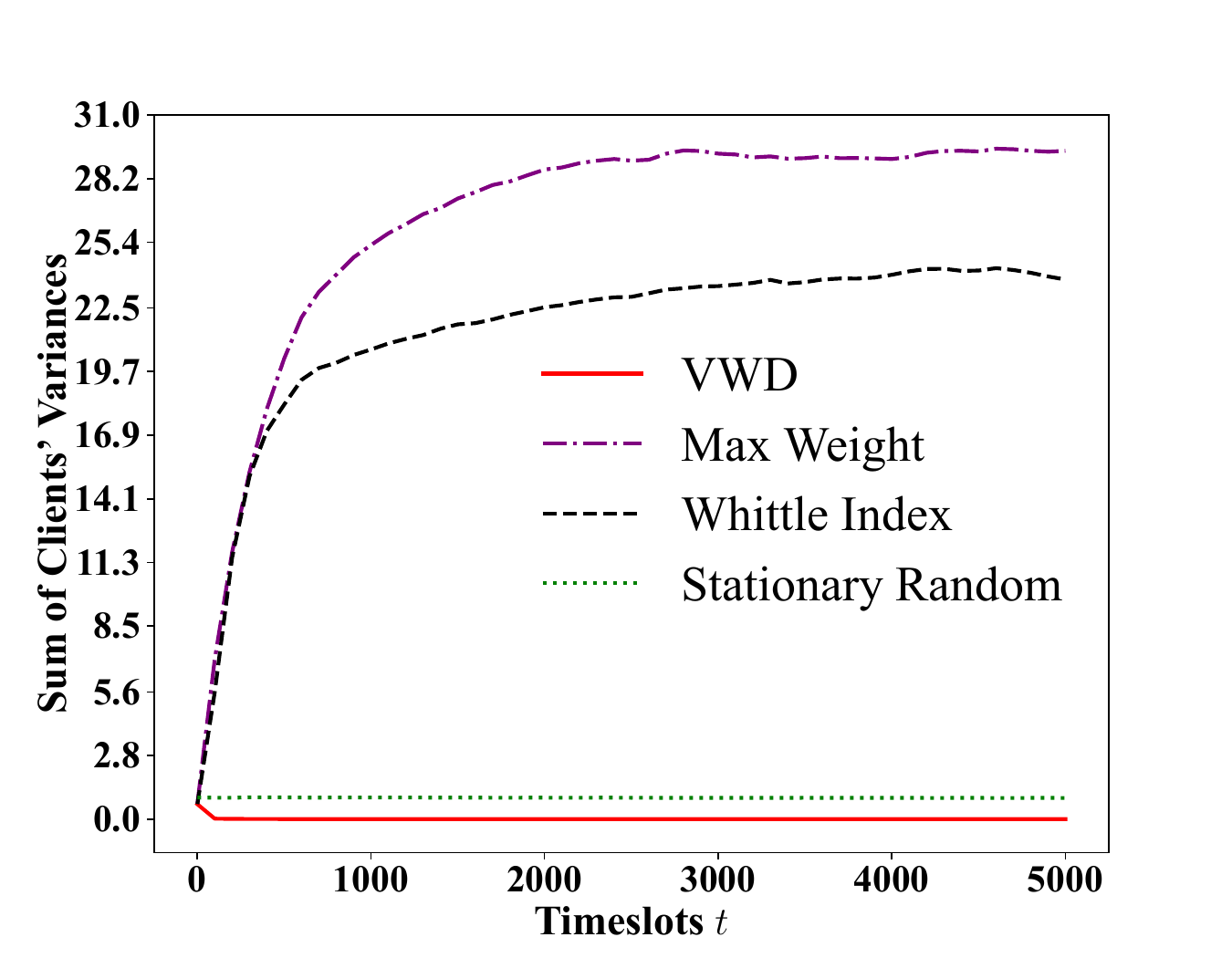}}
\end{center}
\caption{Empirical variance of all real-time sensing clients.}
\label{fig:variance_policies}
\end{figure*}

In this section, we present the simulation results for the proposed scheduler VWD for three different cases, one where all clients are real-time sensing clients, one where all clients are live video streaming clients, and one where both kinds of clients are present.
We compare our policy, VWD, against baseline policies designed for either AoI minimization or timely-throughput maximization.

\subsection{Real-Time Sensing Clients Optimization}  \label{subsection:aoi_min}

The objective is to minimize the system-wide AoI of $N = I$ real-time sensing clients, $\sum_n \alpha_n \overline{AoI}_n$, where $\alpha_n$ is the weight of client $n$. The system model is the one discussed in Section~\ref{sec: aoi}. Each client has a Gilbert-Elliott channel with transition probabilities $p_n$ and $q_n$. In each time slot, each client $n$ generates a new packet with probability $\lambda_n$. VWD is evaluated against three recently designed scheduling policies on the AoI minimization problem. We provide a description of each policy, along with modifications needed to fit the testing setting.

\begin{itemize}
\vspace{0.5em}
    \item \textbf{Whittle index policy}: This policy is based on the Whittle index policy by Hsu in \cite{Hsu18}. Under our setting, the policy calculates an index for ON clients based on their AoIs as $W_n(t) = \frac{AoI_n^2(t)}{2} - \frac{AoI_n(t)}{2} + \frac{AoI_n(t)}{q_n/(p_n+q_n)}$, and then schedules the ON client with the largest index. Hsu in \cite{Hsu18} has shown that $W_n(t)$ is indeed the Whittle index of a client when the channel is i.i.d., i.e., $p_n+q_n=1$, and $\lambda_n=1$.
    \vspace{0.5em}
    \item \textbf{Stationary randomized policy}: This policy calculates a weight $\mu_n$ for each client. In each time slot, it randomly picks an ON client, with the probability of picking $n$ being proportional to $\mu_n$. In the setting of Kadota and Modiano \cite{kadota2019minimizing}, it has been shown that, when $\mu_n$ is properly chosen, this policy achieves an approximation ratio of four in terms of total weighted AoI. In our setting, we choose $\mu_n$ to be the optimal $\mu_n$ from solving (\ref{eq:objective_function}). 
    \vspace{0.5em}
    \item \textbf{Max weight policy} \cite{kadota2019minimizing}: This policy schedules the ON client with the largest $(AoI_n(t) - z_n(t))/\mu_n$. In the setting of Kadota and Modiano \cite{kadota2019minimizing}, $z_n(t)$ is the time since client $n$ generates the latest packet. It has been shown that the total weighted AoI under this policy is no larger than that under the stationary randomized policy, and therefore this policy also achieves an approximation ratio of four. In our setting, the AP does not know when each client generates a new packet. Hence, we choose $z_n(t)$ to be $\frac{1}{\lambda_n}$, which is the expected time since client $n$ generates the latest packet.
    \vspace{0.5em}
\end{itemize}

\begin{figure*}[hbt!]
\begin{center} 
\subfigure[N = 5 Clients.]{\includegraphics[width=0.3\textwidth, height=1.8in]{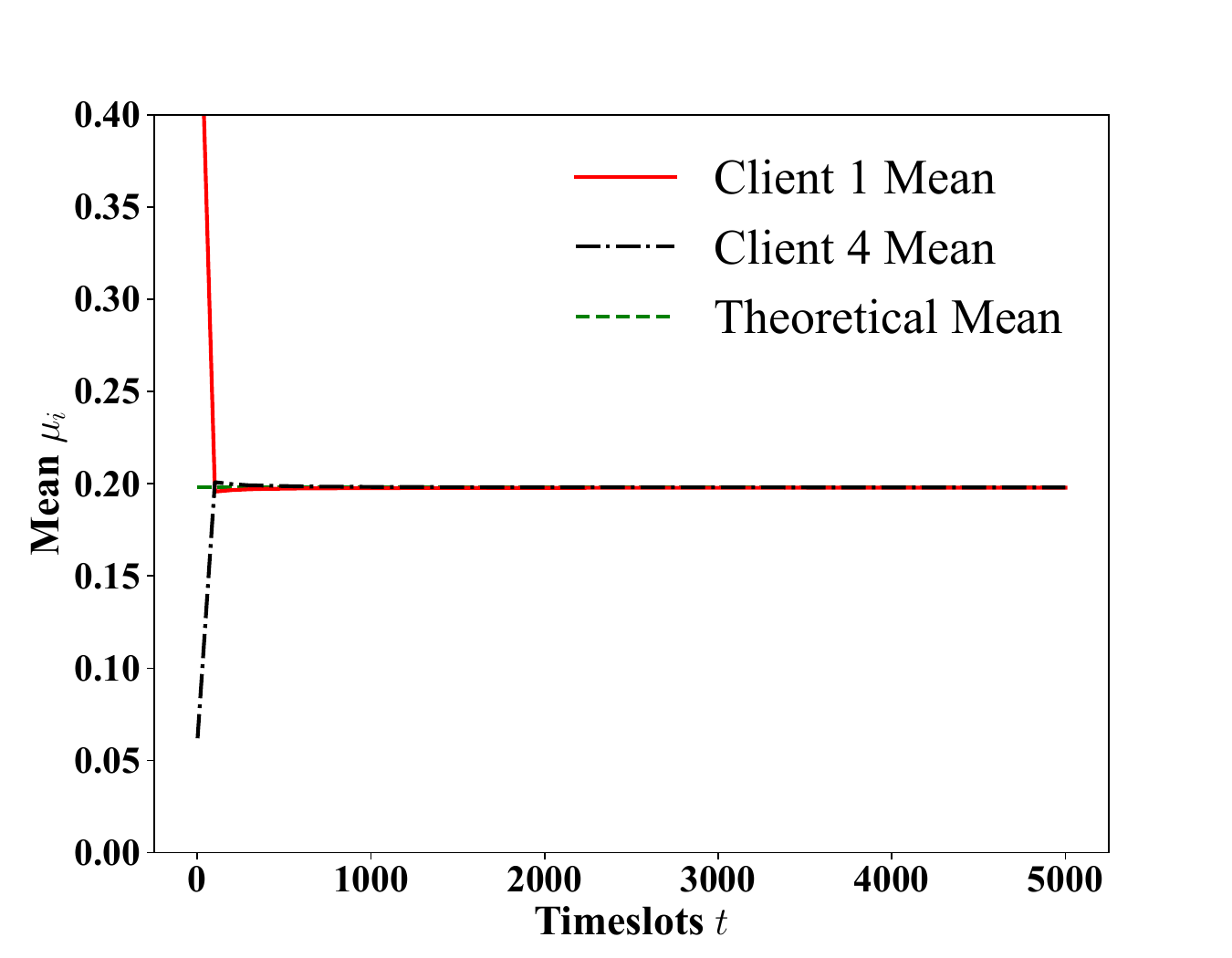}}
\subfigure[N = 10 Clients.]{\includegraphics[width=0.3\textwidth, height=1.8in]{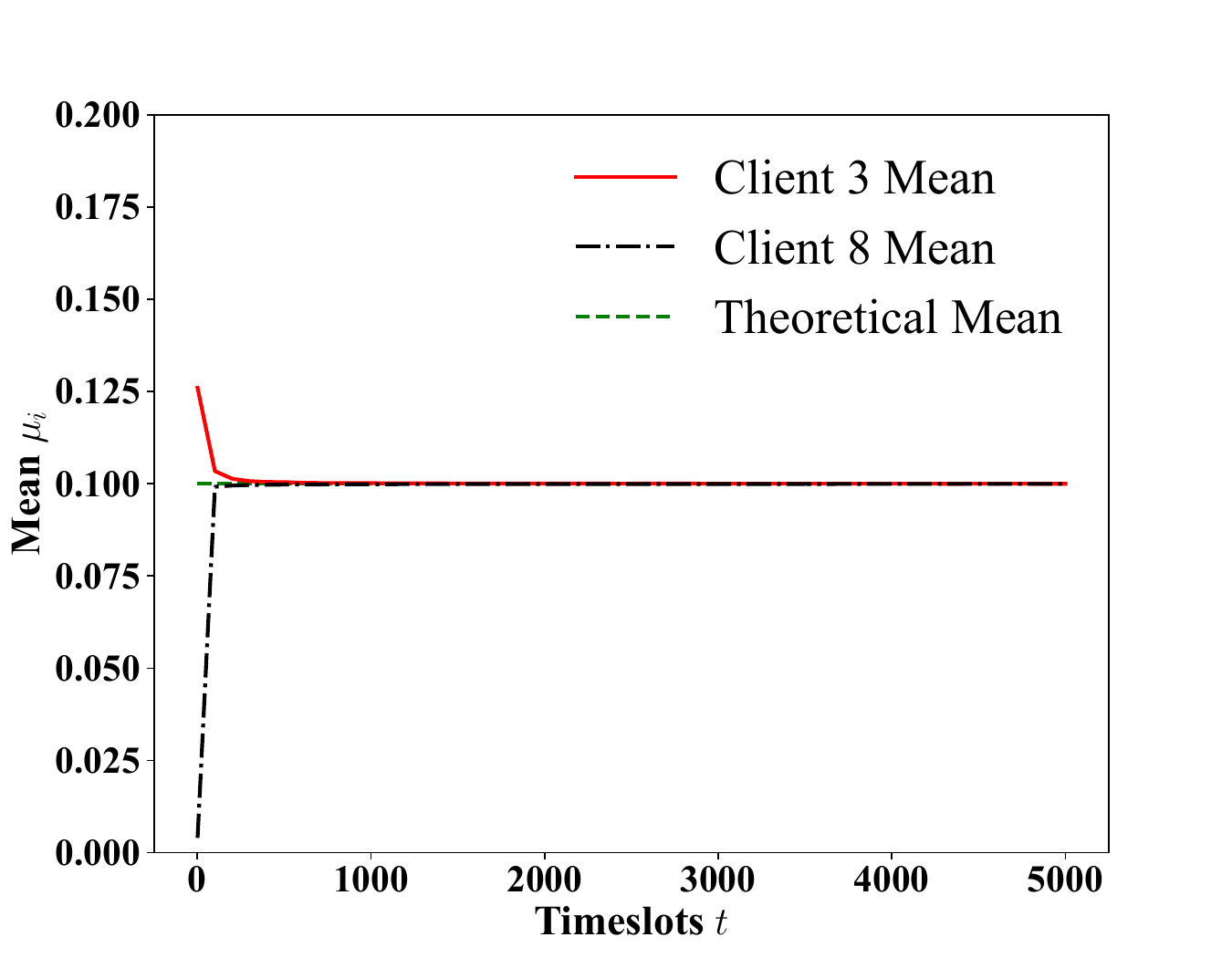}}
\subfigure[N = 20 Clients.]{\includegraphics[width=0.3\textwidth, height=1.8in]{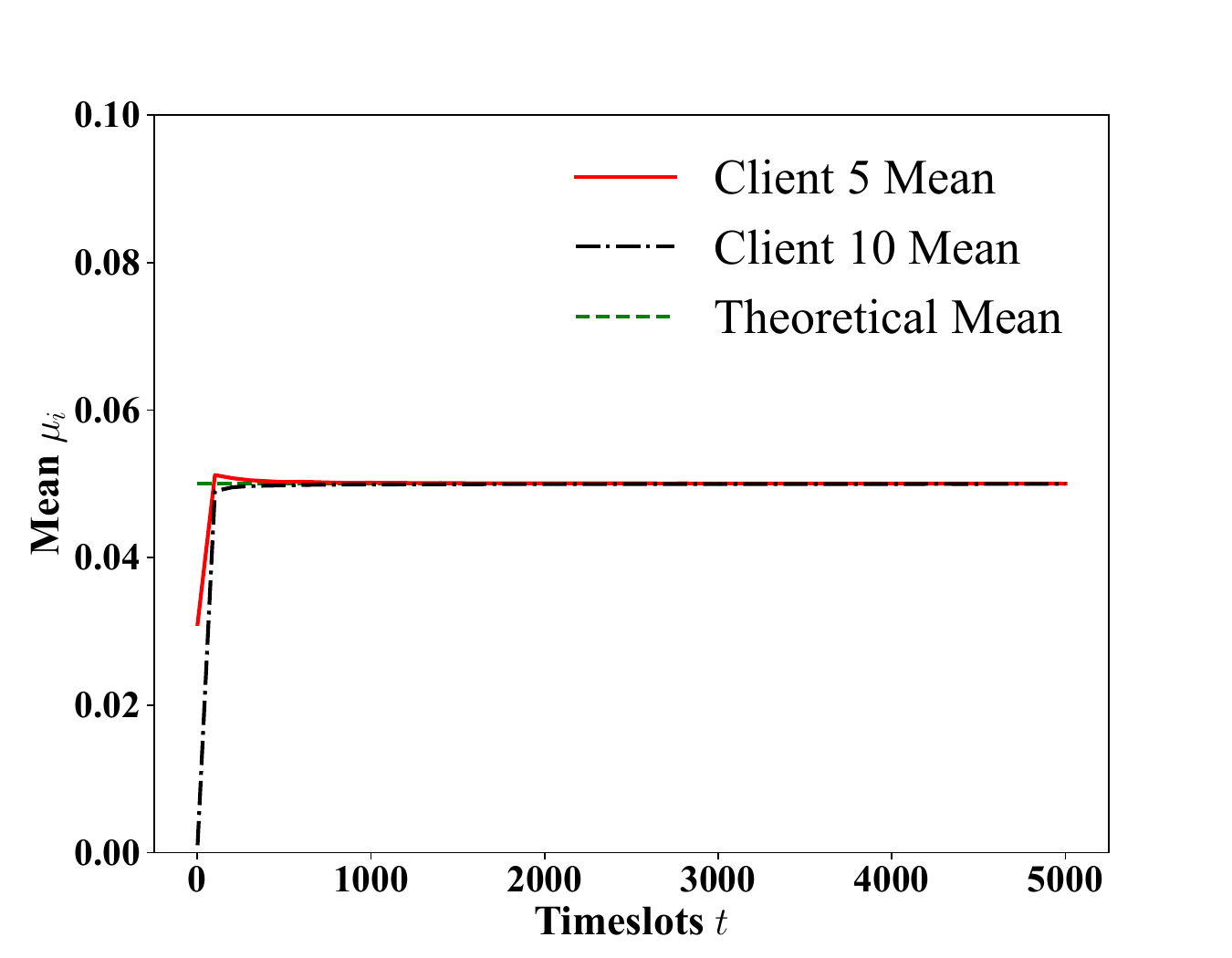}}
\end{center}
\caption{Mean convergence of two randomly selected real-time sensing clients.}
\label{fig:mean_convergence}
\end{figure*}

\begin{figure*}[hbt!]
\begin{center} 
\subfigure[N = 5 Clients.]{\includegraphics[width=0.3\textwidth, height=1.8in]{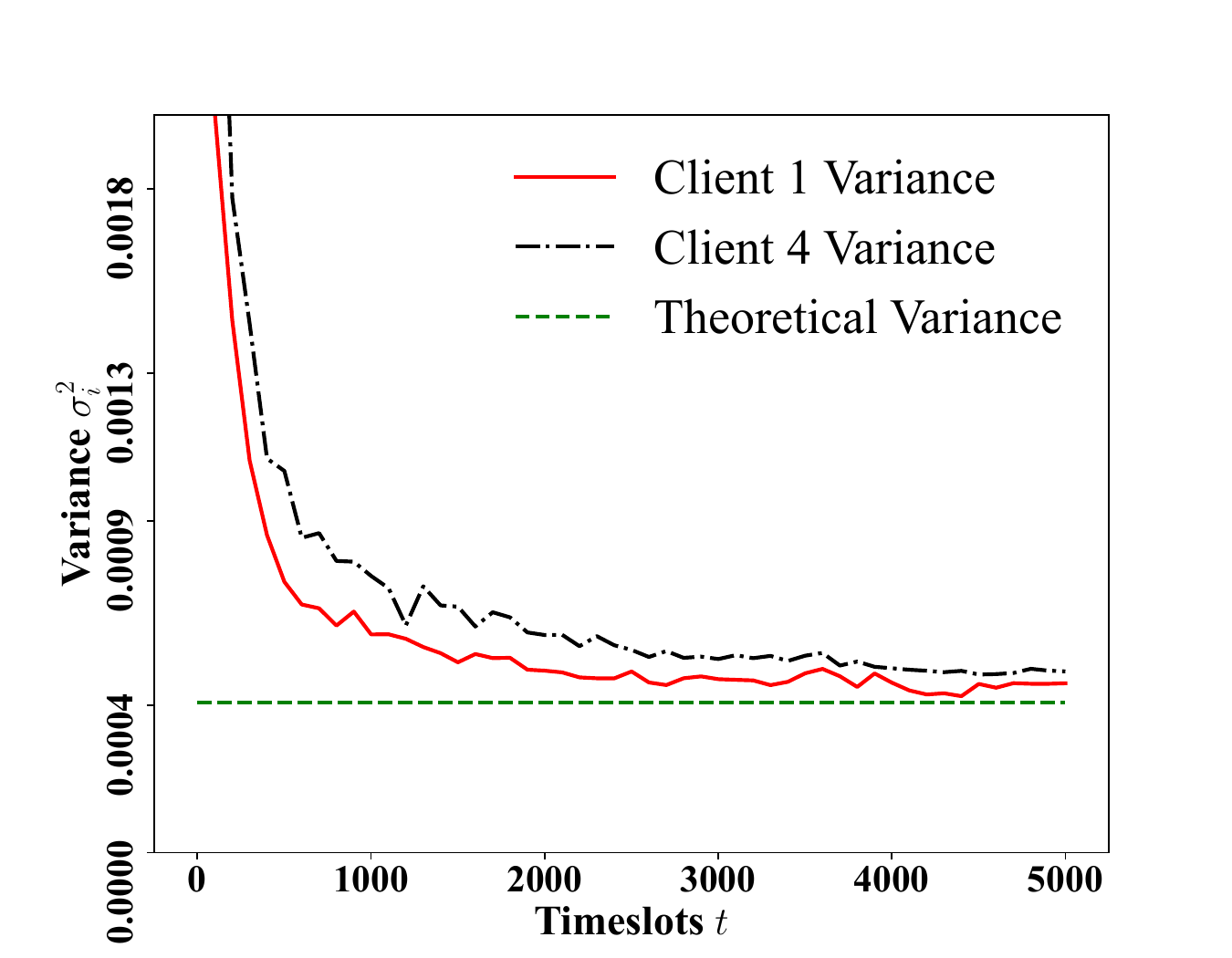}}
\subfigure[N = 10 Clients.]{\includegraphics[width=0.3\textwidth, height=1.8in]{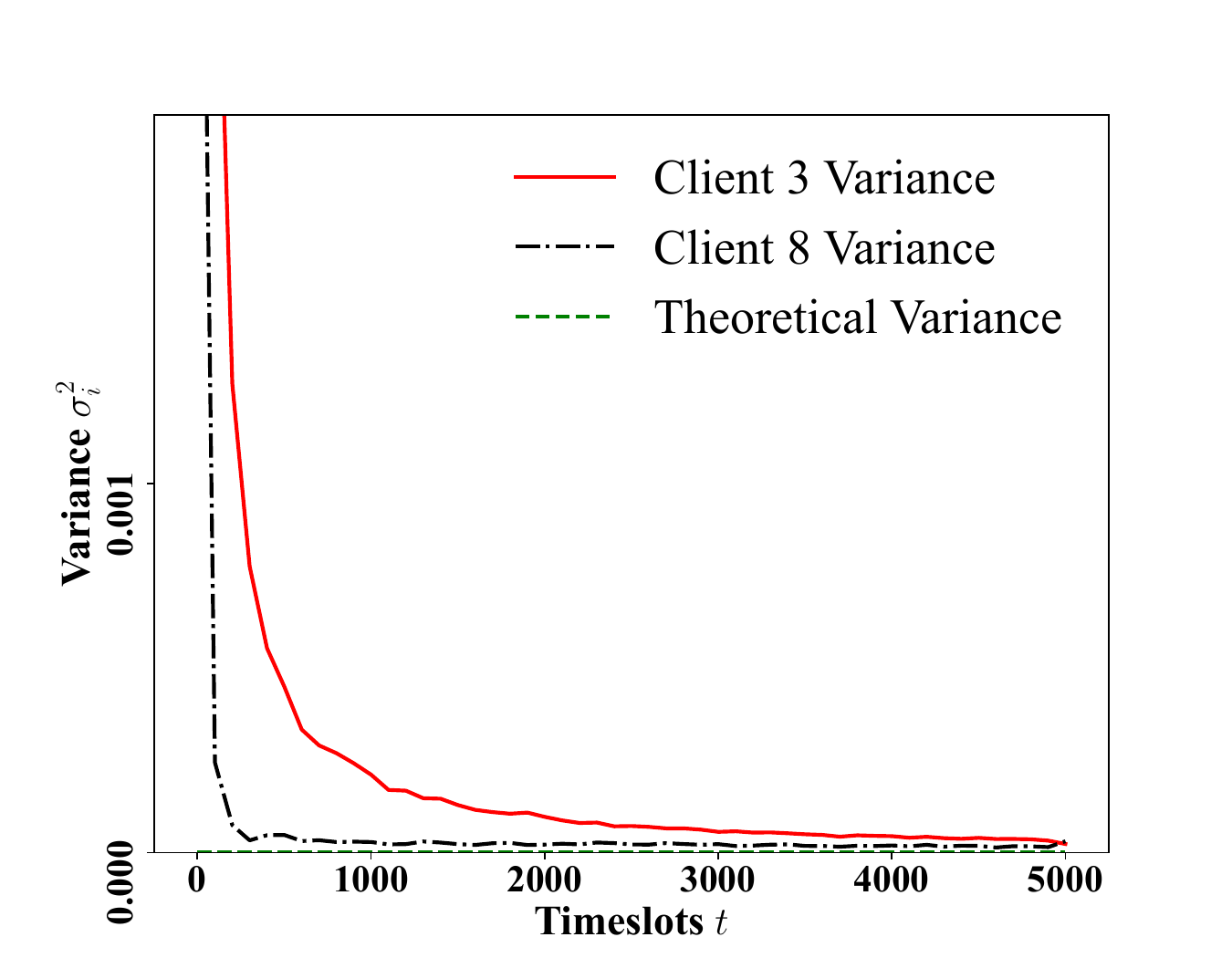}}
\subfigure[N = 20 Clients.]{\includegraphics[width=0.3\textwidth, height=1.8in]{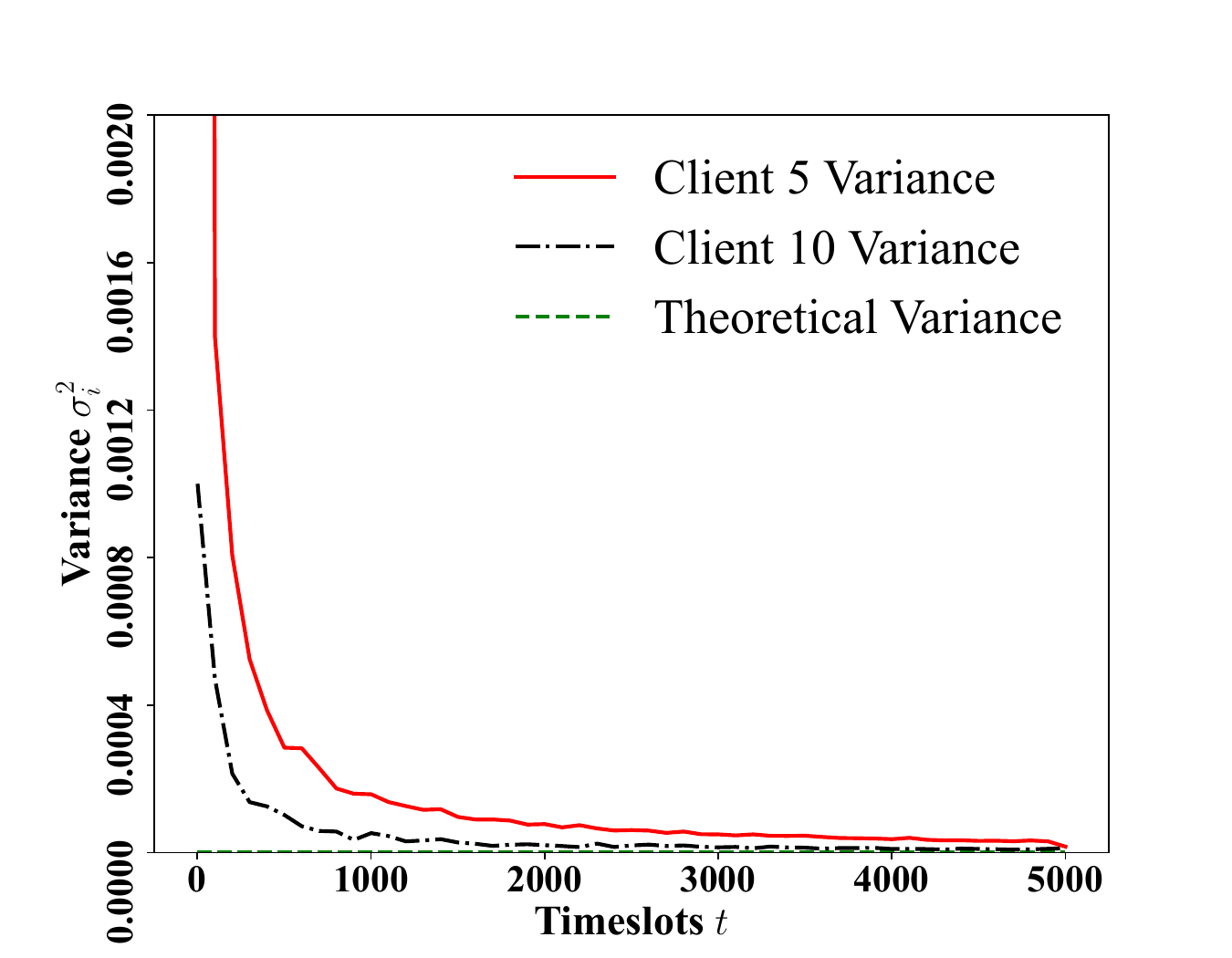}}
\end{center}
\caption{Variance convergence of two randomly selected real-time sensing clients.}
\label{fig:var_convergence}
\end{figure*}

\begin{figure*}[hbt!]
\begin{center} 
\subfigure[N = 5 Clients.]{\includegraphics[width=0.3\textwidth, height=1.8in]{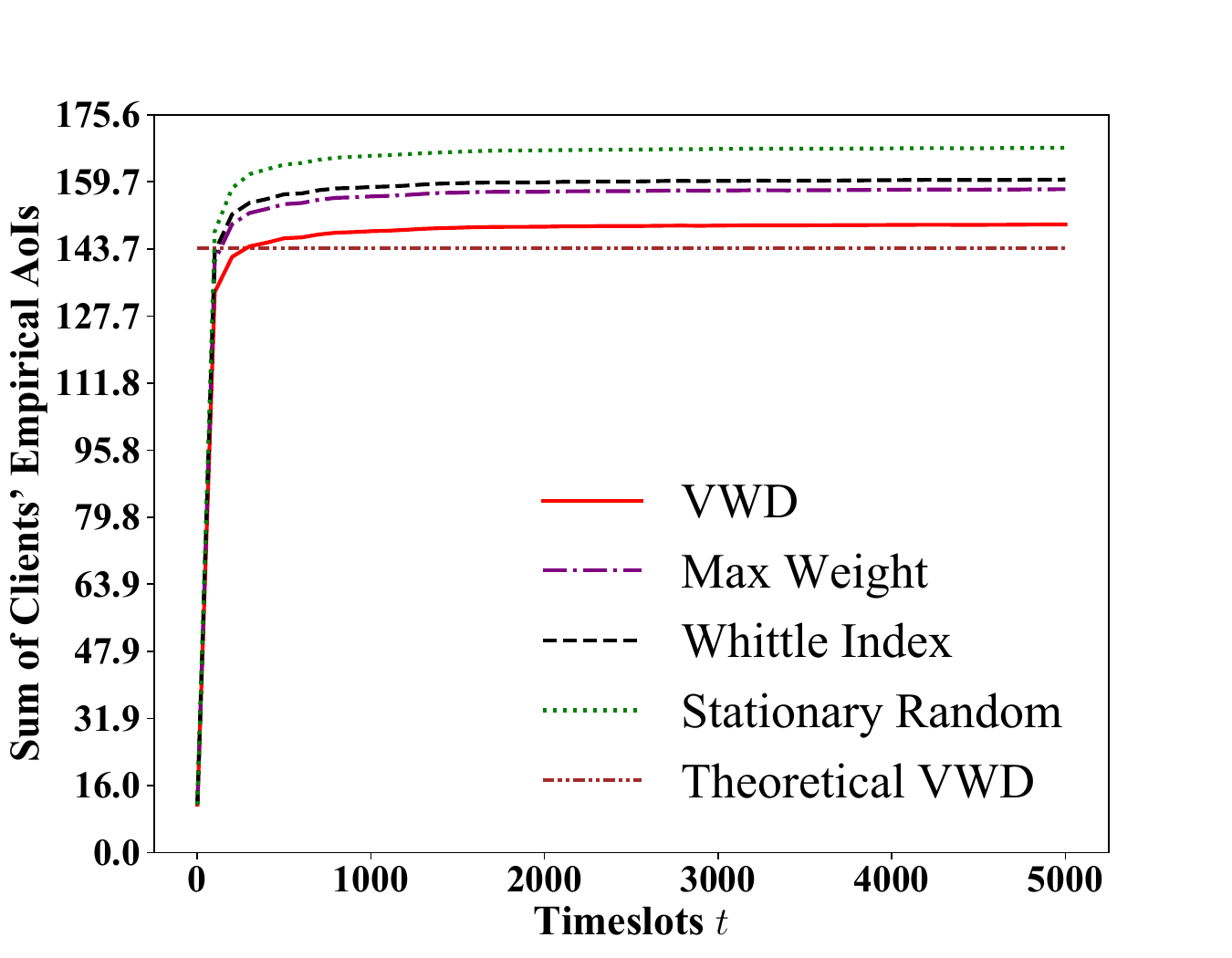}}
\subfigure[N = 10 Clients.]{\includegraphics[width=0.3\textwidth, height=1.8in]{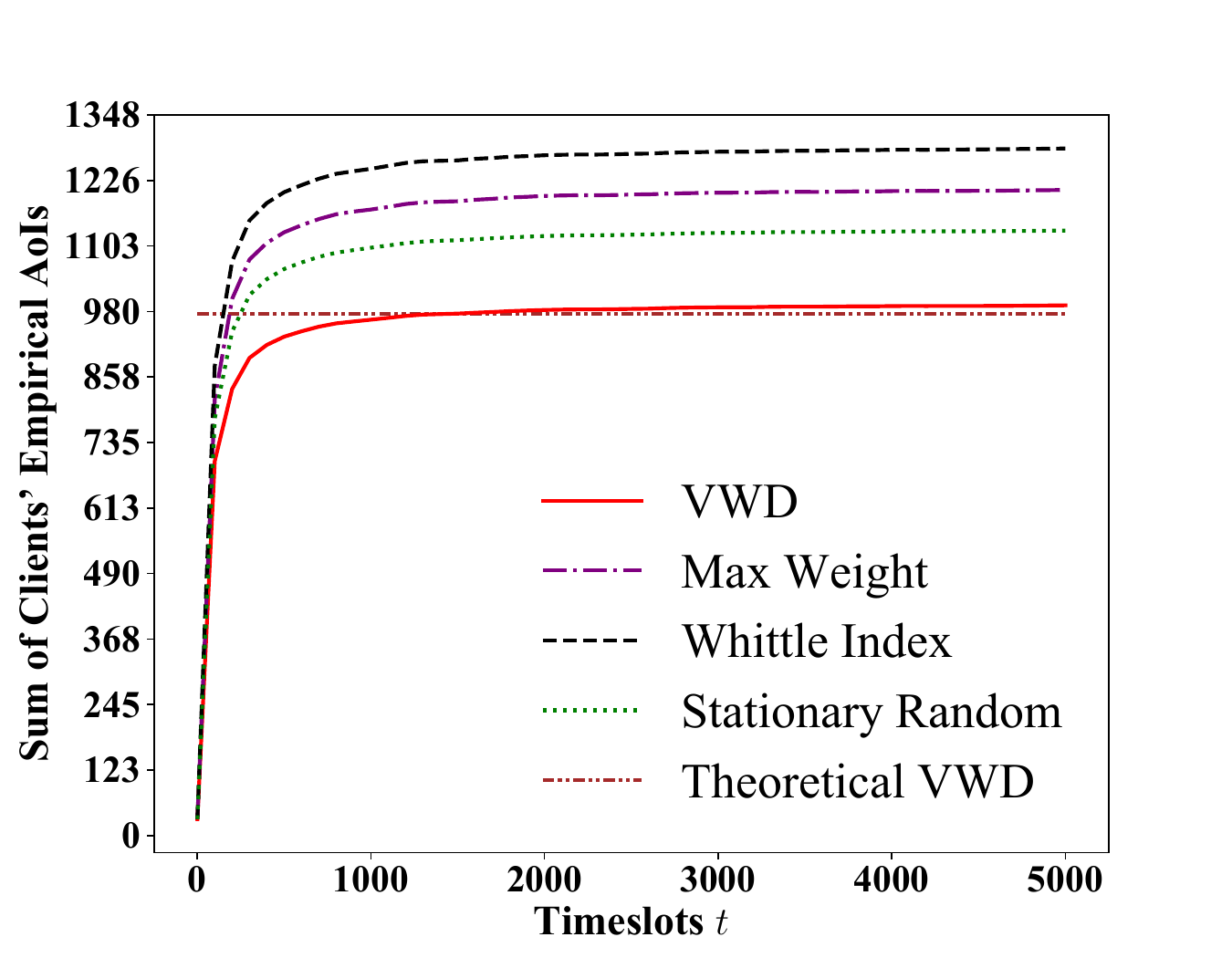}}
\subfigure[N = 20 Clients.]{\includegraphics[width=0.3\textwidth, height=1.8in]{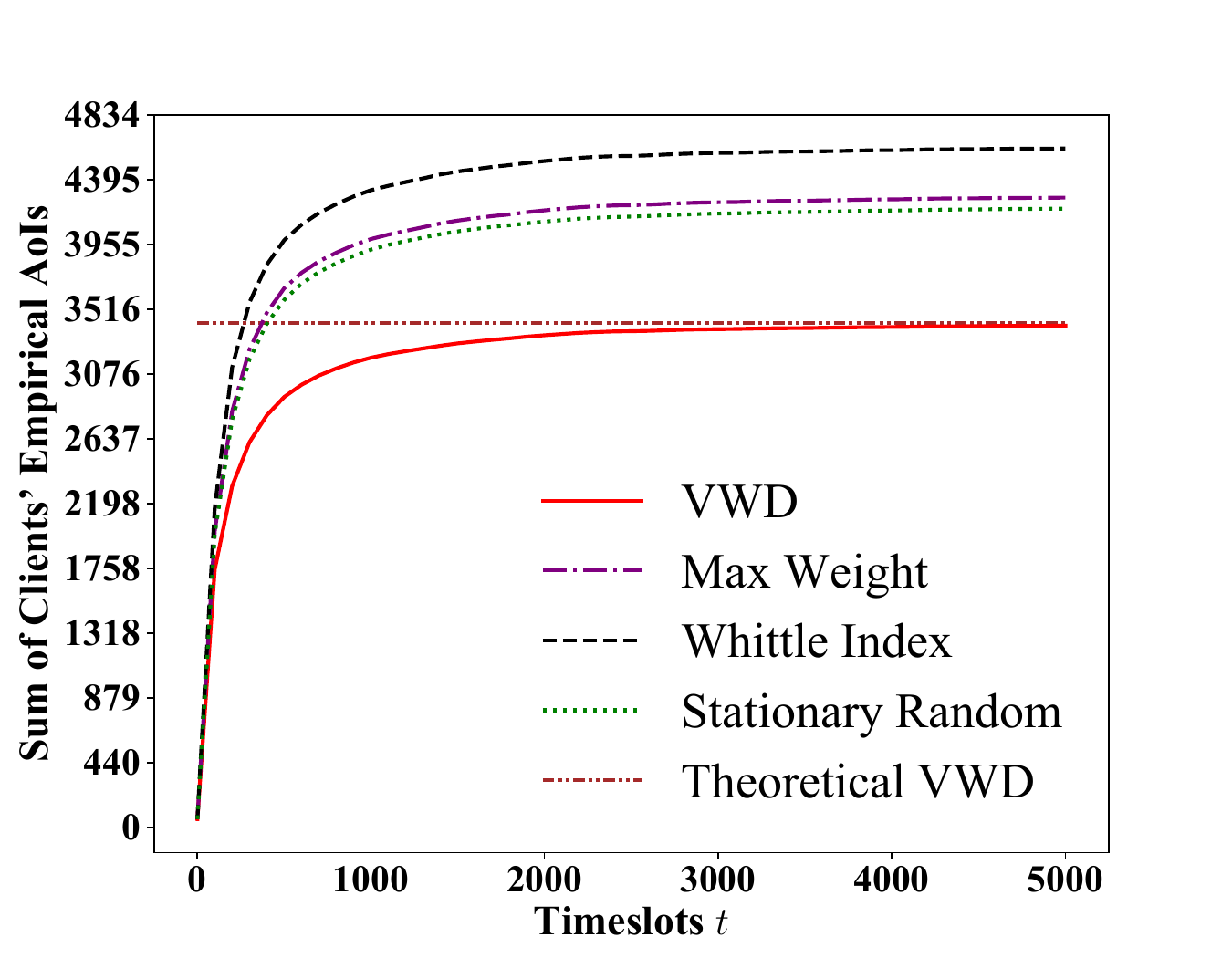}}
\end{center}
\caption{Total weighted empirical age of information (AoI) for real-time sensing clients.}
\label{fig:weighted_aoi}
\end{figure*}

We consider three different systems, each with 5, 10, and 20 real-time sensing clients, respectively. For each system, $p_n$ and $q_n$ are randomly chosen from the range $(0.05, 0.95)$, and $\{\lambda_n\}$ is randomly chosen from $(\frac{0.1}{N}, \frac{1}{N})$. After determining the values of $p_n$, $q_n$ and $\lambda_n$, we generate $1000$ independent traces of channels and packet arrivals. The performance of each policy is the average over these $1000$ independent traces. We consider both the non-weighted case, i.e., $\alpha_n\equiv 1, \forall n$, and the weighted case. In addition to the evaluated policies, we also include the numerical solutions from solving the problem \eqref{eq:objective_function}, which is referred to as the Theoretical VWD.

\vspace{0.5em}

\subsubsection{Non-Weighted Clients}

Fig.~\ref{fig:empirical_aoi_sum} shows the average total AoI for different network sizes $N = I = \{5, 10, 20\}$ when $\alpha_n = 1$. It can be observed that VWD achieves the smallest total AoI in all systems. VWD's superiority becomes more significant as $N$ increases. It can also be observed that the empirical AoI under VWD becomes virtually the same as the theoretical AoI based on the solution to (\ref{eq:objective_function}) as the number of clients $N$ increases. 

To understand why VWD performs better than the other three policies, we evaluate the total empirical variance under each policy. Specifically, let $d_n(t)$ be the total number of packet deliveries for client $n$ from time 1 to time $t$. The empirical variance of a client $n$ at time $t$ is defined as the variance of $\frac{d_n(t)}{\sqrt{t}}$ across all $1000$ independent runs. The total empirical variance is then the sum of the empirical variances of all clients. Fig.~\ref{fig:variance_policies} shows that VWD has much smaller variances than the other three policies. The ability to properly control variance enables VWD to achieve small AoIs.

We also evaluate the convergence time of VWD. For each system, we randomly select two clients and plot their empirical means, i.e., the average of $\frac{d_n(t)}{t}$ across all independent runs, and empirical variances. Since the objective is to minimize the non-weighted sum of AoIs, the optimal solution to (\ref{eq:objective_function}) has $\mu_n=\mu_u$ and $\sigma_n^2=\sigma_u^2$ for all $n\neq u$. We call the optimal $\mu_n$ and $\sigma_n^2$ obtained from solving (\ref{eq:objective_function}) the theoretical mean and the theoretical variance, respectively. The results are shown in Figs.~\ref{fig:mean_convergence} and \ref{fig:var_convergence}. It can be observed that both the empirical means and the empirical variances of clients indeed converge to their respective theoretical values. The empirical means converges to the theoretical ones very fast. On the other hand, it takes up to $355$ slots for the empirical variances to be within $0.001$ from the theoretical variances. Convergence time may be the reason why empirical AoI is larger than the theoretical one for $N=5$.

\vspace{0.5em}

\subsubsection{Weighted Clients}

We now present the results for the weighted real-time sensing clients. The weights $\alpha_1, \alpha_2,\dots$ are randomly chosen from the range $(1,5)$ and independently from each other. All other parameters are the same as in the non-weighted case. Fig.~\ref{fig:weighted_aoi} shows results for network sizes $N = I = \{5, 10, 20\}$.
VWD still outperforms other policies for all tested systems. Similar to the non-weighted case, it can be observed that the superiority of VWD becomes more significant, and the empirical VWD and the theoretical VWD performance becomes virtually the same with more clients in the system. 

\vspace{0.5em}
\subsubsection{I.I.D. Channels and Predictable Packet Generation}

\begin{figure*}[hbt!]
\begin{center} 
\subfigure[N = 5 Clients.]{\includegraphics[width=0.3\textwidth, height=1.8in]{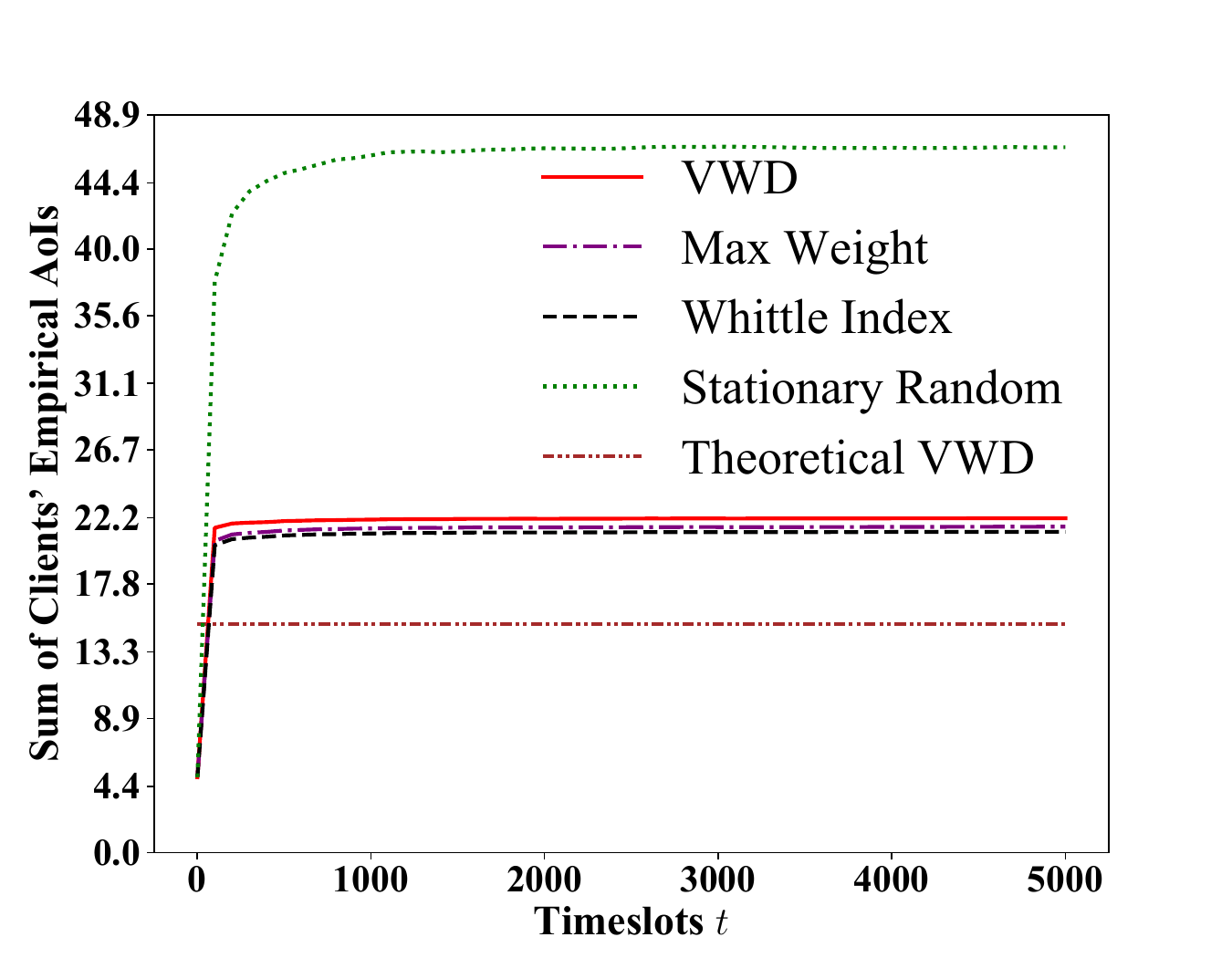}}
\subfigure[N = 10 Clients.]{\includegraphics[width=0.3\textwidth, height=1.8in]{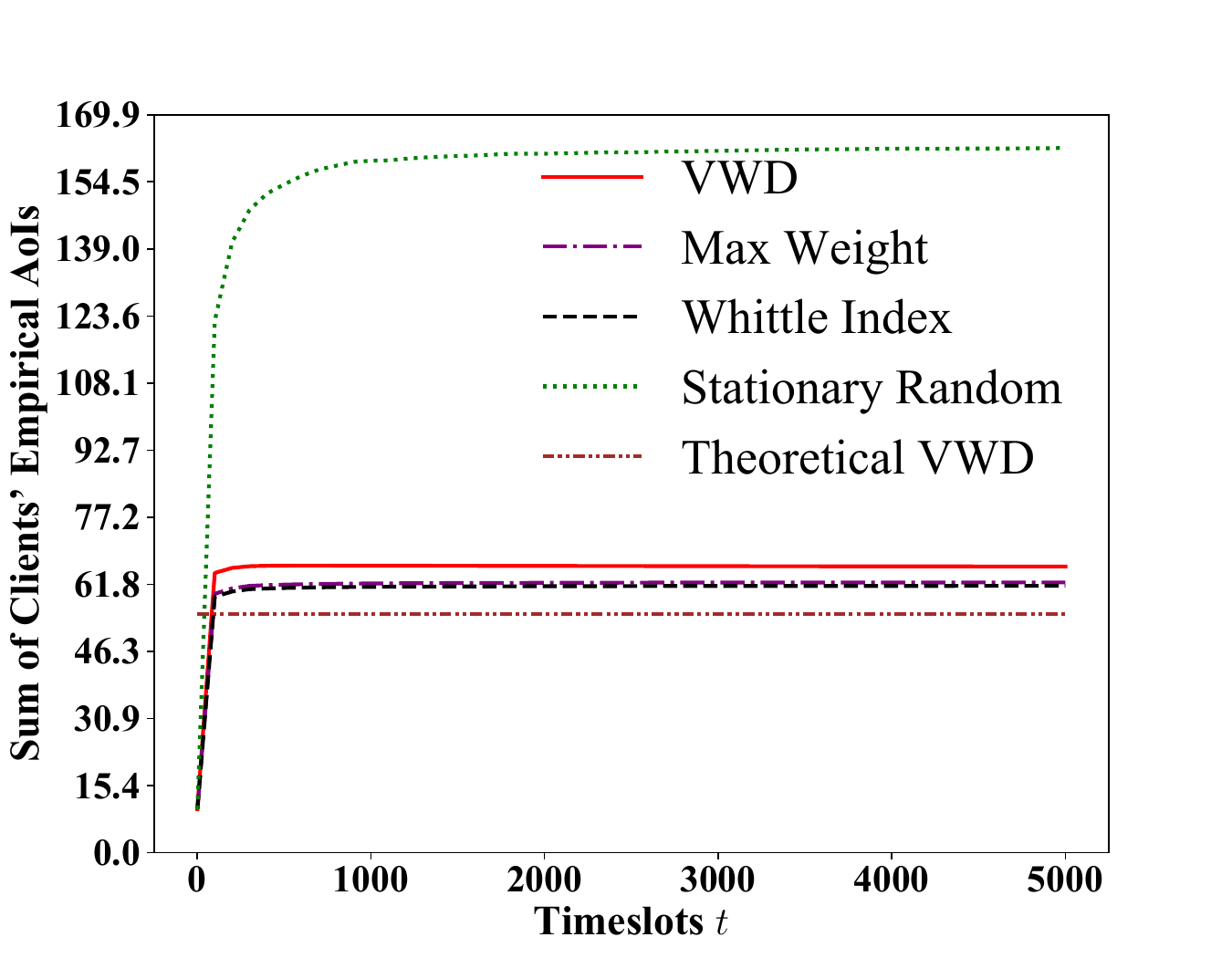}}
\subfigure[N = 20 Clients.]{\includegraphics[width=0.3\textwidth, height=1.8in]{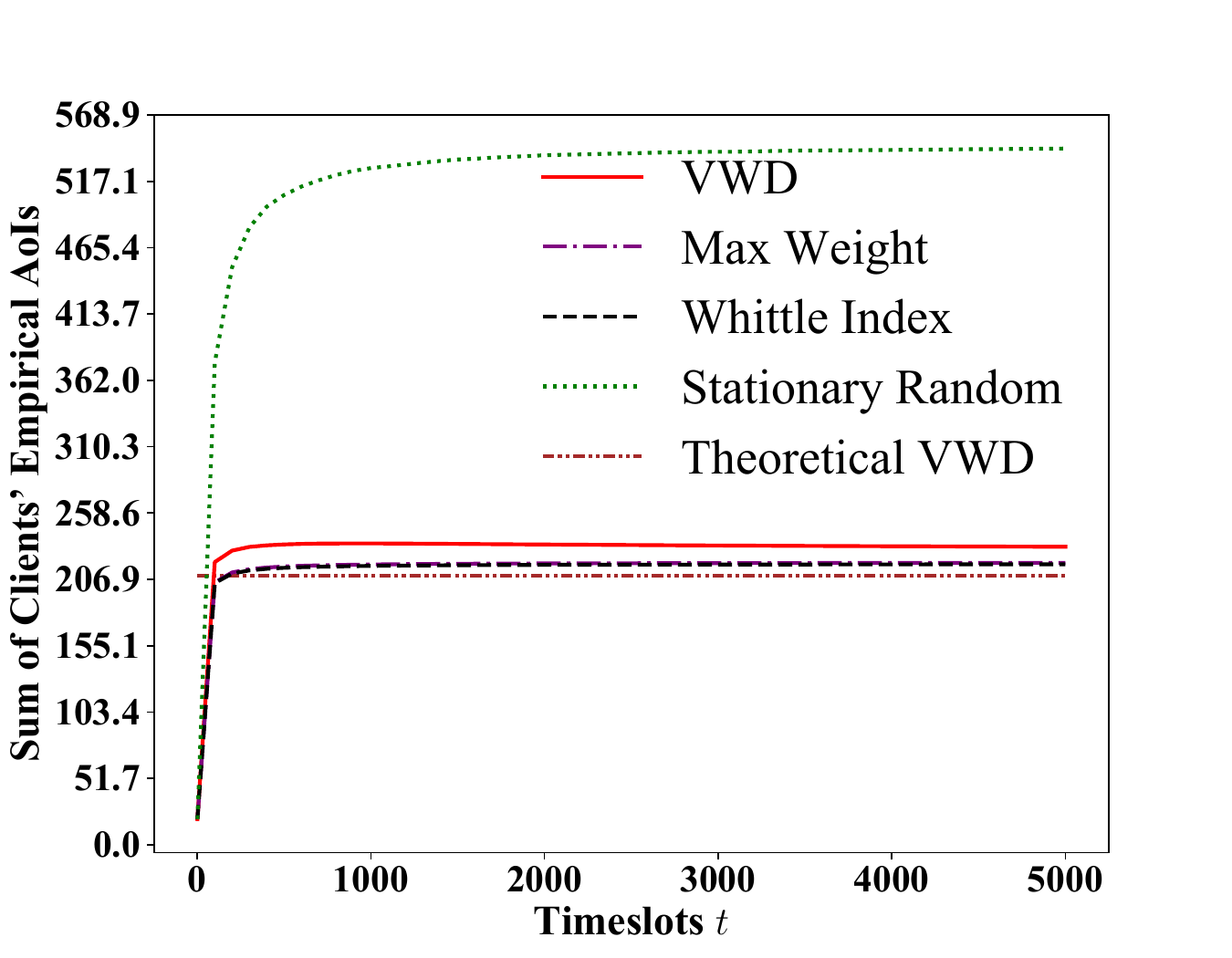}}
\end{center}
\caption{Total empirical age of information (AoI) with i.i.d. channels and predictable packet generation.}
\label{fig:iid_aoi}
\end{figure*}

{\color{blue}
As noted above, all three baseline policies, namely, Whittle index, stationary randomized, and max weight, were all developed for the special case when the channels are i.i.d. and the controller knows the packet generation times. Under our model, we can create such a scenario by making $p_n+q_n=1$ and $\lambda_n=1$, for all $n$. We have evaluated such scenarios for 5, 10, and 20 clients. We choose $p_n$ randomly from the range $[0.05, 0.95]$ and set $q_n=1-p_n$ and $\lambda_n=\alpha_n=1$. The simulation results are shown in Fig.~\ref{fig:iid_aoi}. We note that the Whittle index and max weight policies perform slightly better than VWD because they are specifically designed for this setting. Still, VWD performs very close to these two policies and the difference between VWD and these two policies is less than $6\%$ in all three systems.}

\subsection{Live Video Streaming Clients Optimization} \label{subsection:drop_min}

In this section, the objective is to maximize the timely-throughput of live video streaming clients, or equivalently minimize the outage rate for $N=J$ clients, $\sum_n \beta_n \overline{Out}_n + \gamma_n \ell_n^2$.
Each client also has a Gilbert-Elliott channel with transition probabilities $p_n$ and $q_n$.

We compare our policy, VWD, against two other policies on this problem. We first provide a description of the policies we compare against.
 \begin{itemize}
 \vspace{0.5em}
    \item \textbf{Weighted Largest Deficit (WLD):} This policy was introduced by Hsieh and Hou in \cite{hsieh20}. The policy considers clients in the ON channel state, and picks the client with the largest $(\mu_n t - \sum_{\tau=1}^{t} z(\tau)) / \ell_n$ at time slot $t$.
    \vspace{0.5em}
    \item \textbf{Delivery Based Largest Debt First (DBLDF):} Similar to the WLD policy, the DBLDF policy consider clients in the ON channel state, and schedules one client with the largest $(\mu_n t - \sum_{\tau=1}^{t} z(\tau))$ at time slot $t$.
    \vspace{0.5em}
\end{itemize}

We consider three systems, each with 5, 10, and 20 live video streaming clients, respectively. For each client, we set each of the clients' period values to be $w_n = 1/(N+1)$ to satisfy the means' constraints~\eqref{eq:sufficient:mean}--\eqref{eq:sufficient:total mean}.  In addition, we choose the clients' channel parameters $p_n,q_n$ so that the system is operating in the heavy-traffic regime. More specifically, the heavy-traffic regime has $\sum_n 1/w_n= 1-\prod_{n\in S}\frac{p_n}{p_n+q_n}$. 
After determining the $w_n, p_n$ and $q_n$ values, we obtain the $\sigma_n^2$ values from solving the problem  $\sum_n \beta_n \overline{Out}_n + \gamma_n \ell_n^2$ given the constraints~\eqref{eq:sufficient:variance} -- \eqref{eq:sufficient:non-negative}. We run the simulations for $20$ independent runs each with $10^9$ timeslots, and plot the empirical outage rates given parameters $\beta_n$ and $\gamma_n$ for each client $n$. The performance of each policy is the average of the $20$ independent runs. 
Finally, we also include the solution to the problem~\eqref{eq:objective_function} given the constraints~\eqref{eq:sufficient:mean} -- \eqref{eq:sufficient:non-negative} for temporal variance values, and refer to them as the theoretical VWD in the figures.

\begin{figure*}[t]
\begin{center} 
\subfigure[N = 5 Clients.]{\includegraphics[width=0.3\textwidth, height=1.8in]{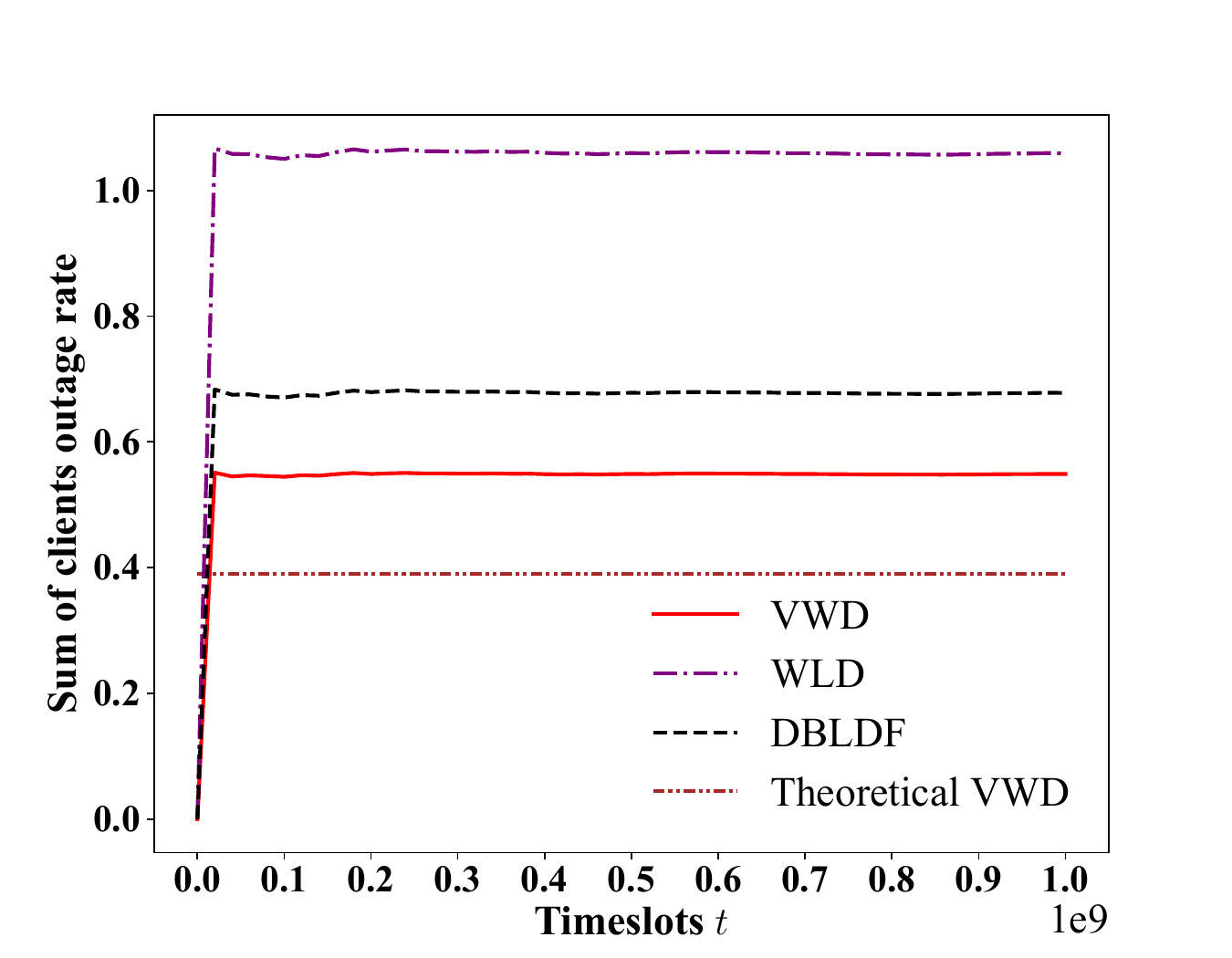}}
\subfigure[N = 10 Clients.]{\includegraphics[width=0.3\textwidth, height=1.8in]{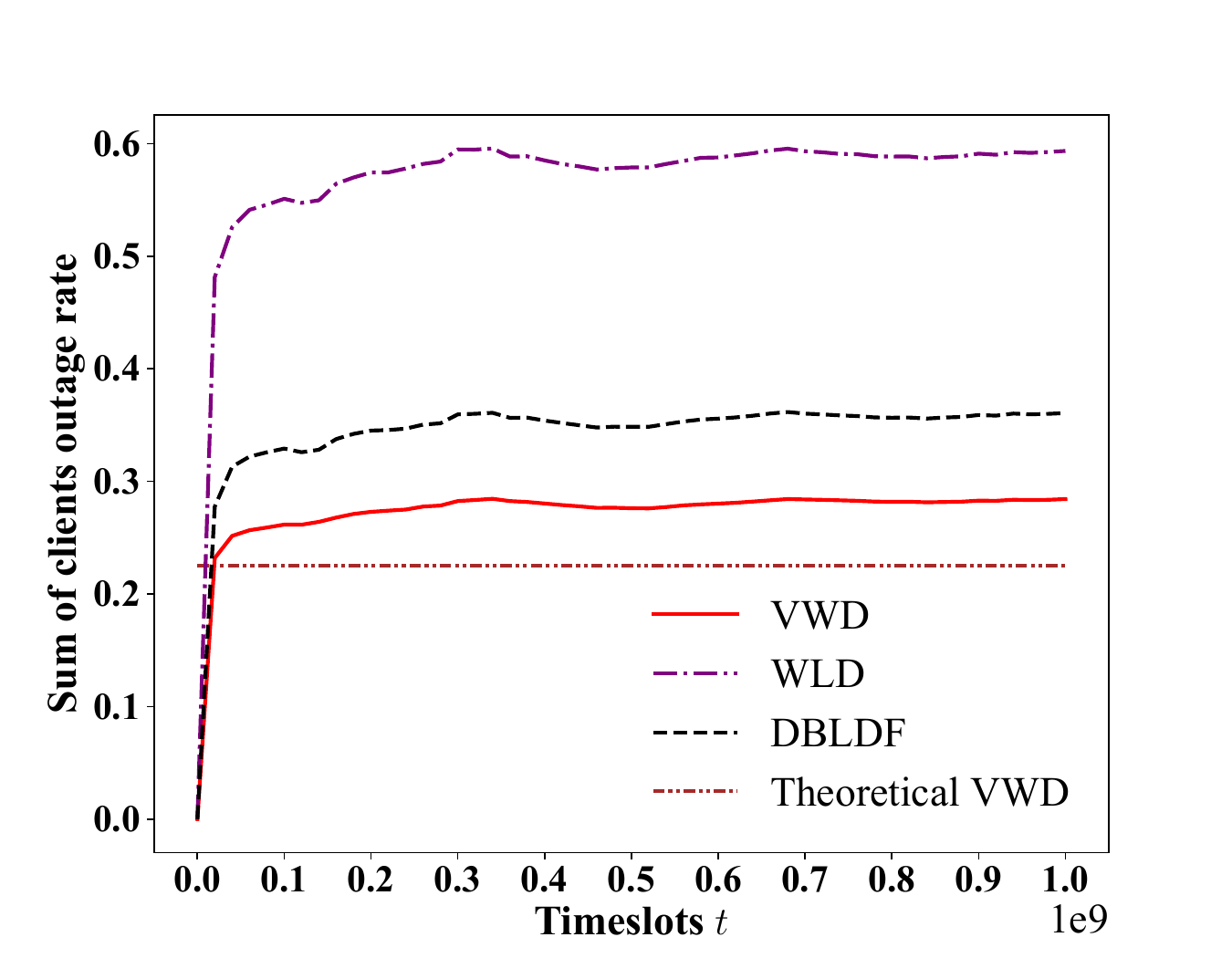}}
\subfigure[N = 20 Clients.]{\includegraphics[width=0.3\textwidth, height=1.8in]{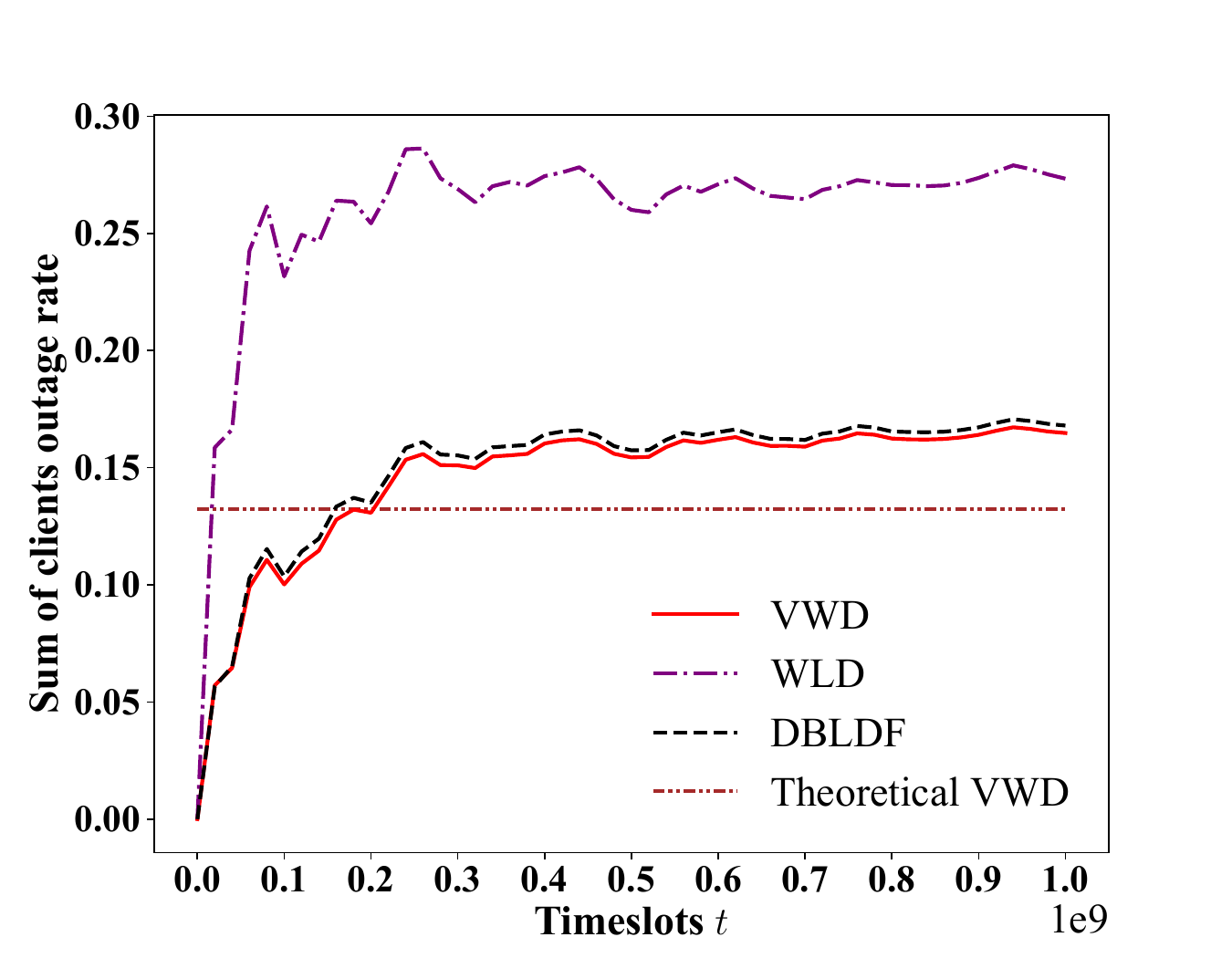}}
\end{center}
\caption{Sum of live video streaming clients' weighted outage rate with fixed delays.}
\label{fig:packet_drop_rate_results}
\end{figure*}

\vspace{0.5em}
\subsubsection{Fixed Delay Values}

We first test the policies for the case when the delay values $\ell_n$ are known in advance. The objective is to minimize the weighted outage rate $\sum_n \beta_n \overline{Out}_n$ given constraints~\eqref{eq:sufficient:mean} -- \eqref{eq:sufficient:non-negative} with $\beta_n = \ell^2_n$ and $\gamma_n = 0$.
For the 5 and 20 clients' systems, the delay value for the first client was selected as $\ell_1 = 10$, with an increment of $10$ for each subsequent client. For the $10$ clients' setup, the first client's delay value was set to $\ell_1 = 15$ with an increment of 10.

Fig.~\ref{fig:packet_drop_rate_results} shows the averaged weighted outage rate for different network size $N=\{5, 10, 20\}$.
 From the figure, it is shown that VWD has the lowest empirical system-wide outage rate of all considered policies. We observe also that VWD performance is close to the theoretical value, with the performance gap decreasing as the number of clients increases.

 \vspace{0.5em}
\subsubsection{Configurable Delay Values} 
\begin{figure*}[t]
\begin{center} 
\subfigure[N = 5 Clients.]{\includegraphics[width=0.3\textwidth, height=1.8in]{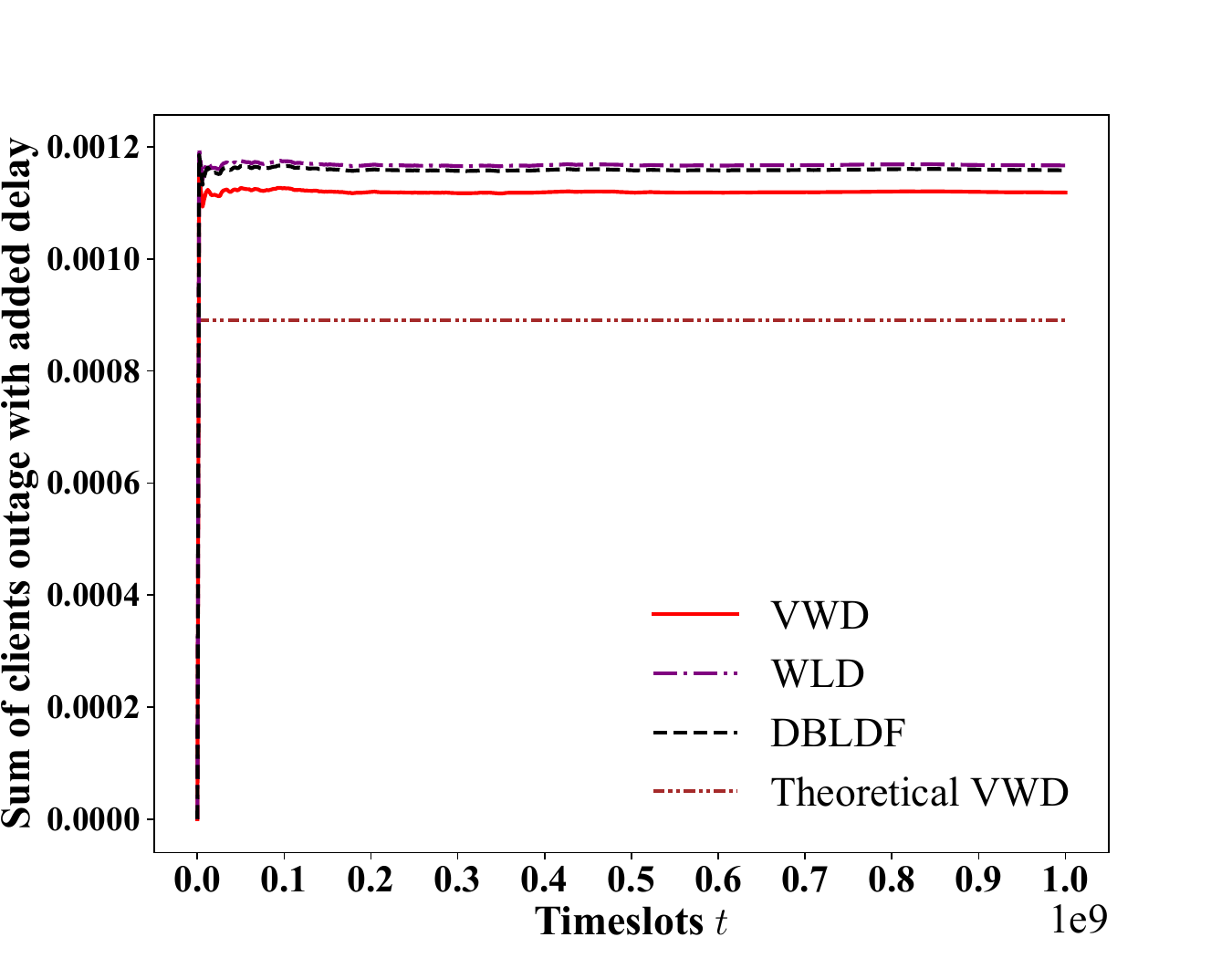}}
\subfigure[N = 10 Clients.]{\includegraphics[width=0.3\textwidth, height=1.8in]{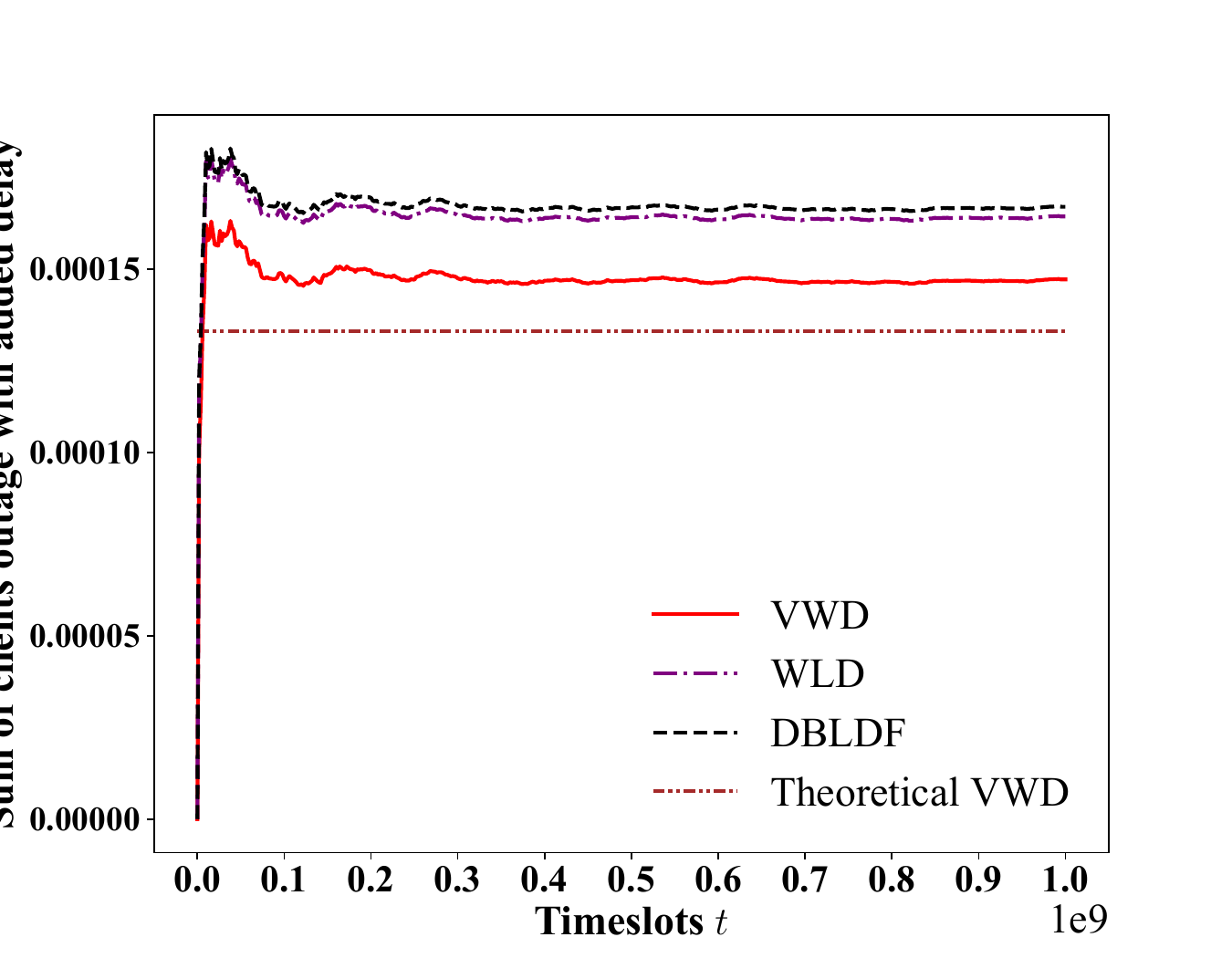}}
\subfigure[N = 20 Clients.]{\includegraphics[width=0.3\textwidth, height=1.8in]{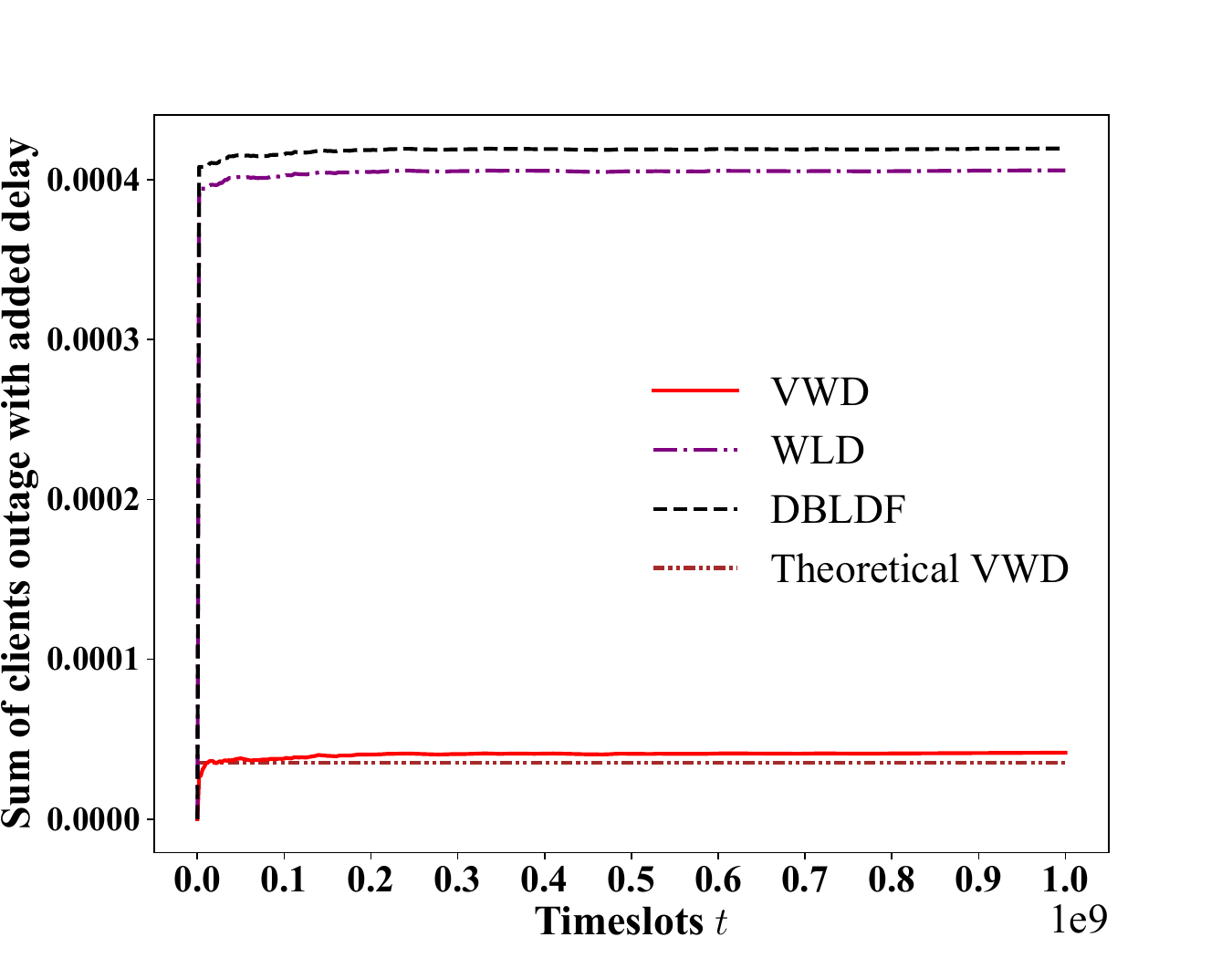}}
\end{center}
\caption{Sum of live video streaming clients' outage rate with added configurable delays.}
\label{fig:added_delay_results}
\end{figure*}

We now show the policies' performance for the case when the delay values $\ell_n$ must be set by the centralized server for each live video streaming client. 
The objective here is to minimize $\sum_n \beta_n \overline{Out}_n + \gamma_n \ell_n^2$ given constraints~\eqref{eq:sufficient:variance} -- \eqref{eq:sufficient:non-negative}.
Here, the parameters are chosen as $\beta_n = 1 \forall n$, and $\gamma_n$ is chosen from the range $[10^{-7}, 10^{-13}]$ for each client $n$.
After solving the network optimization problem and obtaining the temporal variance and delay values, we sum the obtained configurable delay values and refer to it as $\ell_{tot}$. In order to ensure we have approximately the same total delay for other policies, we allocate delay values using $\ell_{tot}$.
For WLD, we solve the problem ~\eqref{eq:objective_function} with delay values picked according to $\frac{\sigma_n}{ \sigma_{tot}} \cdot \ell_{tot}$ for client $n$.
For DBLDF, we set the delay values per client as $\ell_{tot}/N$. 

We provide the averaged results in Fig.~\ref{fig:added_delay_results} for all three systems with 5,10, and 20 clients. It can be seen that VWD performs the best in terms of system-wide average outage rate compared to other policies. Additionally, VWD performance gap between the empirical and theoretical VWD values decreases with a more clients in the system. The other baselines, WLD and DBLDF, have a lower perforamnce for all three systems.

\subsection{System with Both Kinds of Clients} \label{subsec:aoi_outage_results}

\begin{figure*}[t]
\begin{center} 
\subfigure[N = 6 Clients.]{\includegraphics[width=0.3\textwidth, height=1.8in]{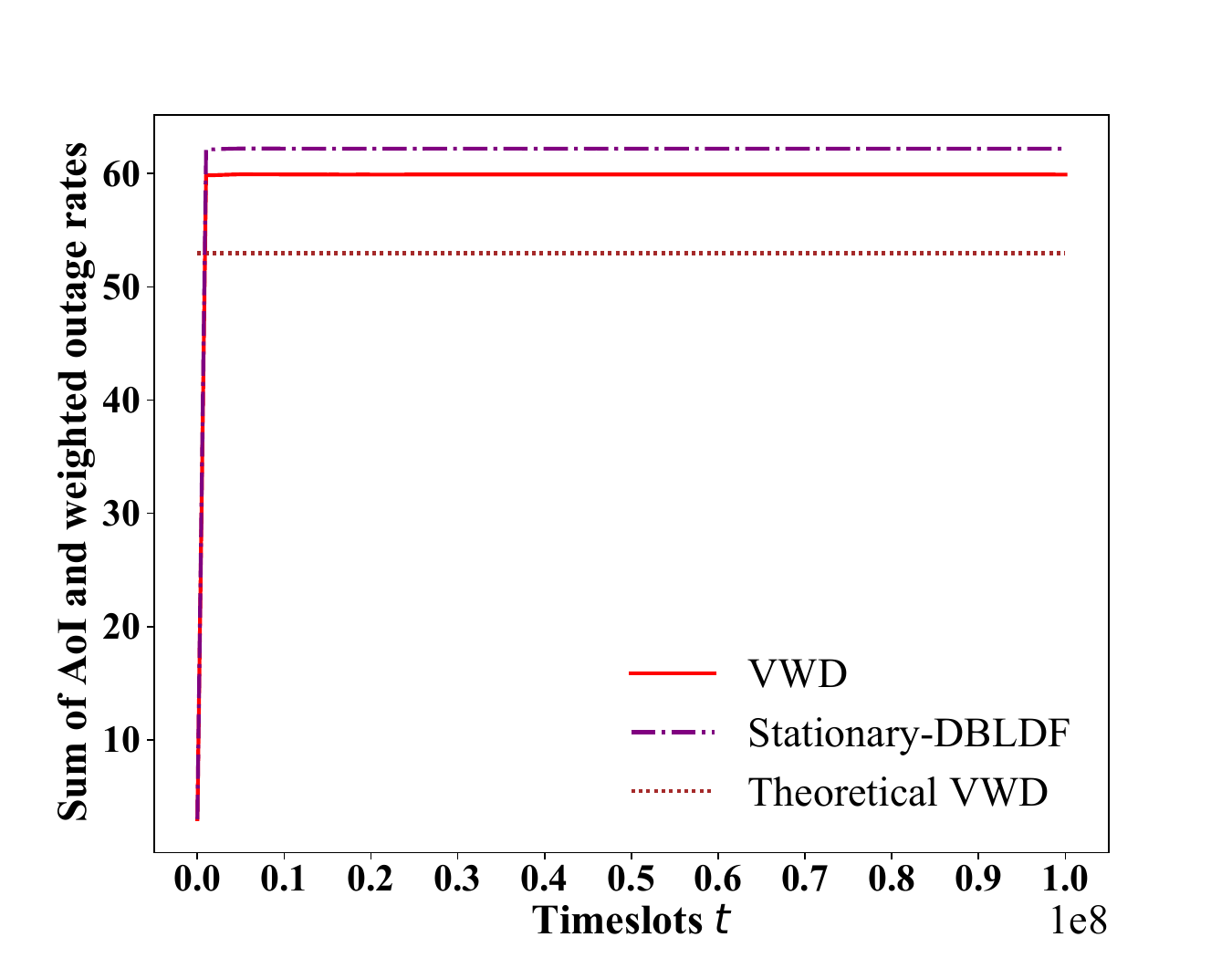}}
\subfigure[N = 10 Clients.]{\includegraphics[width=0.3\textwidth, height=1.8in]{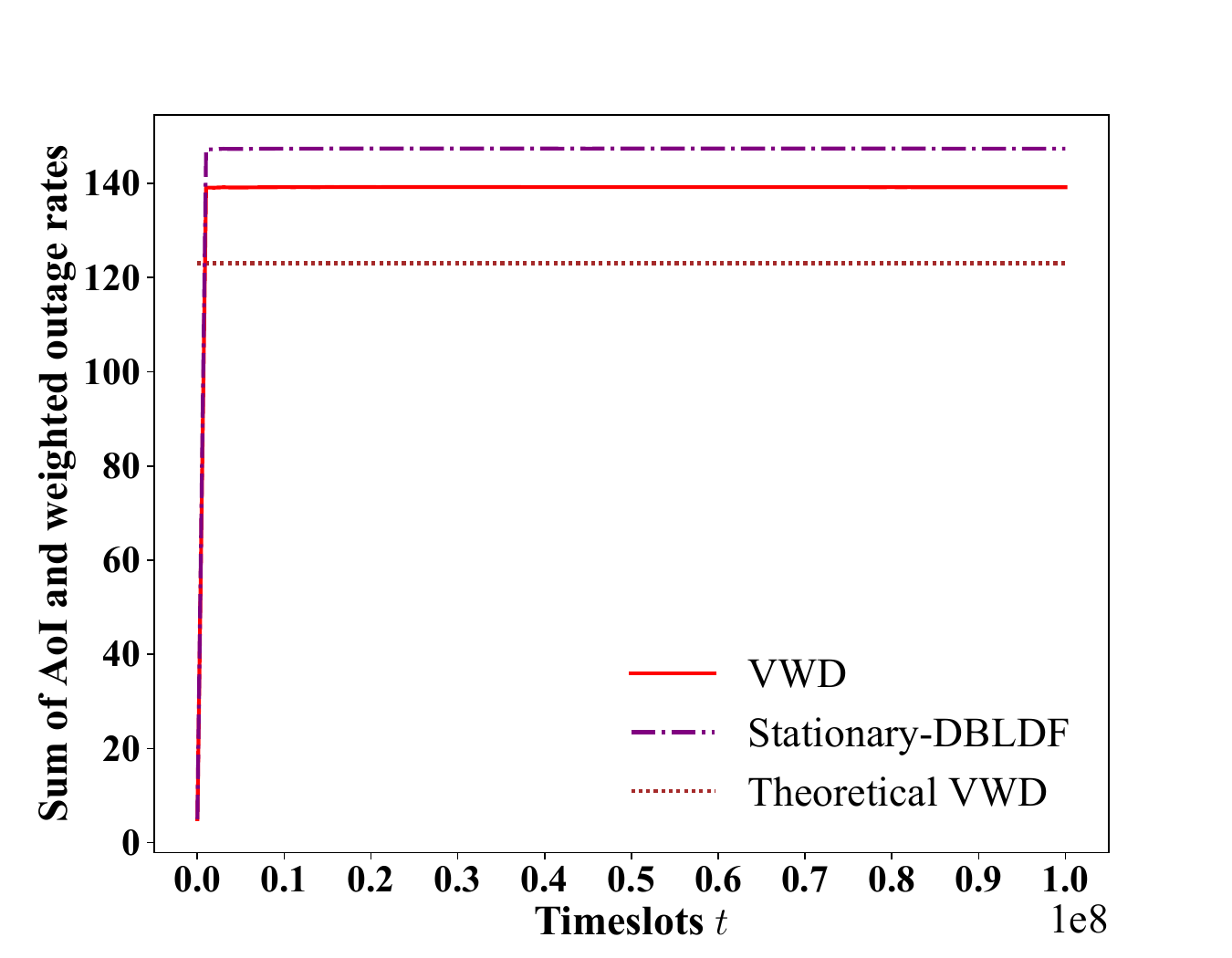}}
\subfigure[N = 20 Clients.]{\includegraphics[width=0.3\textwidth, height=1.8in]{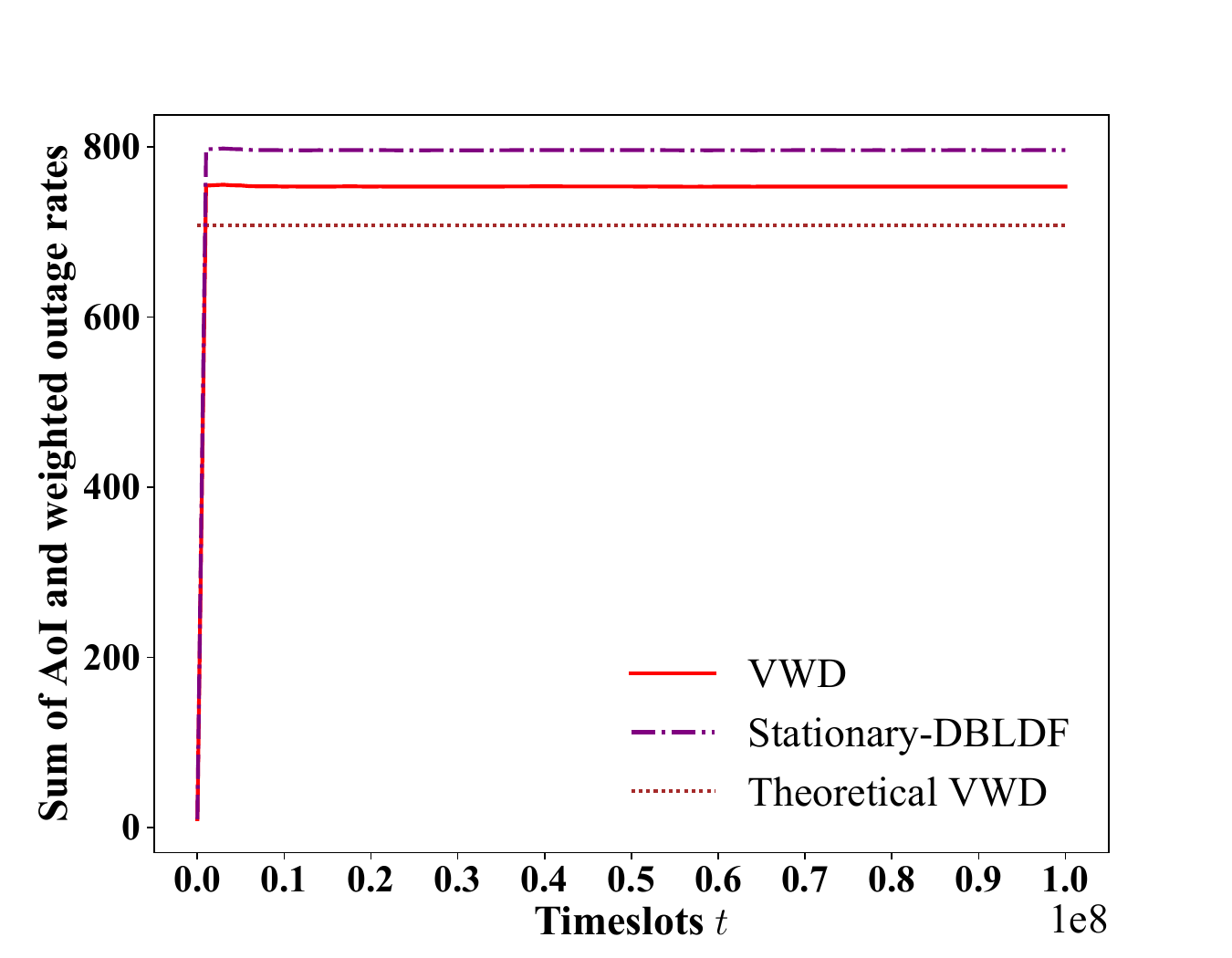}}
\end{center}
\caption{Sum of real-time sensing clients' AoI and live video streaming clients' weighted outage rate.}
\label{Figure:regime_1}
\end{figure*}

For the last case, we evaluate our policy, VWD, on the unsolved optimization problem where we minimize $I$ real-time sensing clients' $\overline{AoI}_n$, and $J$ live video streaming client maximize their timely-throughput or equivalently, minimize their outage rate. The network objective function is then to minimize the sum of real-time sensing AoIs and live video streaming outage rates $\big( \sum_{n=1}^{I} \overline{AoI}_n + \sum_{n=I+1}^{I+J} \beta_n \overline{Out}_n + \gamma_n \ell_n^2\big)$.

We emphasize that optimizing different clients' objective functions was not studied in prior works. For this reason, we can only evaluate VWD against a policy we designed and call stationary-DBLDF. 
The stationary-DBLDF policy combines both the stationary random from the real-time sensing clients' case, and the DBLDF policy from the live video streaming clients' case. 
We combine these two policies together since they were the best-performing baselines in the real-time sensing and live video streaming cases.
In odd timeslots, the stationary-DBLDF policy picks an ON client according to the stationary random policy described in Section~\ref{subsection:aoi_min}. In even timeslots, the policy picks a client in the ON channel state according to the DBLDF policy described in Section~\ref{subsection:drop_min}.

We consider three systems with $6, 10, $ and $20$ clients. In each system, the first half of clients are real-time sensing clients, while the second half are live video streaming clients, i.e. $I = J = N/2$.
The $I$ real-time sensing clients $\{\lambda_n\}$ values were randomly chosen from the range $(\frac{0.01}{N}, \frac{0.1}{N})$. For the $J$ live video streaming clients, their respective packet arrival rate was set to $w_n = 1/(N+1)$. For the real-time sensing clients, we solve for the mean values $\mu_n$ such that the sum of client means is equal to the channel mean $\sum_{n=1}^I \mu_n - \sum_{n=I+1}^{I+J} w_n = m_s$. 
We run the simulations for $20$ independent runs each with $10^8$ timeslots, and plot the sum of empirical AoI and outage rates for the three considered network systems. The plotted performance of the two policies, VWD and stationary-random, is the average of $20$ independent runs.

We test the two policies, VWD and stationary-random, with $\alpha_n = 1$ for all real-time sensing clients $n=1,2,\hdots,I$.
For live video streaming clients, we set $\beta_n = \ell_n^2$ and $\gamma_n = 0$ for all $n = I+1, I+2,\hdots, I+J$.
For the $6$ and $20$ clients' systems, the first live video streaming delay values $\ell_{n=I+1}$ was set to $10$ with an increment of $10$. For the $10$ clients' system, the first live video streaming client's delay $\ell_{n=I+1}$ was set to 15 with an increment of $10$ for each subsequent live video streaming client.

The averaged results are shown in Fig.~\ref{Figure:regime_1} along with the theoretical VWD value obtained by plugging in mean and temporal variance values into Eq.~\eqref{eq:objective_function}. In this setting, VWD outperforms the stationary-DBLDF policy in all three clients' systems. Moreover, VWD's empirical performance gap compared to the stationary-random policy also increases as the number of clients $N$ increases.

\section{Related Works} \label{sec:related}
There have been many works on scheduling in wireless networks for minimizing AoI. In~\cite{tripathi17}, Tripathi and Moharir schedule over multiple orthogonal channels and propose Max-Age Matching and Iterative Max-Age Scheduling, which they show to be asymptotically optimal.
Hsu, Modiano and Duan~\cite{Hsu17} studied the problem of scheduling updates for multiple clients where the updates arrive i.i.d. Bernoulli, and formulate the Markov decision process (MDP) and prove structural results and finite-state approximations. In~\cite{Hsu18}, Hsu follows up this work by showing that a Whittle index policy can achieve near optimal performance with much lower complexity.
Sun et al.~\cite{Sun18} studied scheduling for multiple flows over multiple servers, and show that maximum age first (MAF)-type policies are nearly optimal for i.i.d. servers.
In~\cite{Talak20}, Talak, Karaman and Modiano study scheduling a set of links in a wireless network under general interference constraints.
The optimization of AoI and timely-throughput were studied in~\cite{Kadota18a,Lu18}. All of these works assume i.i.d. channels while our work focuses on different Gilbert-Elliot channels per client.

There have been a limited number of works on Markov channel and source models related to AoI. In the recent work~\cite{pan20}, Pan et al. study scheduling a single source and choosing between a Gilbert-Elliott channel and a deterministic lower rate channel.
Buyukates and Ulukus~\cite{buyukates20} study the age-optimal policy for a system where the server is a Gilbert-Elliott model and one where the sampler follows a Gilbert-Elliott model.
In~\cite{nguyen19}, Nguyen et al. analyze the Peak Age of Information (PAoI) of a two-state Markov channel with differing cases of channel state information (CSI) knowledge.
Kam et al.~\cite{kam18} study the remote estimation of a Markov source, and they propose effective age metrics that capture the estimation error. Our work differs in that we focus on scheduling for multiple clients from a single AP over parallel non-i.i.d. channels. 
There have been some recent efforts on studying short-term performance through Brownian motion approximation~\cite{hsieh16,hou17,hsieh20,guo21}. However, each of these listed works is limited to a specific channel model and a specific application. 

Recent works proposed scheduling policies for maximizing timely-throughput over wireless networks. Hou et al. proposed a model for wireless networks with strict per-packet deadlines and studied the capacity of wireless networks to support timely-throughputs \cite{hou2009theory}.
Kim et al. \cite{kim2014} studied the multicast scheduling problem for traffic with hard deadlines over unreliable wireless channels and provided a greedy policy that maximizes immediate weighted sum throughput per timeslot.
Lu et al. \cite{srikant2016} considered ad hoc wireless networks with real-time traffic with hard scheduling deadlines, and demonstrated that their algorithm achieves the QoS requirements.
Talak and Modiano \cite{modiano2019} considered a single-server queuing system serving one packet and studied the tradeoff between AoI and packet delay rates, and showed that as the AoI approaches its minimum, the packet delay and its variance approach infinity. 
Additionally, Liu et al. \cite{srikant2021} studied the scheduling problem of scheduling in wireless networks under both packet deadline and power constraints.
Sun et al. \cite{sun2021} also considered minimizing weighted average AoI subject to timely throughput constraints and provided two scheduling policies that are close to the theoretical AoI bound.
All of these works considered the same delay value for all packets generated within a fixed time frame.


Another line of work studied timely-throughput for wireless networks with delays assigned per packet. 
 Singh and Kumar \cite{singh2015} described decentralized policies for maximizing throughput with deadline constraints by solving a Lagrangian dual of a Markov Decision Process (MDP).
Chen and Huang \cite{chen2018} studied a stochastic single-server multi-user system, and identified the timely-throughput improvement from predictive scheduling. Singh and Kumar \cite{singh2019} proposed an optimal scheduling policy for multihop wireless network serving multiple flows with hard packet deadlines.
Singh and Kumar \cite{singh2021} also studied the problem of maximizing the throughput of packets with hard end-to-end deadlines in multihop wireless networks. The studied model considered stochastic links and provided a decentralized network controller to maximize timely throughput of the network. The solution was obtained from the decomposition of the Lagrangian of a constrained MDP.
All the previous works solves a per-packet decomposition of an MDP that captures the system's evolution.

\section{Conclusion} \label{sec:conclusion}

In this paper, we presented a theoretical second-order framework for wireless network optimization. This framework captures the behaviors of all random processes by their second-order models, namely, their means and temporal variances. We analytically established a simple expression of the second-order capacity region of wireless networks. A new scheduling policy, VWD, was proposed and proved to achieve every interior point of the second-order capacity region. 
The framework utility is demonstrated by applying it to the problems of real-time sensing optimization, live video streaming optimization, and the problem of jointly optimizing a system with real-time sensing clients and live video streaming clients over Gilbert-Elliott channels. 

We derived closed-form expressions of second-order models for both Gilbert-Elliott channels, AoIs, and timely-throughput. Moreover, we formulated the objective function as an optimization problem over the means and temporal variances of delivery processes. The solution of this optimization problem can then be used as parameters for VWD. Simulation results show that VWD achieves better system-wide performance compared to the baselines in all cases. This result is significant when one considers that the baselines are limited to only minimizing AoI or maximizing timely throughput, while our general-purpose second-order policy achieves a better performance in all settings.

\bibliographystyle{IEEEtran}
\bibliography{reference}

\begin{IEEEbiography}[{\includegraphics[width=1in,height=1.25in,clip,keepaspectratio]{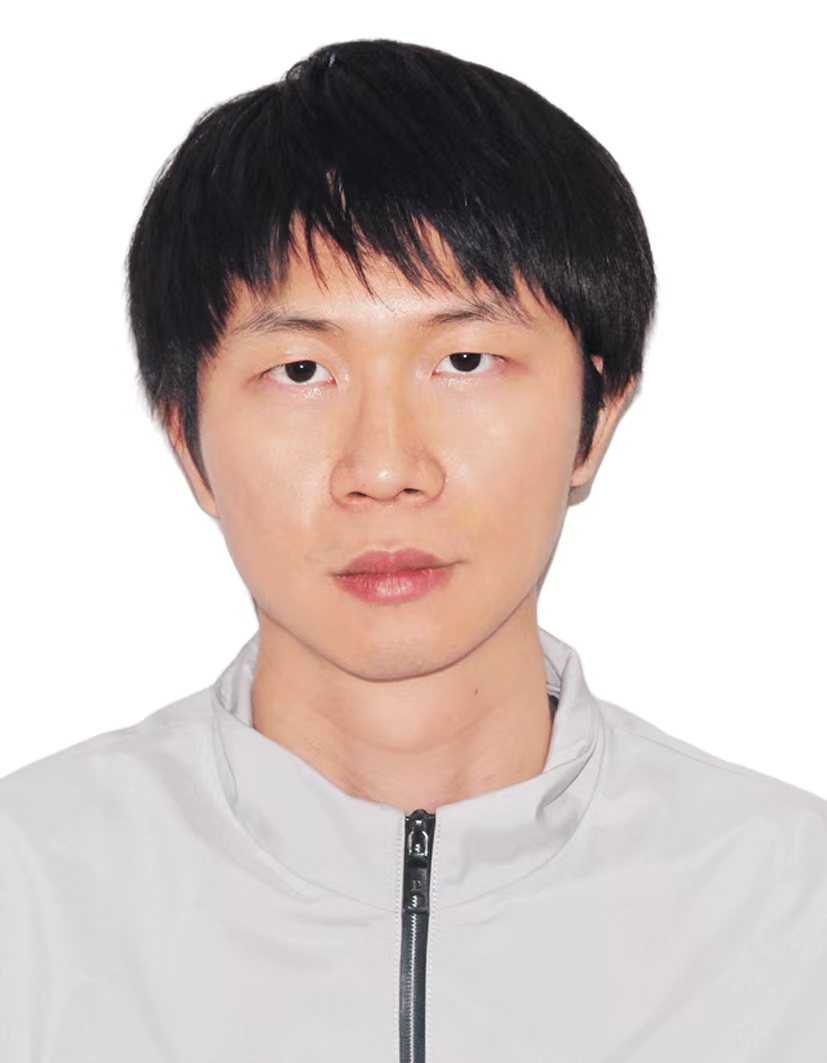}}]{Daojing Guo}
 received his Ph.D. from the Electrical and Computer Department of Texas A\&M University. His research interests include wireless networks, edge/cloud computing, and Vehicle-to-X. 
\end{IEEEbiography}

\vspace{11pt}

\begin{IEEEbiography}[{\includegraphics[width=1in,height=1.25in,clip,keepaspectratio]{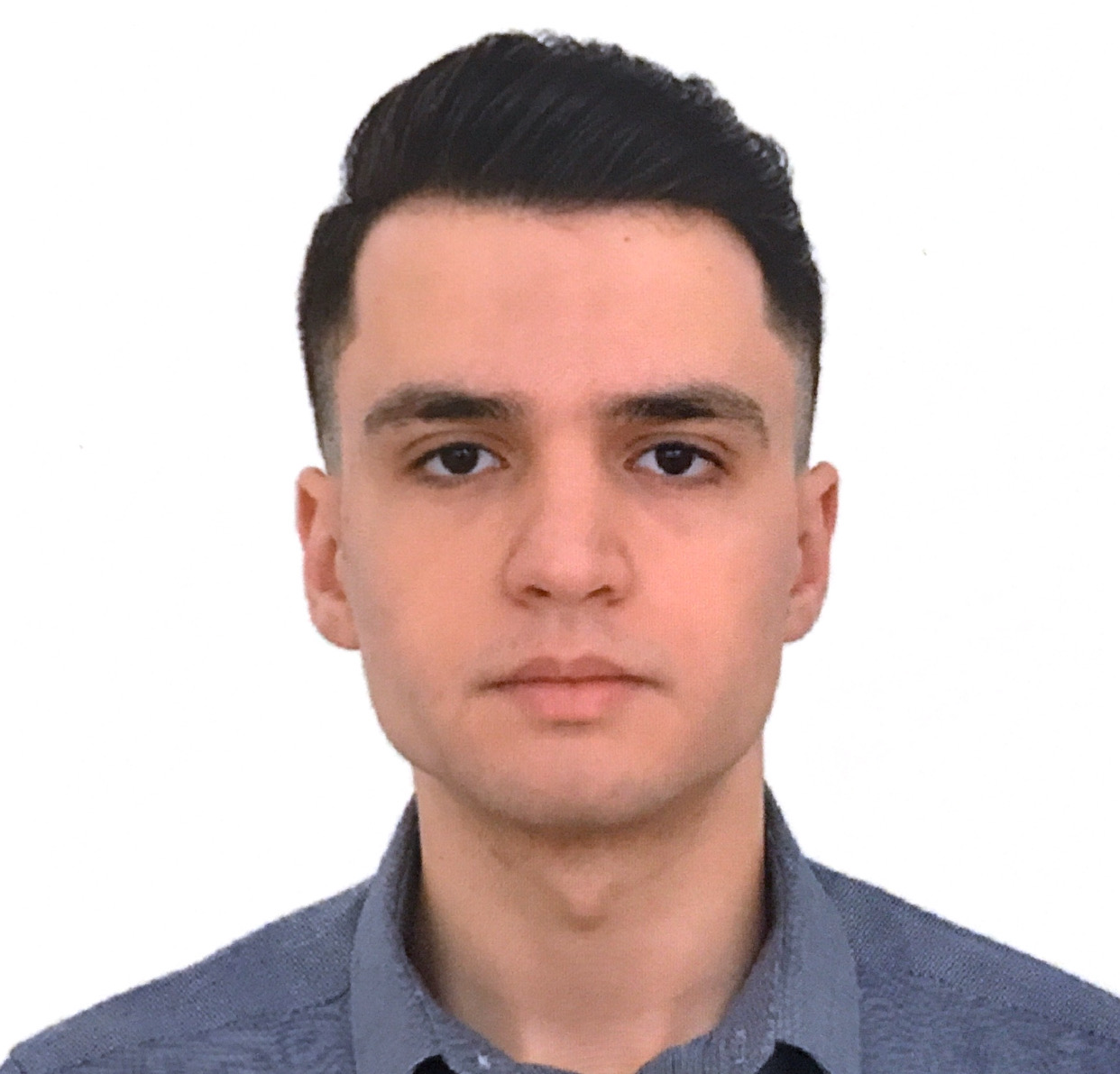}}]{Khaled Nakhleh} (Student Member, IEEE) received the M.S. degree in electrical engineering from Texas A\&M University, College Station, TX, where he is currently a Ph.D. electrical engineering student.
His research interests include reinforcement learning, distributed control, and multi-agent systems.
\end{IEEEbiography}

\vspace{11pt}

\begin{IEEEbiography}[{\includegraphics[width=1in,height=1.25in,clip,keepaspectratio]{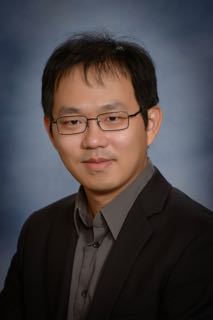}}]{I-Hong Hou}
(Senior Member, IEEE) is an Associate Professor in the ECE Department of the Texas A\&M University. He received his Ph.D. from the Computer Science Department of the University of Illinois at Urbana-Champaign. His research interests include wireless networks, edge/cloud computing, and reinforcement learning. His work has received the Best Paper Award from ACM MobiHoc 2017 and ACM MobiHoc 2020, and Best Student Paper Award from WiOpt 2017. He has also received the C.W. Gear Outstanding Graduate Student Award from the University of Illinois at Urbana-Champaign, and the Silver Prize in the Asian Pacific Mathematics Olympiad.
\end{IEEEbiography}

\vspace{11pt}

\begin{IEEEbiography}[{\includegraphics[width=1in,height=1.25in,clip,keepaspectratio]{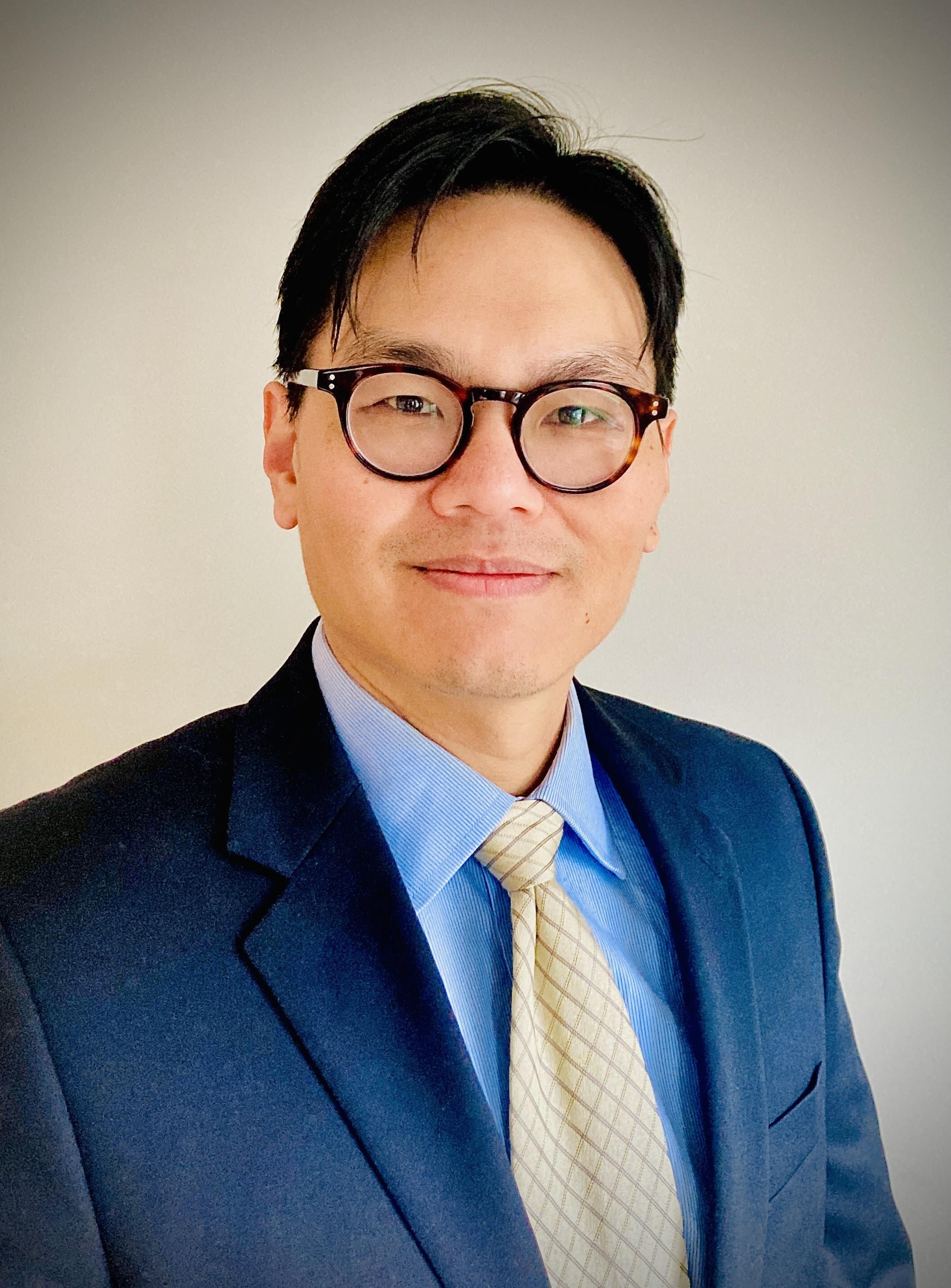}}]{Clement Kam}
(Senior Member, IEEE) received the BS degree in electrical and computer engineering from Cornell University, Ithaca, NY, in 2001, and the M.S. and Ph.D. degrees in electrical engineering from the University of California, San Diego, in 2006 and 2010, respectively. He is currently the Head of the Wireless Network Theory Section at the U.S. Naval Research Laboratory, Washington, DC. His research interests include ad hoc networks, cross-layer design, Age of Information, Semantics of Information, and reinforcement learning.
\end{IEEEbiography}

\vspace{11pt}

\begin{IEEEbiography}[{\includegraphics[width=1in,height=1.25in,clip,keepaspectratio]{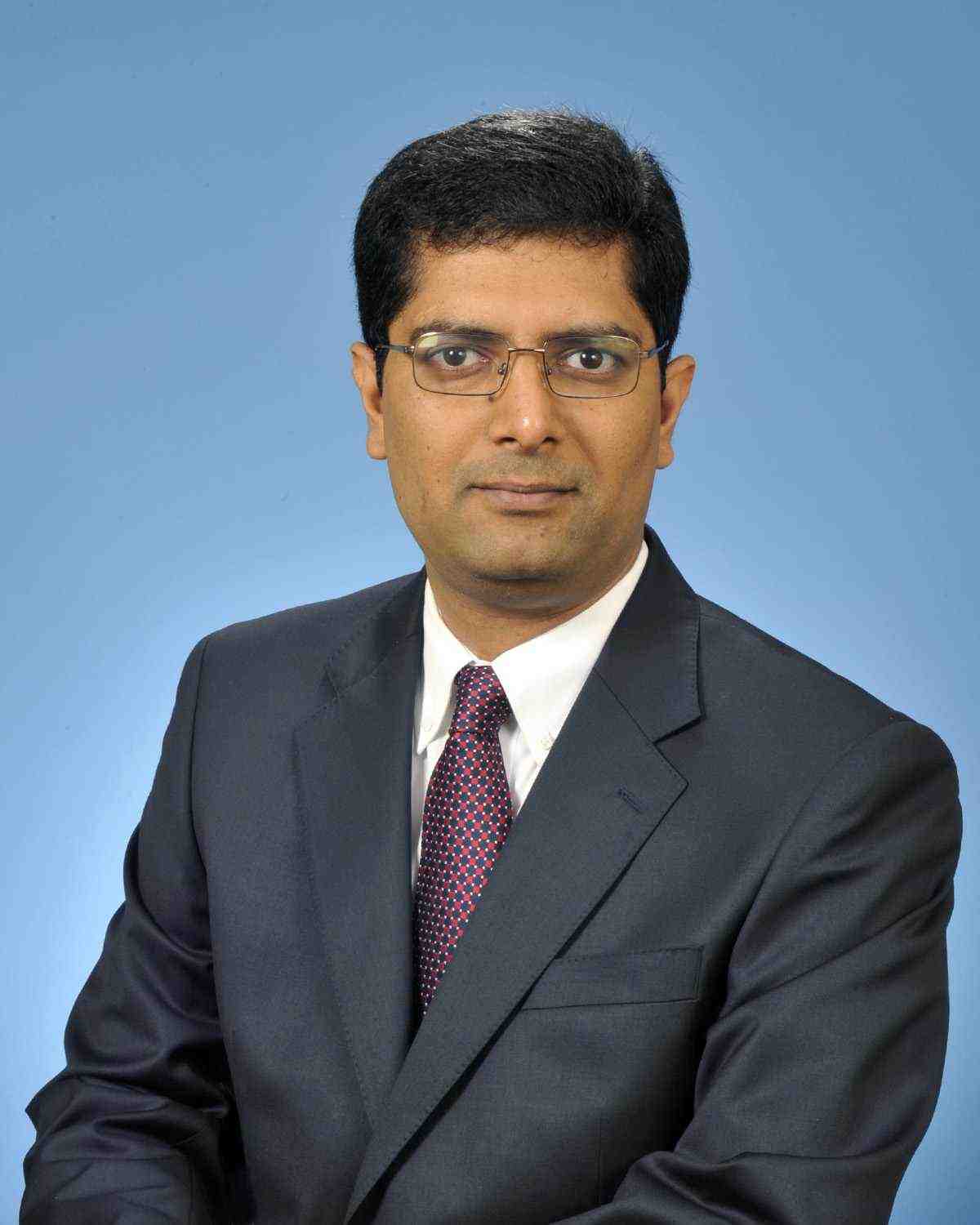}}]
{Sastry Kompella} (Fellow, IEEE) earned his Ph.D. in computer engineering from the Virginia Polytechnic Institute and State University in 2006. Currently, he serves as the CEO of Nexcepta, Inc., an advanced R\&D company providing cutting-edge technical solutions and mission-critical capabilities to the Department of Defense (DoD). Prior to this role, he held the position of Section Head for the Wireless Network Research Section within the Information Technology Division at the U.S. Naval Research Laboratory in Washington, DC, USA. He has authored over 200 (six award winning) publications, 14 NRL Reports, one patent, one book, and 5 book chapters. He currently serves as an associate editor for IEEE/ACM Transactions on Networking. His research interests encompass various aspects of wireless networks, including cognitive radio, dynamic spectrum access, and age of information.
\end{IEEEbiography}

\end{document}